\documentclass{article} % For LaTeX2e
\usepackage{iclr2025_conference,times}
\usepackage{tabularx}

\usepackage{amsmath,amsfonts,bm}

\def\eqref#1{equation~\ref{#1}}
\def\1{\bm{1}}

\def\eps{{\epsilon}}

\DeclareMathAlphabet{\mathsfit}{\encodingdefault}{\sfdefault}{m}{sl}
\SetMathAlphabet{\mathsfit}{bold}{\encodingdefault}{\sfdefault}{bx}{n}

\usepackage{tabularx}
\usepackage{hyperref}
\usepackage{url}
\usepackage{wrapfig}
\usepackage{booktabs}
\usepackage{subcaption}
\usepackage{multirow}
\usepackage{tikz}
\usepackage{marvosym}
\usetikzlibrary{arrows.meta, positioning}
\usetikzlibrary{arrows.meta, shapes, decorations.pathreplacing}
\usepackage{xcolor}
\usepackage{colortbl}
\definecolor{lightblue}{RGB}{173,216,230}
\definecolor{midblue}{RGB}{0,0,139}
\definecolor{deepblue}{RGB}{25,25,112}
\definecolor{myRed}{rgb}{0.9,0.1,0.1}
\definecolor{myBlue}{rgb}{0.4,0.7,0.9}
\definecolor{myPurple}{rgb}{0.7,0.4,0.9}
\definecolor{myOrange}{rgb}{0.9,0.7,0.4}
\definecolor{myGreen}{rgb}{0.4,0.9,0.4}
\definecolor{myYellow}{rgb}{1,0.8,0}
\usepackage{algorithm}
\usepackage{algpseudocode}
\usepackage{appendix}
\usepackage{titletoc}
\usepackage{changes}
\usepackage{amssymb}
\usepackage{colortbl}
\usepackage{forloop}
\usepackage{xcolor}
\usepackage{tcolorbox} % For fancy boxes
\usepackage{textcomp}
\iclrfinalcopy

\newcommand\DoToC{%

  \startcontents
  \printcontents{}{1}{\hrule\textbf{\begin{center}
       Appendix Table of Contents
  \end{center}}\vskip3pt\hrule\vskip5pt}
  \vskip3pt\hrule\vskip5pt
}

\usepackage[utf8]{inputenc} % allow utf-8 input
\usepackage[T1]{fontenc}    % use 8-bit T1 fonts
\usepackage{hyperref}       % hyperlinks
\usepackage{url}            % simple URL typesetting
\usepackage{booktabs}       % professional-quality tables
\usepackage{amsfonts}       % blackboard math symbols
\usepackage{nicefrac}       % compact symbols for 1/2, etc.
\usepackage{microtype}      % microtypography
\usepackage{xcolor}         % colors
\usepackage{amsmath}
\usepackage{graphicx} 
\usepackage{amsmath, amsthm, amssymb}

\newtheorem{theorem}{Theorem}
\newtheorem*{theorem*}{Theorem}

\newtheorem{corollary}{Corollary}

\newtheorem{claim}{Claim}

\newtheorem{remark}{Remark}

\title{Enhancing the Scalability and Applicability of Kohn-Sham Hamiltonians for  Molecular Systems}

\author{%
  \textbf{Yunyang Li}$^{\spadesuit}$\thanks{Equal contribution.} \quad
  \textbf{Zaishuo Xia}$^{\heartsuit}$\footnotemark[1] \quad
  \textbf{Lin Huang}$^{\clubsuit}$\textsuperscript{\Letter}\footnotemark[1] \quad
  Xinran Wei$^{\clubsuit}$ \quad 
  \textbf{Han Yang}$^{\clubsuit}$ \quad
  \textbf{Sam Harshe}$^{\spadesuit}$ \quad \\
  \textbf{Zun Wang}$^{\clubsuit}$ \quad
  \textbf{Chang Liu}$^{\clubsuit}$ \quad 
  \textbf{Jia Zhang}$^{\clubsuit}$\textsuperscript{\Letter} \quad
  \textbf{Bin Shao}$^{\clubsuit}$ \quad
  \textbf{Mark B. Gerstein}$^{\spadesuit}$\textsuperscript{\Letter} \\
  \textsuperscript{\Letter}\texttt{\{huang.lin, jia.zhang\}@microsoft.com} \quad \texttt{mark@gersteinlab.org} \\
  $^{\spadesuit}$Yale University \quad
  $^{\heartsuit}$UC Davis \quad
  $^{\clubsuit}$MSR AI4Science
}

\begin{document}

\maketitle
%\icmlauthor{}{sch}
%\icmlauthor{}{sch}
% Yunyang Li, Zaishuo Xia, Lin Huang, Xinran Wei, Han Yang, Sam Harshe, Zun Wang, Chang Liu, Jia Zhang, Bin Shao, Mark B. Gerstein,

\begin{abstract}

Density Functional Theory (DFT) is a pivotal method within quantum chemistry and materials science, with its core involving the construction and solution of the Kohn-Sham Hamiltonian. Despite its importance, the application of DFT is frequently limited by the substantial computational resources required to construct the Kohn-Sham Hamiltonian. In response to these limitations, current research has employed deep-learning models to efficiently predict molecular and solid Hamiltonians, with roto-translational symmetries encoded in their neural networks. However, the scalability of prior models may be problematic when applied to large molecules, resulting in non-physical predictions of ground-state properties. In this study, we generate a substantially larger training set (PubChemQH) than used previously and use it to create a scalable model for DFT calculations with physical accuracy. For our model, we introduce a loss function derived from physical principles, which we call \underline{W}avefunction \underline{A}lignment \underline{Loss} (WALoss). WALoss involves performing a basis change on the predicted Hamiltonian to align it with the observed one; thus, the resulting differences can serve as a surrogate for orbital energy differences, allowing models to make better predictions for molecular orbitals and total energies than previously possible. WALoss also substantially accelerates self-consistent-field (SCF) DFT calculations. Here, we show it achieves a reduction in total energy prediction error by a factor of 1347 and an SCF calculation speed-up by a factor of 18\%. These substantial improvements set new benchmarks for achieving accurate and applicable predictions in larger molecular systems.

\end{abstract}

\section{Introduction}

% short intro to DFT
Density functional theory (DFT)~\citep{kohn1965self,hohenberg1964inhomogeneous,martin2020electronic} has been widely used in physics~\citep{argaman2000density,RevModPhys.87.897, van2014density}, chemistry~\citep{levine2009quantum, van2014density}, and materials science~\citep{march1999electron, neugebauer2013density} to study the electronic properties of  molecules and solids. This methodology is particularly valued for its balanced blend of computational efficiency and accuracy, rendering it a versatile choice for investigating electronic structure~\citep{kohn1996density,parr1995density}, spectroscopy~\citep{neese2009prediction,orio2009density}, lattice dynamics~\citep{dal1997density,wang2021dfttk}, transport properties~\citep{bhamu2018revealing}, and more. The most critical step in applying DFT to a molecule is constructing the Kohn-Sham Hamiltonian, which consists of the kinetic operator, the external potential, the Coulomb potential (also known as the Hartree potential), and the exchange-correlation potential~\citep{kohn1965self,hohenberg1964inhomogeneous,martin2020electronic}.
% Density Functional Theory (DFT) stands as a cornerstone in molecular modeling, essential for probing diverse intermolecular and intramolecular phenomena \citep{kohn1965self,hohenberg1964inhomogeneous,martin2020electronic}. Its allure stems from an advantageous balance between computational expediency and precision, making it suitable for studying a wide array of molecular systems. At the core of DFT are the Kohn-Sham equations, which establish the Kohn-Sham Hamiltonian \citep{kohn1965self} incorporating elements such as kinetic energy, external potential, electron Coulomb potential, and exchange-correlation potential, crucial for delineating the electronic configuration of molecules.

% more about DFT hamiltonian and the application
% The Hamiltonians of molecules encode crucial information of a molecule, from which the ground-state energy and electron density can be derived. In addition, Hamiltonians also facilitate the prediction of electronic structure, including the highest occupied molecular orbital (HOMO), the lowest unoccupied molecular orbital (LUMO) and the HOMO-LUMO gap.

The Hamiltonian matrix
% , which becomes the Fock matrix upon convergence of the self-consistent field (SCF) iterations in a DFT computation, 
contains crucial information about molecular systems and their quantum states~\citep{kohn1965self,hohenberg1964inhomogeneous,yu2023efficient, zhang2024selfconsistency}. This matrix facilitates the extraction of various properties, including the highest occupied molecular orbital (HOMO) and lowest unoccupied molecular orbital (LUMO) energies, the HOMO-LUMO gap, total energy, and spectral characteristics ~\citep{eisberg1985quantum,zhang2024selfconsistency}. These properties are vital for analyzing conformational energies~\citep{st1995calculation}, reaction pathways~\citep{farberow2014density}, and vibrational frequencies~\citep{watson2002density}. However, the efficiency of DFT is constrained by the self-consistent field (SCF) iterations required to solve the Kohn-Sham equations and achieve consistent charge density. These iterations scale as $O(N^3T) \sim O(N^4T)$ \citep{tirado2008performance,yu2024qh9}, where $N$ represents the number of electrons and $T$ denotes the number of SCF cycles, making them particularly resource-intensive for large systems.

% the past developments
Consequently, this computational intensity underscores the critical need to develop more efficient methodologies to determining the Hamiltonian without relying on SCF iterations. This challenge has catalyzed interest in leveraging deep learning to predict Hamiltonians directly from atomic configurations while adhering to the inherent symmetries of molecular systems~\citep{unke2021se3,kochkov2021learning,yu2023efficient,li2022marriage,schutt2019machine,gastegger2020emergence,hermann2020deep,yin2024harmonizing,zhang2024selfconsistency,luo2025efficient}. Notably, the Hamiltonian matrix is subject to unitary transformations under molecular rotations due to its spherical harmonics component. To tackle these transformations, researchers have developed SE(3)-equivariant neural networks such as PhisNet \citep{unke2021seequivariant} and QHNet \citep{yu2023efficient}. These networks employ high-order spherical tensors and the Clebsch-Gordon tensor product to construct the predicted Hamiltonians.

% what we propose
However, current implementations of SE(3)-equivariant neural networks face considerable scalability challenges when applied to large molecular structures. Notably, prevailing techniques for predicting the Hamiltonian matrix predominantly utilize mean absolute error (MAE) or mean square error (MSE) as the loss function \citep{unke2021seequivariant,yu2023efficient}. We contend that these \textit{elementwise losses alone are insufficient} for accurately learning Hamiltonians in large systems. To underscore this point, we introduce the PubChemQH dataset, which comprises molecular Hamiltonians for structures with atom counts ranging from 40 to 100. In contrast, the previously curated dataset \citep{yu2024qh9} is limited to molecules with no more than 31 atoms. Figure \ref{fig:SAD} demonstrates a phenomenon enabled by the introduction of the PubChemQH dataset, which we term \textit{\underline{S}caling-Induced MAE-\underline{A}pplicability \underline{D}ivergence} (SAD).  The applicability of the Hamiltonian is assessed by evaluating the system energy error derived from \texttt{pyscf} \citep{sun2020recent,sun2015libcint,sun2018pyscf}. While small molecules yield system energies close to the ground truth with relative Hamiltonian MAEs up to 200\%, larger molecules exhibit profound inaccuracies. Energy errors escalate to as much as 1,000,000 kcal/mol at only 0.01\% relative MAE, rendering the Hamiltonians inapplicable. This striking disparity emphasizes the inadequacy of MAE for large systems and highlights the pressing need for new methodologies to enhance the scalability\footnote{We define the scalability as the model's accuracy under large molecules.} and applicability of predicted Hamiltonians.

To address these limitations, this work aims to enhance Hamiltonian learning for large systems by utilizing wavefunctions and their corresponding energies as surrogates to improve Hamiltonian applicability. The wavefunction is crucial for Hamiltonian prediction as it encapsulates the quantum state of a system, enabling the validation of Hamiltonians based on their ability to accurately reflect the system's energy and dynamics. However, learning the electron wavefunction is non-trivial; 
inaccuracies in the machine learning-based Hamiltonian may lead to significant error in both electronic structure and wavefunctions.
% minor variations in the Hamiltonian can lead to significant changes in wavefunctions~\citep{lennard1930perturbation,fernandez2000introduction,lowdin1951note}.
This challenge motivated us to introduce the \textit{\underline{W}avefunction \underline{A}lignment \underline{Loss}} function (WALoss), which aligns the eigenspaces of predicted and ground-truth Hamiltonians without explicit backpropagation through eigensolvers. Additionally, to improve the scalability of Hamiltonian learning, we present \textit{WANet}, a modernized architecture for Hamiltonian prediction that leverages eSCN~\citep{passaro2023reducing} convolution and a sparse mixture of pair experts. Our pipeline and main contributions are summarized in Figure \ref{fig:pipeline}.

% while these objective perform adequately for relatively small systems, they become problematic as the system size increases. We define \textit{error tolerance}, where given  The error tolerance becomes exceedingly stringent, making the learning target for the neural network overly challenging and hindering the model's ability to generalize to larger molecular structures. This limitation highlights the need for novel approaches that can effectively scale SE(3)-equivariant neural networks to handle more complex systems while maintaining computational efficiency and predictive accuracy.

\begin{figure}[t]
\vspace{-14mm}
\begin{center}
\includegraphics[width=0.95\columnwidth]{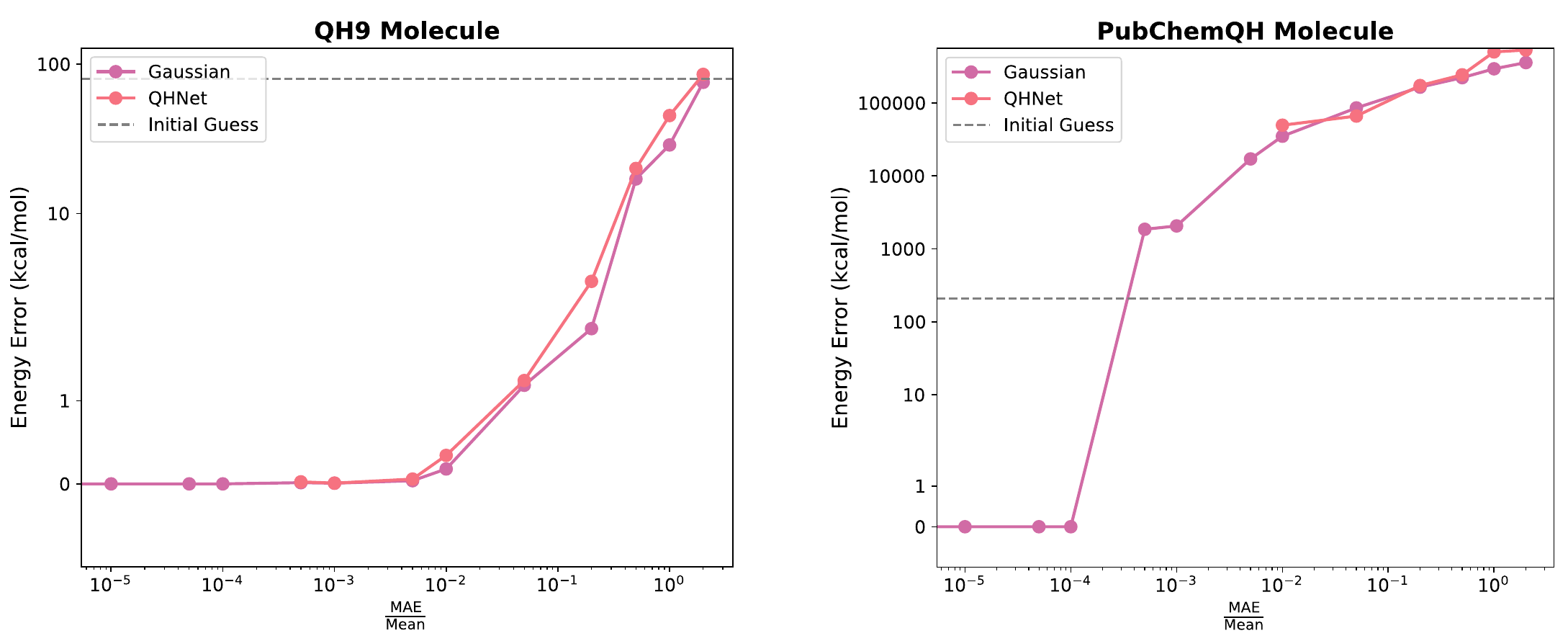}

\caption{
\small{Visualization of the SAD phenomenon.
The y-axis (in symmetrical log scale~\citep{webber2012bi}) represents the system energy error derived from the perturbed Hamiltonian, while the x-axis shows the relative MAE, defined as $\frac{\mathrm{MAE}(\hat{\mathbf{H}}, \mathbf{H}^\star)}{\mathrm{Mean}(\mathbf{H}^\star)}$, where $\mathbf{H}^\star$ represents the ground-truth Hamiltonian and $\hat{\mathbf{H}}$ denotes the predicted or perturbed Hamiltonian. The relative MAE is induced by model learning errors or Gaussian perturbation. The Gaussian perturbation ensures that the perturbed matrix remains Hermitian.  The initial guess line is derived from \texttt{minao}~\citep{almlof1982principles,van2006starting}. Current state-of-the-art models, such as QHNet, achieve a relative MAE of up to $10^{-2}$ on PubChemQH molecules. For small molecules, a $10^{-2}$ relative MAE is sufficient for accurate system energy predictions (left 
panel), but this accuracy does not extend to large systems (right panel). }}
\label{fig:SAD}
\end{center}
\vspace{-10mm}
\end{figure} 

\section{Background}

In the framework of Density Functional Theory (DFT), a molecular system is defined by its nuclear configuration \( \mathcal{M} := \left\{ \mathbf{Z}, \mathbf{R} \right\}\), where \( \mathbf{Z} \) represents the atomic numbers of the nuclei and \( \mathbf{R} \) their positions within the molecule. DFT focuses on determining the ground state of a system consisting of \( N \) electrons by minimizing the total electronic energy with respect to the electron density \( \rho(\mathbf{r}) \). Here, \( \rho(\mathbf{r}) \) is a functional of the set of \( N \) one-electron orbitals \( \{\psi_i(\mathbf{r})\}_{i=1}^N \), where \( \mathbf{r} \in \mathbb{R}^3 \) specifies the spatial coordinates of an electron.

For computational efficiency, these orbitals are represented using a basis set \( \{\phi_\alpha(\mathbf{r})\}_{\alpha=1}^B \) depending on the molecular geometry, \(B\) denotes the number of basis and varies for different basis sets. The expansion coefficients of the orbitals are organized in a matrix \( \mathbf{C} \in \mathbb{R}^{B \times N} \), allowing each orbital \( \psi_i(\mathbf{r}) \) to be expressed as:
\(
\psi_i(\mathbf{r}) = \sum_{\alpha=1}^B C_{\alpha i}\phi_\alpha(\mathbf{r}).
\)
DFT seeks the minimal electronic energy by solving for the optimal coefficient matrix \( \mathbf{C} \) via the Kohn-Sham equations, represented as
\(
\mathbf{H}(\mathbf{C})\mathbf{C} = \mathbf{S}\mathbf{C}\boldsymbol{\epsilon},
\)
where \( \mathbf{H}(\mathbf{C})\in \mathbb{R}^{B \times B} \) denotes the Hamiltonian matrix. Notice that $\mathbf{H}$ is a function of $\mathbf{C}$, and can be computed using \textit{density fitting} with a complexity of $O(B^3)$. \( \mathbf{S} \in \mathbb{R}^{B\times B} \) is the overlap matrix with elements \( S_{\alpha\beta} := \int  \phi^\dagger_\alpha(\mathbf{r}) \phi_\beta(\mathbf{r}) \mathrm{d}\mathbf{r} \), where $\phi^\dagger$ denotes the complex conjugate of \(\phi\). \( \boldsymbol{\epsilon} \) represents a diagonal matrix containing the orbital energies.
This forms a generalized eigenvalue problem, where \( \mathbf{C} \) comprises the eigenvectors and the diagonal elements of \( \boldsymbol{\epsilon} \) are the eigenvalues. However, the solution to this problem is complicated by the interdependence between \( \mathbf{H}(\mathbf{C}) \) and \( \mathbf{C} \).
To resolve this, traditional DFT employs the self-consistent field (SCF) method, an iterative process that refines the coefficients \( \mathbf{C}^{(k)} \) through successive approximations of the Hamiltonian matrix \( \mathbf{H}^{(k)} \). Each iteration commences by  computing \(\mathbf{H}^{(k)}\) leveraging \(\mathbf{C}^{(k-1)}\) and solves a new eigenvalue problem:
\[
\mathbf{H}^{(k)}\left(\mathbf{C}^{(k-1)}\right)\mathbf{C}^{(k)} = \mathbf{S}\mathbf{C}^{(k)}\boldsymbol{\epsilon}^{(k)},
\]
converging to a stable Hamiltonian \( \mathbf{H}^\star \) and its corresponding eigenvectors \( \mathbf{C}^\star \),  which define the electron density and other molecular properties derived from the Kohn-Sham equations.

\textbf{Problem Formulation} \quad The objective of Hamiltonian prediction is to obviate the need for self-consistent field (SCF) iteration by directly estimating the target Hamiltonian \( \mathbf{H}^\star \) from a given molecular structure \( \mathcal{M} \). To achieve this, one could parameterize a machine learning model \( \hat{\mathbf{H}}_\theta (\mathcal{M}) \).  The learning process is guided by an optimization process defined as:

\begin{equation}
 \theta^\star = \arg\min_{\theta} \frac{1}{|\mathcal{D}|} \sum_{(\mathcal{M},\mathbf{H}^\star_\mathcal{M}) \in \mathcal{D}}  \mathrm{dist}\left(\hat{\mathbf{H}}_\theta (\mathcal{M}) , \mathbf{H}^\star_\mathcal{M} \right),
\end{equation}

where \( |\mathcal{D}| \) denotes the cardinality of dataset \( \mathcal{D} \), and \(\mathrm{dist}(\cdot, \cdot)\) is a predefined distance metric.

\textbf{SE(3) Equivariant Networks} \quad
SE(3)-equivariant neural networks incorporate strong prior knowledge through equivariance. These networks utilize equivariant irreducible representation (irreps) features built from vector spaces of irreducible representations to achieve 3D rotation equivariance. The vector spaces are \((2\ell+1)\)-dimensional, where degree \(\ell \in \mathbb{N}\) represents the angular frequency of the vectors. Higher \(\ell\) values are critical for tasks sensitive to angular information, such as predicting the Hamiltonian matrix. Vectors of degree \(\ell\), referred to as type-\(\ell\) vectors, are rotated using Wigner-D matrices \(D^{(\ell)} \in \mathbb{R}^{(2\ell+1) \times (2\ell+1)}\) when rotating coordinate systems. Eigenfunctions  of rotation in \(\mathbb{R}^3\) can be projected into type-\(\ell\) vectors using the real spherical harmonics \(Y^{(\ell)} : \mathbb{S}^2 \to \mathbb{R}^{2\ell+1}\), where \(\mathbb{S}^2 = \{\hat{\mathbf{r}} \in \mathbb{R}^3 : \Vert\hat{\mathbf{r}}\Vert = 1\}\) 
% \chang{Unify whether using $|\cdot|$ or $\Vert \cdot \Vert$ for vector norm.} 
is the unit sphere. Equivariant GNNs update irreps features through message passing of transformed irreps features between nodes, using tensor products:
\(
(u^{\ell_1} \otimes v^{\ell_2})^{\ell_3}_{m_3} = \sum_{m_1=-\ell_1}^{\ell_1} \sum_{m_2=-\ell_2}^{\ell_2} C^{(\ell_3,m_3)}_{(\ell_1,m_1),(\ell_2,m_2)} u^{\ell_1}_{m_1} v^{\ell_2}_{m_2},
\) where \(C\) denotes the Clebsch-Gordan coefficient, \(\ell_3\) satisfies \(|\ell_1 - \ell_2| \leq \ell_3 \leq \ell_1 + \ell_2\). Note that \(m\) denotes the \(m\)-th element in the irreducible representation with \(-\ell \leq m \leq \ell\) and \(m \in \mathbb{N}\). In equivariant graph neural networks, the message function \(m_{ts}\) from source node \(s\) to target node \(t\) is calculated using $\mathrm{SO}(3)$ convolution. The \(\ell_0\)-th degree of \(m_{ts}\) can be expressed as \(
m^{(\ell_o)}_{ts} = \sum_{\ell_i, \ell_f} W_{\ell_i, \ell_f, \ell_o} \left( x^{(\ell_i)}_s \otimes Y^{(\ell_f)}(\hat{\mathbf{r}}_{ts}) \right)^{\ell_o},
\)
where \(W_{\ell_i, \ell_f, \ell_o} \) are the weight vectors, $x_s$ represents the irreducible representation of the source node $s$ and $\hat{\mathbf{r}}_{ts} := \frac{\mathbf{r}_{ts}}{\|\mathbf{r}_{ts}\|_2}$.

\begin{figure}[t]
\vspace{-15mm}
\begin{center}
\includegraphics[width=\columnwidth]{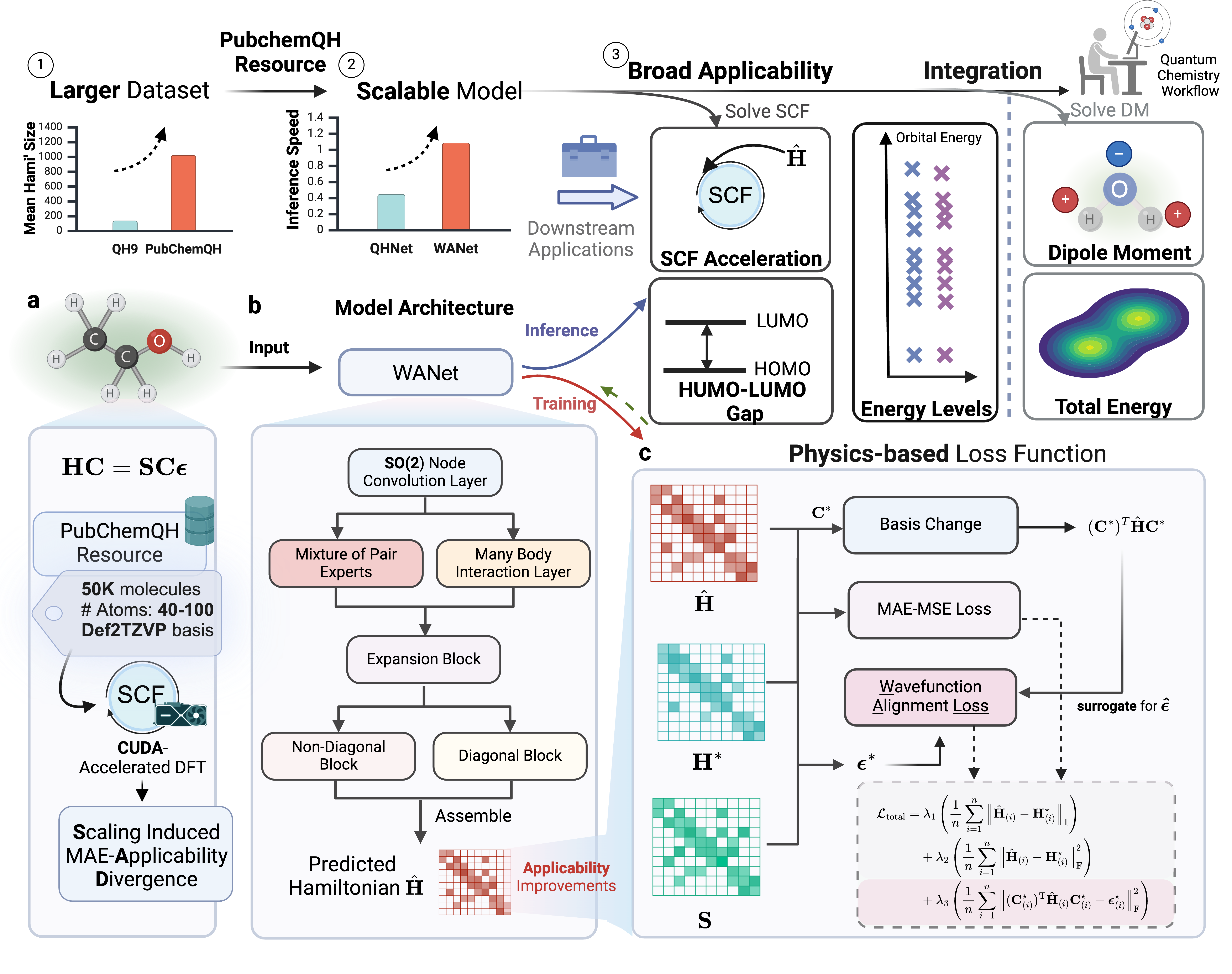}

\caption{ \small{
(a) We introduce PubChemQH, a new resource for Hamiltonian learning that facilitates the exploration of the scalability challenge known as SAD.
(b) We present WANet, a modern architecture designed for accurate Hamiltonian prediction. WANet incorporates $\mathrm{SO}(2)$ convolution, a mixture of pair experts, and a many-body interaction layer. The mixture of pair experts constructs the non-diagonal block, while the many-body interaction layer constructs the diagonal block.
(c) Our loss module, WALoss, performs a basis transformation of the predicted Hamiltonian using the ground-truth Hamiltonian and the overlap matrix. This enhancement aims to improve the applicability of the predicted Hamiltonian in real-world scenarios. Our final loss function combines MAE-MSE loss with WALoss.}
}
\label{fig:pipeline}
\end{center}
\vspace{-10mm}
\end{figure} 

\section{Wavefunction-Alignment Loss}

The applicability of the predicted Hamiltonian \(\hat{\mathbf{H}}\) depends significantly on the accuracy of its eigenvalues (orbital energies) and eigenvectors (basis coefficients). To assess the alignment between \(\hat{\mathbf{H}}\) and the actual Hamiltonian \(\mathbf{H}^\star\) with respect to their eigenspaces, we present the following theorem:

\begin{theorem}
    Let \(\mathbf{H}\) and \(\hat{\mathbf{H}}\) represent Hamiltonian matrices, and \(\mathbf{S}\) the overlap matrix. Define the perturbation matrix as \(\Delta \mathbf{H} := \hat{\mathbf{H}} - \mathbf{H}\). Let \(\lambda_i(\mathbf{H, S})\) and \(\lambda_i(\hat{\mathbf{H}}, \mathbf{S})\) be the \(i^{th}\) generalized eigenvalues of \(\mathbf{H}\) and \(\hat{\mathbf{H}}\), respectively. Assume a spectral gap \(\delta\) separates the generalized eigenvalues of \(\mathbf{H}, \hat{\mathbf{H}}\). $\kappa(\cdot)$ denotes the conditional number of a given matrix, $\|\cdot\|_2$ represents the spectral norm. \(
\|\Delta \mathbf{H}\|_{1,1} = \sum_{i,j} |\Delta \mathbf{H}_{ij}|
\). Then, the difference in eigenvalues and the angle \(\theta\) between the eigenspace of \(\mathbf{H}\) and \(\hat{\mathbf{H}}\) satisfy:

\[ | \lambda_i(\hat{\mathbf{H}}, \mathbf{S}) - \lambda_i(\mathbf{H, S}) |\leq \frac{\kappa(\mathbf{S})}{\|\mathbf{S}\|_2} \cdot \|\Delta \mathbf{H}\|_{1, 1}, \quad \sin \theta \leq \frac{\kappa(\mathbf{S})}{\|\mathbf{S}\|_2} \cdot \frac{\|\Delta \mathbf{H}\|_{1,1}}{\delta}. \]
\end{theorem}

\begin{proof}
    The proof is delegated to Appendix~\ref{app:weyl-davis}.
\end{proof}

\begin{corollary}[Perturbation Sensitivity Scaling]
Assume $\lambda_{\text{min}}(S)$ of the overlap matrix $S$ scales as:
$
\lambda_{\text{min}}(S) = c + \frac{A}{1 + \left( \dfrac{B}{N_0} \right)^\alpha},
$
where $c > 0$, $A > 0$, and $\alpha > 0$ describe the saturation behavior. Then, for a perturbation matrix $\Delta H$, under the first-order perturbation approximation, the eigenvalue perturbation is bounded by:

\[
\left| \lambda_i(\hat{\mathbf{H}}, \mathbf{S}) - \lambda_i(\mathbf{H}, \mathbf{S}) \right| \leq O \left( \sigma B^{1/2} + B |\mu| \right),
\]
\end{corollary}
\begin{proof}
    The proof is delegated to Appendix~\ref{app:Eigen-sensitivity}.
\end{proof}

\textbf{Remark} \quad The theorem highlights that the difference between the predicted and actual Hamiltonian matrices, when only considering the element-wise norm, can lead to unbounded differences in eigenvalues/eigenvectors due to a significant \(\frac{\kappa(\mathbf{S})}{\|\mathbf{S}\|}\) ratio.  The corollary further elucidates the sensitivity of eigenvalue perturbations to the size of the basis \( B \).This phenomenon underscores the catastrophic scaling associated with increasing $B$, which is a manifestation of the aforementioned SAD phenomenon. We provide a thorough discussion on $\alpha$ in Appendix \ref{app:scale-eigen}. To validate our theoretical analysis, we empirically evaluated the distribution of  \(\frac{\kappa(\mathbf{S})}{\|\mathbf{S}\|_2}\) across both the QH9 and PubChemQH datasets (Figure~\ref{fig:dist}). The results demonstrate that molecules in PubChemQH exhibit substantially higher ratios compared to QH9, indicating increased perturbation sensitivity in larger molecular systems. This provides strong empirical evidence for the SAD phenomenon and aligns with our theoretical predictions. Consequently, these findings suggest that when designing an effective supervisory signal for learning or optimization tasks involving Hamiltonian matrices, it is crucial to take into account the interaction of the overlap matrix and the corresponding Hamiltonian to mitigate the potential instability caused by perturbations.

In light of this, we propose a new loss function: the Wavefunction Alignment Loss. It is designed to preserve the integrity of the eigenstructure related to molecular orbitals. Let \(\hat{\boldsymbol{\epsilon}}\) and \(\boldsymbol{\epsilon}^\star\) represent the eigenvalues (orbital energies) of \(\hat{\mathbf{H}}\) and \(\mathbf{H}^\star\), respectively, and \(\hat{\mathbf{C}}\) and \(\mathbf{C}^\star\) denote the corresponding eigenvectors (basis coefficients). We define a primary form of the loss function by directly applying the Frobenius norm to the eigenvalues \(
\mathcal{L}_{\mathrm{align}} = \frac{1}{n} \sum_{i=1}^n \| \hat{\boldsymbol{\epsilon}}_{_{(i)}} - \boldsymbol{\epsilon}_{_{(i)}}^\star \|^2_\mathrm{F},
\)
where \(\hat{\boldsymbol{\epsilon}}_{_{(i)}}\) and \(\boldsymbol{\epsilon}_{_{(i)}}^\star\) are derived from solving the generalized eigenvalue problems: \(\hat{\mathbf{H}} \hat{\mathbf{C}} = \mathbf{S} \hat{\mathbf{C}} \hat{\boldsymbol{\epsilon}}\) and \(\mathbf{H}^\star \mathbf{C}^\star = \mathbf{S} \mathbf{C}^\star \boldsymbol{\epsilon}^\star\), respectively.

However, this formulation has notable limitations: (1) Generalized eigenvalue problems are susceptible to numerical instabilities due to ill-conditioned matrices, leading to erroneous gradients during backpropagation through iterative eigensolvers, complicating optimization. (2) The loss function assigns uniform weights to all orbital energies, which is not practical as some orbital energies hold more significance.
To address these issues, we begin by applying the following algorithm~\citep{ghojogh2019eigenvalue,golub2013matrix} to perform a simultaneous reduction of the matrix pair \((\mathbf{H}^\star, \mathbf{S})\):
\vspace{-3mm}

\begin{algorithm}

\caption{Simultaneous reduction of a matrix pair $(\mathbf{H}^\star, \mathbf{S})$}
\label{alg:diagonalization}
\begin{algorithmic}[1]
    \Require Groud-truth Hamiltonian matrix $\mathbf{H}^\star$ and overlap matrix $\mathbf{S}$
    \Ensure Diagonal matrix $\boldsymbol{\epsilon}^\star$ and matrix $\mathbf{C}^\star$ such that $(\mathbf{C^\star})^{\top} \mathbf{SC^\star} = \mathbf{I}$ and $(\mathbf{C}^\star)^{\top} \mathbf{H^\star C^\star} = \boldsymbol{\epsilon}^\star$
    \State Compute the Cholesky decomposition $\mathbf{S} = \mathbf{G}\mathbf{G}^{\top}$.
    \State Define $\mathbf{M}^\star = \mathbf{G}^{-1}\mathbf{H}^\star\mathbf{G}^{-\top}$.
    \State Apply the symmetric QR algorithm to find the Schur form $(\mathbf{Q}^\star)^{\top}\mathbf{M}^\star\mathbf{Q}^\star = \boldsymbol{\epsilon}^\star$.
    \State Compute $\mathbf{C}^\star = \mathbf{G}^{-\top}\mathbf{Q}^\star$.
\end{algorithmic}

\end{algorithm}
\vspace{-3mm}

When the overlap matrix \(\mathbf{S}\) is ill-conditioned, the eigenvalues \(\boldsymbol{\epsilon}^\star\) computed by Algorithm \ref{alg:diagonalization} can suffer from significant roundoff errors. To mitigate this, we modify the algorithm by replacing the Cholesky decomposition of \(\mathbf{S}\) with its eigen (Schur) decomposition \(\mathbf{V}^{\top} \mathbf{S} \mathbf{V} = \mathbf{\Sigma}\), where \(\mathbf{V}\) is the matrix of eigenvectors and \(\mathbf{\Sigma}\) is the diagonal matrix of eigenvalues. We then substitute \(\mathbf{G}\) with \(\mathbf{V\Sigma}^{-1/2}\). This modification effectively reorders the entries of \(\mathbf{M}^\star\), placing larger values towards the upper left-hand corner, thereby enhancing the precision in computing smaller eigenvalues \citep{golub2013matrix, wilkinson1988algebraic}. 

% Additionally, when formulating the loss function, we project the predicted Hamiltonian matrix \(\hat{\mathbf{H}}\) into an orthogonal basis using the transformation \(\hat{\mathbf{M}} = \mathbf{G}^{-1} \hat{\mathbf{H}} \mathbf{G}\).
\begin{claim}
    Under optimal convergence condition, the ground truth eigenvector $\mathbf{C}^\star$ should diagonalize the predicted transformed matrix $\hat{\mathbf{H}}$, thereby satisfying the relation $(\mathbf{C}^\star)^{\top} \hat{\mathbf{H}} \mathbf{C}^\star = \boldsymbol{\epsilon}^\star$.
\end{claim}

Accordingly, we propose a refined loss function:
\begin{equation}
\mathcal{L}_{\mathrm{WA}} = \frac{1}{n} \sum_{i=1}^n \| (\mathbf{C}_{(i)}^\star)^\top \hat{\mathbf{H}}_{(i)} \mathbf{C}_{(i)}^\star - \boldsymbol{\epsilon}^\star_{(i)} \|^2_{_\mathrm{F}},
\end{equation}
where the subscript \((\cdot)_{_{(i)}}\) denotes the \(i\)-th sample. This modified loss function explicitly penalizes deviations from the expected eigenstructure.  It is important to note that while \((\mathbf{C}^\star)^{\top} \hat{\mathbf{H}} \mathbf{C}^\star\) may not yield a diagonal matrix, the diagonal nature of \(\boldsymbol{\epsilon}^\star\) implicitly enforces the predicted eigenstructure through the non-diagonal elements of the loss function. Here, we provide another derivation of the WALoss by drawing a parallel to first-order perturbation theory. In perturbation theory, for a Hamiltonian \( \mathbf{\hat{H}} = \mathbf{H^{\star}} + \Delta \mathbf{H} \), where \(\Delta \mathbf{H}\) is the perturbation.  The first-order energy shift is \( \delta\boldsymbol{\epsilon} = \int \psi^{\star\dagger} \Delta{\mathbf{H}} \psi^{\star} d\mathbf{r} \), using the unperturbed eigenvectors \( \psi^{\star} \) of \( \mathbf{H}^{\star} \), with the assumption that these eigenvectors remain fixed while adjusting the eigenvalues. Similarly, WALoss employs the ground-truth eigenvectors \( \mathbf{C}^\star \) from \( \mathbf{H}^\star \mathbf{C}^\star = \mathbf{S} \mathbf{C}^\star \boldsymbol{\epsilon}^\star \) to evaluate the predicted Hamiltonian \( \mathbf{H}^\star + \Delta \mathbf{H} \). The loss could be rewritten as \( \mathcal{L}_{\mathrm{WA}} = \frac{1}{n} \sum_{i=1}^n \| (\mathbf{C}_{(i)}^\star)^\top (\mathbf{H}^\star + \Delta \mathbf{H})_{(i)} \mathbf{C}_{(i)}^\star - \boldsymbol{\epsilon}^\star_{(i)} \|^2_{_\mathrm{F}} \), assessing how well \( \mathbf{H}^\star + \Delta \mathbf{H} \) reproduces \( \boldsymbol{\epsilon}^\star \) in the \( \mathbf{C}^\star \) basis. Here, \( (\mathbf{C}^\star)^\top (\mathbf{H}^\star + \Delta \mathbf{H}) \mathbf{C}^\star = \boldsymbol{\epsilon}^\star + (\mathbf{C}^\star)^\top \Delta \mathbf{H} \mathbf{C}^\star \), where \( (\mathbf{C}^\star)^\top \Delta \mathbf{H} \mathbf{C}^\star \) mirrors the perturbation term \( \int \psi^{\star\dagger} \Delta \mathbf{H} \psi^\star d\mathbf{r} \).

It is widely known that energies associated with occupied molecular orbitals (and LUMO) make up the most energy of the molecule. Thus, to capture the most relevant part of the eigenspectrum, we calculate the loss function with increased weight on the \(k+1\) lowest eigenvalues. Here, \(k\) corresponds to the number of occupied orbitals \footnote{\(k=\lfloor \frac{N}{2} \rfloor\) for paired orbitals.}, and an additional eigenvalue accounts for the Lowest Unoccupied Molecular Orbital (LUMO). Let \(\mathcal{I}\) represent the set of indices corresponding to the \(k+1\) lowest eigenvalues, and  \(((\cdot)_{_{(i)}})_j\) indexes the \(j\)-th eigenvalue for the \(i\)-th sample. The loss function is defined as follows:

\begin{equation}
\mathcal{L}_{_{\mathrm{WA}}} = \frac{1}{n}\sum_{i=1}^{n} \left(\rho  \sum_{j \in \mathcal{I}} \| (\mathbf{C}^\star_{_{(i)}})_{j}^{\top} \hat{\mathbf{H}}_{_{(i)}} (\mathbf{C}^\star_{(i)})_{j} - (\boldsymbol{\epsilon}^\star_{_{(i)}})_{j} \|^{2}_{_\mathrm{F}} + \xi  \sum_{j \notin \mathcal{I}} \| (\mathbf{C}^\star_{_{(i)}})_{j}^{\top} \hat{\mathbf{H}}_{_{(i)}} (\mathbf{C}^\star_{_{(i)}})_{j} - (\boldsymbol{\epsilon}^\star_{_{(i)}})_{j} \|^{2}_{_{\mathrm{F}}}\right),
\end{equation}

where \(n\) denotes the number of samples, and \(\rho, \xi\) are hyperparameters where \(\rho \gg \xi \). This adaptation ensures that the loss function places greater emphasis on the eigenvectors and eigenvalues corresponding to the occupied orbitals and LUMO. This improves the model's focus on the critical part of the eigenspace, which is crucial for practical applications.

% The final step involves transforming the eigenvectors \( \mathbf{C}' \) back to the original basis using \( \mathbf{X} \), resulting in the coefficients \( \mathbf{C} \) in the original basis. This systematic approach not only facilitates the diagonalization of \( \mathbf{F} \) relative to \( \mathbf{S} \) but also ensures the stability and accuracy of the eigenvalues and eigenvectors computed in the process, addressing common numerical challenges in quantum chemical calculations. Our method provides a robust framework for efficiently solving one of the core equations in molecular orbital theory, demonstrating its applicability across various molecular systems.

\section{WANet}
Here, we present WANet, a modernized architecture for Hamiltonian prediction. First, unlike previous approaches, we propose a streamlined design for Hamiltonian prediction that consists of two essential components: the Node Convolution Layer and the Hamiltonian Head. The Node Convolution Layer operates on a localized radius graph, performing graph convolution to capture intricate atomic interactions. This block serves a dual purpose: first, it generates an irreducible node representation, providing a powerful input for the subsequent Hamiltonian Head. Second, it can be initialized with a pretrained EGNN or reprogrammed for other downstream tasks, constituting a unified framework for molecular modeling.  The Hamiltonian Head constructs both pairwise and many-body irreducible representations using the Clebsch-Gordon tensor product. These representations are then utilized to assemble both the non-diagonal and diagonal components of the Hamiltonian matrix. The model architecture and ablation studies are detailed in Figure~\ref{fig:wanet} and Table~\ref{tab:performance}, respectively.
\subsection{Node Convolution Layer}

For the Node Convolution Layer, we replace the traditional \(\mathrm{SO}(3)\) convolutions with Equivariant Spherical Channel Network (eSCN) \citep{passaro2023reducing,equiformer_v2}. The eSCN framework primarily utilizes \(\mathrm{SO}(2)\) linear operations, optimizing the computation of tensor products involved in the convolution process. Traditionally, \(\mathrm{SO}(3)\) convolutions operate on input irreducible representation (irrep) features \(u^{\ell_i}_{m_i}\) and spherical harmonic projections \(Y^{\ell_f}_{m_f}(\hat{\mathbf{r}}_{ts})\). By applying a rotation matrix \(D_{ts}\) to \(\hat{\mathbf{r}}_{ts}\), aligning it with the canonical axis where \(\ell = 0\) and \(m = 0\), the spherical harmonic projection \(Y^{\ell_f}_{m_f}(D_{ts}\hat{\mathbf{r}}_{ts})\) becomes non-zero exclusively for \(m_f = 0\). This condition simplifies the tensor product to \(C^{(\ell_o,m_o)}_{(\ell_i,m_i),(\ell_f,0)}\), which remains non-zero only when \(m_i = \pm m_o\). Subsequently, eSCN has demonstrated that this reformulation could be reparameterized using \(\mathrm{SO}(2)\) operations on the rotated tensors, and simplify the computation from $O(L^6)$ to $O(L^3)$, $L$ is the degree of the representaion. A detailed mathematical framework of this method is elaborated in the Appendix \ref{app:escn}.

\subsection{Hamiltonian Head}

\textbf{Sparse Mixture of Long-Short-Range Experts}\quad We introduce a variant of the Gated Mixture-of-Experts (GMoE) \citep{shazeer2017outrageously,clark2022unified,riquelme2021scaling,zoph2022st,jiang2024mixtral} model by incorporating a Mixture-of-Experts (MoE) layer tailored for pairwise molecular interactions. This enhancement draws inspiration from the Long-Short-Range Message Passing framework \citep{li2023longshortrange}, which differentiates between handling proximal and distal interactions through specialized layers. Our approach differentiates interaction dynamics based on distance, with closer pairs experiencing distinct interaction profiles compared to more distant pairs. This differentiation is achieved through a novel layer that substitutes the conventional pair interaction layer with a sparse assembly of expert modules, each functioning autonomously as a Pair Construction Layer. We define the Pair Construction layer with the function \( F_{\mathrm{pair}}^n \) for the \( n \)-th expert, and delineate the output of the MoE layer with \( N \) experts as:
\(
F_{\mathrm{MoE}}(x_t, x_s) = \sum_{n=1}^N p_n(x_t, x_s) \cdot F_{\mathrm{pair}}^n(x_t, x_s),
\)
where \( p_n(x_t, x_s) \) are the gating probabilities computed by the gating network, and \( \cdot \) denotes scalar multiplication.
The gating probabilities are obtained by applying the \emph{Softmax} function over the gating scores of all experts:
\(
p_n(x_t, x_s) = \frac{\exp\left( G_n\left( z \right) \right)}{\sum_{m=1}^N \exp\left( G_m\left( z \right) \right)},
\)
where \( z = \mathrm{rbf}(\Vert\mathbf{r}_{ts}\Vert_2) \) applies a radial basis function to the Euclidean distance \( \Vert\mathbf{r}_{ts}\Vert_2 \) with a distance cutoff, and \( G_n(z) \) represents the gating score for the \( n \)-th expert, computed as:
\(
G_n(z) = z \cdot W_{g_n} + \epsilon_n.
\)
Here, \( W_{g_n} \) are learned gating weights for expert \( n \), and \( \epsilon_n \) is injected noise to encourage exploration and promote load balancing among experts. Specifically, we use \emph{Gumbel} noise \citep{jang2016categorical}:
\(
\epsilon_n = -\log\left( -\log\left( U_n \right) \right), \quad U_n \sim \mathrm{Uniform}(0, 1).
\)
This noise enables a differentiable approximation of the top-\( K \) selection, allowing for sparse expert utilization while maintaining gradient flow during training. To further promote load balancing among the experts, we introduce an auxiliary load balancing loss \citep{shazeer2017outrageously}:
\(
\mathcal{L}_{\mathrm{load\_balancing}} = N \sum_{n=1}^N \left( \frac{\sum_{(t,s)} p_n(x_t, x_s)}{\sum_{(t,s)} 1} \right)^2,
\)
which encourages the gating network to allocate routing probabilities evenly across experts, preventing underutilization of any single expert.
The Pair Construction layer for each expert is defined as:
\(
\left( F_{\mathrm{pair}}^n(x_t, x_s) \right)^{l_o} = \sum_{l_i, l_j} W^n_{l_i, l_j, l_o} \left( x_s^{l_i} \otimes x_t^{l_j} \right)^{l_o},
\)
where \( x_s^{l_i} \) and \( x_t^{l_j} \) are the \( l_i \)-th and \( l_j \)-th irreducible representations of source node \( s \) and target node \( t \), respectively; \( W^n_{l_i, l_j, l_o} \) are the learned weights that couple these representations into the output representation \( l_o \); and \( \otimes \) denotes the tensor product.

\textbf{Many-Body Interaction Layer} \quad
Considering many-body interactions for the diagonal components of the Hamiltonian in molecular systems captures essential electron correlation effects~\citep{szabo2012modern,jensen2017introduction}. These interactions provide a more accurate representation of the collective behavior of atoms, beyond pairwise approximations. This leads to a precise description of key quantum phenomena like electron delocalization and exchange interactions~\citep{szabo2012modern}. For this purpose, we employ the methodologies of the MACE framework~\citep{batatia2022mace,kovacs2023evaluation,batatia2023foundation}. Central to MACE is the adept conversion of first-order features into higher-order features using the so-called \textit{density trick} \citep{duval2023hitchhiker}. This procedure initiates with the formation of generalized Clebsch-Gordon tensor products from the first-order features:

\[
B_{M, \nu}^{L} = \sum_{lm} C_{lm, \nu}^{LM} \prod_{\xi=1}^{\nu} w_{l_{\xi}} f_{m}^{l}; \quad \text{where } lm = (l_{1}m_{1}, \ldots, l_{\nu}m_{\nu}),
\]

where \( f_{m}^{l} \) represents the input tensor, and \( B_{M,\nu}^{L} \) the resultant tensor for the \(\nu\)-body. The coefficients \( C_{lm,\nu}^{LM} \) are the generalized Clebsch-Gordan coefficients, ensuring the \( L \)-equivariance of the output tensor \( B_{M,\nu}^{L} \). Moreover, \( C_{lm,\nu}^{LM} \) is notably sparse and can be pre-computed efficiently.

\section{Experiments}
% In this section, we evaluate WANet with WALoss on the QH9 and PubChemQH datasets. Our evaluation metrics includes the MAE for Hamiltonian, $\boldsymbol{\epsilon}_\mathrm{HOMO}$, $\boldsymbol{\epsilon}_\mathrm{LUMO}$, $\boldsymbol{\epsilon}_\Delta$, $\boldsymbol{\epsilon}_\mathrm{occ}$, $\boldsymbol{\epsilon}_\mathrm{orb}$, and System Energy, cosine similarity for the eigenvectors $\mathbf{C}$, and \textit{relative} SCF Iterations compared to initial guess. Detailed descriptions of these metrics are provided in the Appendix~\ref{app:eval-metrics}.
In this section, we evaluate WANet with WALoss on the QH9 and PubChemQH datasets. Our evaluation metrics include MAE for the Hamiltonian, $\boldsymbol{\epsilon}_\mathrm{HOMO}$, $\boldsymbol{\epsilon}_\mathrm{LUMO}$, $\boldsymbol{\epsilon}_\Delta$, $\boldsymbol{\epsilon}_\mathrm{occ}$, cosine similarity for the eigenvectors $\mathbf{C}$, and \textit{relative} SCF iterations compared to the initial guess, which are commonly used in previous works~\citep{unke2021seequivariant,yu2023efficient,yu2024qh9}. Additionally, we introduce two new physics-related metrics—MAE for $\boldsymbol{\epsilon}_\mathrm{orb}$ and System Energy—to provide a more comprehensive evaluation. Detailed descriptions of these metrics are provided in Appendix~\ref{app:eval-metrics}.% We also include our results on QH9 in Appendix~\ref{app:qh9}.

\subsection{Results on the PubChemQH dataset}
\textbf{Dataset Generation Process} \quad
In our study, we investigated the scalability of Hamiltonian learning by utilizing a \textit{CUDA-accelerated SCF implementation}~\citep{cudft} to perform computational quantum chemistry calculations, thereby generating the PubChemQH dataset. We began with geometries from the PubChemQC dataset  by~\citep{nakata2023pubchemqc}, selecting only molecules with a molecular weight above 400. This filtration process resulted in a dataset comprising molecules with 40 to 100 atoms, totaling over 50,000 samples.
% We relaxed the molecular geometries semi-empirically. To accurately determine the Hamiltonian matrices for our calculations, we fine-tuned the hyperparameters of the density functional theory (DFT) algorithms. % We adjusted the grid density to level 3, ensuring precise electronic density computations. The SCF convergence was stringently controlled by setting the SCF tolerance to \(10^{-13}\) and the gradient threshold to \(3.16 \times 10^{-5}\), which facilitated tight convergence of the final states. 
We chose the B3LYP exchange-correlation functional~\citep{lee1988development,beeke1993density,vosko1980accurate,stephens1994ab} and the Def2TZV basis set~\citep{weigend2005balanced,weigend2006accurate} to approximate electronic wavefunctions. Generating this comprehensive dataset represents a substantial computational effort, requiring approximately \textit{one month of continuous processing using 128 NVIDIA-V100 GPUs}. We provide a comparison between the PubChemQH dataset and the QH9 dataset in Appendix~\ref{app:dataset}.

\begin{table}[h]
    \centering
    \vspace{-3mm}
    \caption{\small{PubChemQH experimental results. The energy units are presented in kcal/mol. N/A indicates that the metric is not applicable to a specific model. The best-performing models are highlighted in bold. Detailed training setups are provided in Table~\ref{tab:experimental_settings} and Table~\ref{tab:unimol_settings}. }} 
    \resizebox{\textwidth}{!}{
        \begin{tabular}{lccccccccccc}
            \toprule
            Model & Hamiltonian MAE $\downarrow$ & $\epsilon_\mathrm{HOMO} \text{ MAE} \downarrow$ &  $\epsilon_\mathrm{LUMO} \text{ MAE}\downarrow$ &  $\epsilon_\Delta \text{ MAE}\downarrow$ & $\epsilon_\mathrm{occ} \text{ MAE}\downarrow$ & $\epsilon_\mathrm{orb} \text{MAE}\downarrow$ & $\mathbf{C}$ \text{ Similarity} $\uparrow$ & System Energy MAE $\downarrow$ & relative SCF Iterations$\downarrow$  \\
            \midrule
            \rowcolor{gray!10}QHNet            & 0.7765   & 71.250 & 83.890  & 5.790 & 2087.45 & 1532.672 & 2.32\% & 65721.028 & 371\% \\
            \rowcolor{gray!10}WANet            & 0.6274   & 60.140 & 62.35   & 4.723 & 734.258 & 502.43   & 3.13\% & 63579.233 & 334\% \\
            \rowcolor{cyan!10} QHNet w/ WALoss  & 0.5207   & 13.945 & 14.087 & 4.3982 & 21.805 & 10.930   & 46.66\% & 75.625 & 90\% \\
            \rowcolor{cyan!10} PhisNet w/ WALoss & 0.5166  & 11.872 & 12.075  & 4.054 & 20.045     &  8.917     & 46.78\%   & 60.325      & 90\%   \\
            \rowcolor{cyan!10} WANet w/ WALoss   & \textbf{0.4744}   & \textbf{0.7122} & \textbf{0.730} & \textbf{1.327} & \textbf{18.835} & \textbf{7.330} & \textbf{48.03\%} & \textbf{47.193} & \textbf{82\%} \\
            \midrule
            \rowcolor{gray!10}init guess (\texttt{minao}) & 0.5512   & 29.430  & 28.521  & 4.955 & 42.740 & 35.183 & 0.3293 & 374.313 & 100\% \\
            \midrule
            \rowcolor{magenta!10}Equiformer V2 Regression & N/A & 6.955 & 6.562 & 3.222 & N/A & N/A & N/A & N/A & N/A \\
            \rowcolor{magenta!10}UniMol+ Regression & N/A & 12.250 & 9.472 & 13.001 & N/A & N/A & N/A & N/A & N/A \\
            \rowcolor{magenta!10}UniMol 2 Regression  & N/A & 9.573 & 7.414 & 10.638 & N/A & N/A & N/A & N/A & N/A \\
            \bottomrule
        \end{tabular}
    }
    \label{tab:pubchemqh_result}
    % \vspace{-0.5cm}
    \vspace{-3mm}
\end{table}

Table \ref{tab:pubchemqh_result} presents a comparative analysis of various models' performance on the PubChemQH dataset. 
Despite a higher Hamiltonian MAE, WANet with WALoss significantly outperforms the other models in practical utility, as evidenced by the System Energy MAE and the required SCF iterations. Specifically, the System Energy MAE for WANet with WALoss is dramatically reduced from 63579.233 kcal/mol to 47.193 kcal/mol. Additionally, the relative SCF iterations required for WANet with WALoss is only 82\%, compared to 371\% for QHNet and 334\% for WANet without WALoss. This substantial reduction in SCF iterations demonstrates the effectiveness of WALoss in accelerating the convergence process. Furthermore, it is worth noting that the initial guess matrix, although not achieving as low an MAE as QHNet or WANet without WALoss, shows improved utility with a System Energy MAE of 374.313 kcal/mol. This finding reinforces the idea that \textit{elementwise losses are insufficient}. Additional evidence supporting this idea is provided in Appendix~\ref{app:sad}.

\subsection{Results on the QH9 Dataset}
\label{app:qh9}
The QH9 dataset is a comprehensive quantum chemistry resource designed to support the development and evaluation of machine learning models for predicting quantum Hamiltonian matrices. Built upon the QM9 dataset, QH9 contains Hamiltonian matrices for 130,831 stable molecular geometries, encompassing molecules with up to nine heavy atoms of elements C, N, O, and F. These Hamiltonian matrices were generated using \texttt{pyscf} with the B3LYP functional~\citep{lee1988development,beeke1993density,vosko1980accurate,stephens1994ab} and the def2SVP basis set. 
Table \ref{tab:qh9_result} presents a comparative analysis of the performance of our model, WANet, against the baseline model, QHNet, on the QH9 dataset in both stable and dynamic settings. In the QH9-stable experiments, WANet demonstrates superior performance, achieving higher accuracy compared to QHNet. Specifically, WANet significantly reduces both the Hamiltonian and occupied energy MAE while improving the cosine similarity of $\mathbf{C}$. For the QH9-dynamic dataset, WANet consistently outperforms QHNet, further enhancing prediction accuracy. These results underscore the robustness and effectiveness of WANet in both stable and dynamic scenarios.

% \begin{table}[h]
%     \centering
%     \caption{Experiment results on QH9 dataset. The energy units are kcal/mol.
%     % \chang{1. Mention the unit. 2. Do we need so many digits? 3. Use $\epsilon_\mathrm{orb}$ MAE and $\mathrm{C}$ similarity.}
%     }
%     \resizebox{0.7\textwidth}{!}{
%         \begin{tabular}{llcccccccc}
%             \toprule
%            & Model & Hamiltonian MAE $\downarrow$  & $\epsilon_\mathrm{occ} \text{ MAE}\downarrow$  & $\mathrm{\mathbf{C}}$ similarity $\uparrow$   \\
%             \midrule
% \multirow{3}{*}{QH9 stable} &
% QHNet & 0.0513 & 0.5366 & 95.85\% \\
% &WAnet & \textbf{0.0502}& 0.5231 & 96.86\% \\
% &WAnet w/Waloss & 0.0914 & \textbf{0.4587} & \textbf{96.95\%}  \\
% \midrule
% \multirow{3}{*}{QH9 dynamic} &
% QHNet  & 0.0471 & 0.2744 & 97.13\% \\
% &WAnet & \textbf{0.0469} & 0.2614 & 99.68\% \\
% &WAnet w/WAloss & 0.0512 & \textbf{0.2500} & \textbf{99.81\%} \\
%             \bottomrule
%         \end{tabular}
%     }
%     \label{tab:qh9_result}
% \end{table}

\begin{table}[t]
    \centering
        \vspace{-12mm}\caption{\small{Experimental results on the QH9 dataset. The energy units are presented in kcal/mol.
    % \chang{1. Mention the unit. 2. Do we need so many digits? 3. Use $\epsilon_\mathrm{orb}$ MAE and $\mathrm{C}$ similarity.}
    }}
    \resizebox{0.65\textwidth}{!}{
        \begin{tabular}{llcccccccc}
            \toprule
           & Model & Hamiltonian MAE $\downarrow$  & $\epsilon_\mathrm{occ} \text{ MAE}\downarrow$  & $\mathrm{\mathbf{C}}$ similarity $\uparrow$   \\
            \midrule
\multirow{3}{*}{QH9 stable} &
QHNet & 0.0513 & 0.5366 & 95.85\% \\
&QHNet w/WAloss & 0.0780 & 0.4901 & 96.35\% \\
&WANet & \textbf{0.0502}& 0.5231 & 96.86\% \\
 &WANet w/WAloss & 0.0914 & \textbf{0.4587} & \textbf{96.95\%}  \\
\midrule
\multirow{3}{*}{QH9 dynamic} &
QHNet  & 0.0471 & 0.2744 & 97.13\% \\
&QHNet w/WAloss & 0.0495 & 0.2658 & 98.54\% \\
&WAnet & \textbf{0.0469} & 0.2614 & 99.68\% \\
 &WAnet w/WAloss & 0.0512 & \textbf{0.2500} & \textbf{99.81\%} \\
            \bottomrule
        \end{tabular}
    }
    \label{tab:qh9_result}
    \vspace{-8mm}
\end{table}

\subsection{Comparison with a Property Regression Model}

Conventional machine learning approaches typically employ property regression, mapping molecular features directly to the desired property value~\citep{blum,1367-2630-15-9-095003}. A common question arises: why use Hamiltonians instead of a property regression model? We argue that property regression methods often fail to incorporate underlying quantum mechanical principles, limiting their generalization capability. 
To illustrate this, we compared the performance of WANet with WALoss to a model utilizing Equiformer V2~\citep{liao2023equiformer}, UniMol+~\citep{lu2023highly}, and UniMol2~\citep{ji2024uni} with invariant regression heads, using identical training and test sets. As shown in Table~\ref{tab:pubchemqh_result}, WANet with WALoss demonstrates significantly lower MAE values in predicting key quantum chemical properties. Specifically, WANet with WALoss achieves an 88.88\% improvement in $\epsilon_\textrm{LUMO}$ MAE and a 58.81\% improvement in $\epsilon_\Delta$ MAE. Moreover, the Hamiltonian predicted by WANet with WALoss is not limited to specific properties. It enables the accurate calculation of various critical properties, such as electronic densities, dipole moments, and excited-state energies, all from a single model. Additionally, it can be applied to SCF acceleration. In contrast, the Equiformer V2 regression model is constrained to predicting a narrow set of specific properties, necessitating the training of a new model for each new property.

% This highlights the versatility and efficiency of our approach for obtaining comprehensive quantum chemical properties.

% These findings highlight the importance of incorporating domain knowledge and physical constraints into machine learning models for quantum chemistry. The approach of predicting the Hamiltonian matrix and deriving relevant properties offers a promising direction for accelerating quantum chemical calculations and advancing materials discovery and design.
\vspace{-3mm}
\subsection{Efficiency Evaluation of WANet}

\begin{wrapfigure}{r}{0.33\textwidth}
\vspace*{-0.4in}
    \centering
    \captionsetup[subfigure]{labelformat=empty}
    \begin{minipage}{\linewidth}
        \centering
        \includegraphics[width=\linewidth]{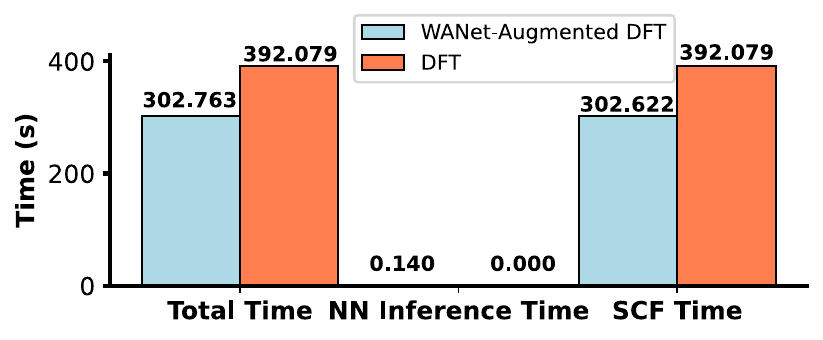}
        \subcaption{\small Figure \thefigure (\thesubfigure): \small Wall-clock comparison of WANet-Augmented DFT with traditional SCF iterations.}
        \label{fig:R3}\par\vfill
        \includegraphics[width=\linewidth]{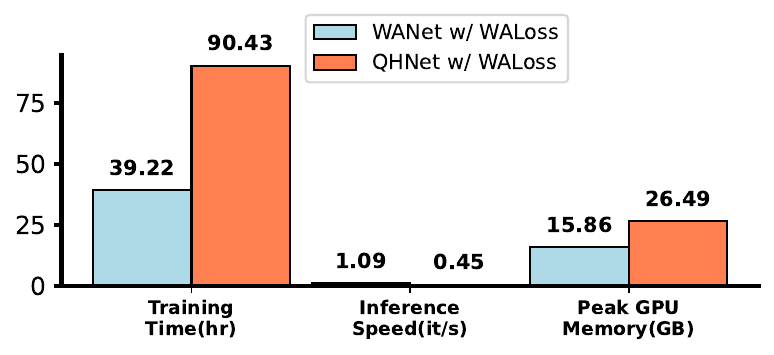}
        \subcaption{\small  (\thesubfigure): Comparison of training and inference efficiency and resource usage between QHNet and WANet on the PubChemQH dataset.}
        \label{fig:R4}
    \end{minipage}
    \label{fig:topk}
    \vspace*{-0.3in}
\end{wrapfigure} 
WANet exhibits efficiency advantages in two aspects: (1) its application to SCF relative to traditional DFT calculations, and (2) its efficiency and resource usage compared to existing state-of-the-art neural networks. Specifically, WANet can predict Hamiltonians for large molecular systems significantly faster than traditional DFT methods. Figure~\ref{fig:R3} presents a wall-clock time comparison between WANet-augmented SCF and traditional DFT calculations. Although the neural network evaluation introduces a small overhead, WANet substantially reduces the number of required SCF iterations, resulting in a faster overall computation time. This notable speed-up makes WANet particularly advantageous for applications requiring rapid predictions for large molecular systems, such as high-throughput virtual screening. Furthermore, WANet outperforms QHNet in training and inference efficiency on the PubChemQH dataset, offering faster training times, improved inference speeds, and lower peak GPU memory usage (Figure~\ref{fig:R4}).

\vspace{-3mm}
\subsection{Molecular Properties Beyond Energy Predictions}

\begin{wraptable}{r}{0.6\textwidth} % Wrap the table to the right, with 50% of text width
    \vspace{-5mm}
    \centering
    \caption{\small Additional property predictions. We evaluated WANet in deriving dipole moment and electronic spatial extent.}
    \resizebox{0.6\textwidth}{!}{ % Resizing to fit within the wraptable width
    \begin{tabular}{lcc}
        \hline
        \textbf{Model} & \textbf{Dipole Moment MAE (D)} & \textbf{Electronic Spatial Extent MAE (a.u.)} \\
        \hline
        Equiformer V2 Regression & \textbf{3.3221} & 0.1098 \\
        EGNN Regression & 4.3412 & 0.0789 \\
        ViSNet Regression & 4.5211 & 0.0982 \\
        Initial Guess & 4.0170 & 0.0161 \\
        WANet w/ WALoss & 3.3928 & \textbf{0.0076} \\
        \hline
    \end{tabular}
    } % Closing the resizebox
    \label{tab:additional_properties}
\vspace*{-0.2in}
\end{wraptable}
To validate the versatility of the Hamiltonian predicted by WANet with WALoss, we extended our experiments to include predictions beyond energy properties, specifically dipole moment and electronic spatial extent. As shown in Table~\ref{tab:additional_properties}, WANet with WALoss performs competitively in dipole moment prediction and excels in predicting electronic spatial extent. This demonstrates the model’s ability to generalize across multiple molecular properties, highlighting its potential for broader quantum mechanical calculations beyond energy-based properties. Further details on the derivation of these properties from the Hamiltonian are provided in Appendix~\ref{app:dipole}.

\vspace{-2mm}
\subsection{Scalability in Elongated Carbon Chain}

To evaluate the scalability of our model trained on PubChemQH, we conducted inference using elongated alkanes (C\textsubscript{x}H\textsubscript{2x+2})\footnote{Elongated alkanes (C\textsubscript{x}H\textsubscript{2x+2}) are only present in the test set.}, a series of saturated hydrocarbons. We compared three models: our model with WALoss, our model without WALoss, and an initial guess algorithm.  
\begin{wrapfigure}{r}{0.7\textwidth}
\vspace{-5mm}
\centering
\includegraphics[width=\linewidth]{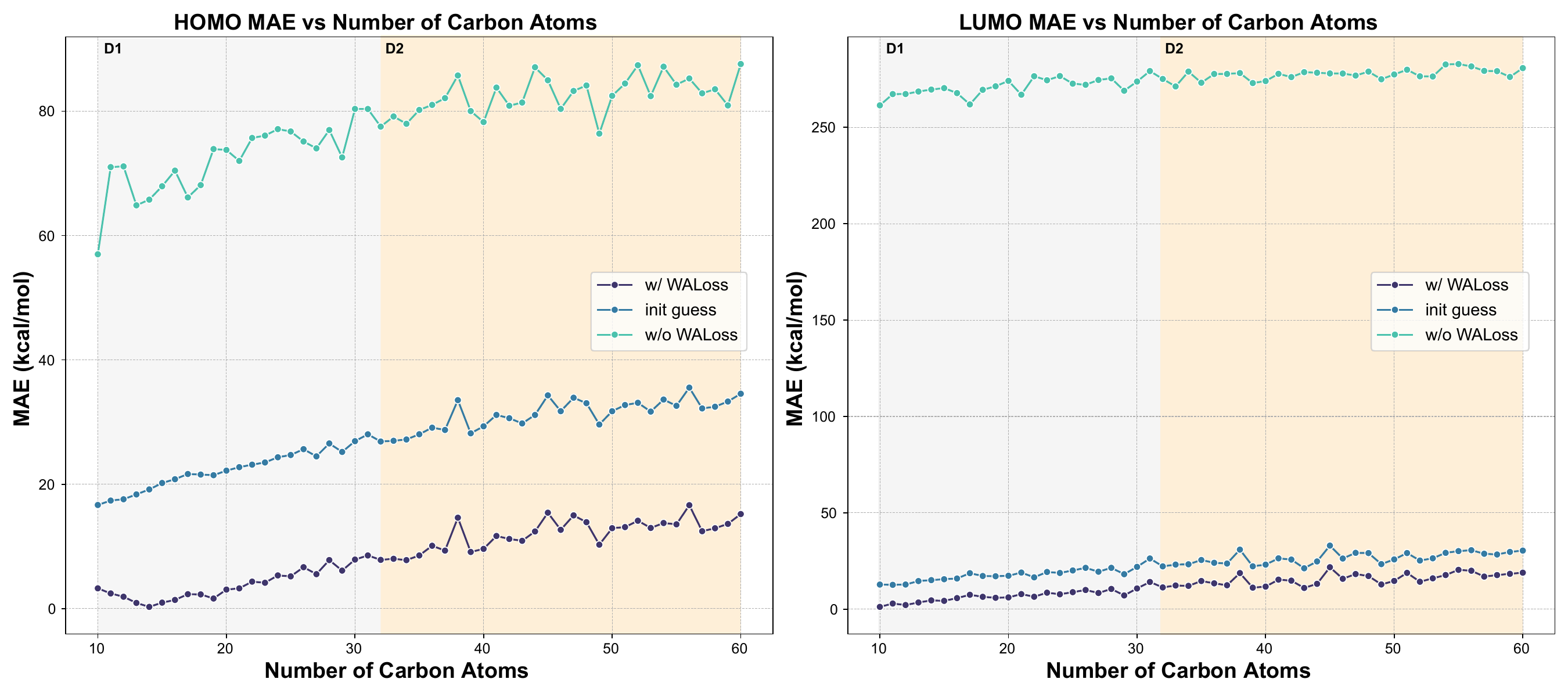}
\caption{
   \small Model performance in predicting HOMO and LUMO energies for elongated alkanes. The left panel shows the MAE for HOMO predictions, while the right panel shows the MAE for LUMO predictions. ``D1''  indicates that the atom count is within the range of the PubChemQH dataset, whereas ``D2'' indicates that the atom count exceeds this range. The models compared include WANet with WALoss (w/ WALoss), our model without WALoss (w/o WALoss), and the initial guess (init guess). Notably, our model with WALoss demonstrates superior performance in LUMO predictions and matches the best HOMO performance, particularly in the ``D2'' region.}
\label{fig:carbonchain}
\vspace{-4mm}
\end{wrapfigure} 
The results, shown in Figure~\ref{fig:carbonchain}, demonstrate that our model with WALoss achieves enhanced performance in predicting LUMO and HOMO energies, particularly in the ``D2'' region, where the atom count exceeds the range of the PubChemQH dataset. 
Notably, our model with WALoss also performs well on elongated alkanes with up to 182 atoms—three times the average atom count of the PubChemQH training set (60). These findings highlight the effectiveness of WALoss in enhancing the scalability and applicability of our model for predicting electronic properties in scalable homogeneous series, demonstrating its potential for application to larger and more complex molecular systems. Additional analysis of the scaling performance is provided in Appendix ~\ref{app:scaling}.

\subsection{Ablation Study on WALoss}
To evaluate the effectiveness of our proposed WALoss, we conducted an ablation study with three variations: full WALoss, naive WALoss, and WALoss without reweighting. The naive WALoss applies the Frobenius norm to the eigenvalues leveraging backpropagation through eigensolvers, defined as \(\mathcal{L}_{\text{naive}} = \frac{1}{n} \sum_{i=1}^{n} \| \hat{\boldsymbol{\epsilon}}_{_{(i)}} - \boldsymbol{\epsilon}_{_{(i)}}^\star \|^2_{\mathrm{F}}\), where \(\hat{\boldsymbol{\epsilon}}_{_{(i)}}\) and \(\boldsymbol{\epsilon}_{_{(i)}}^\star\) are derived from solving generalized eigenvalue problems. The WALoss without reweighting calculates the loss uniformly across all eigenvalues. As shown in Table \ref{tab:ablation-waloss}, the full WALoss achieves the lowest MAEs across all metrics. The naive WALoss performs poorly, highlighting several challenges associated with optimizing the naive loss function. Removing reweighting also degrades performance, though not as drastically. Overall, these results validate the design choices in formulating WALoss to improve Hamiltonian prediction.
\begin{table}[h]
\vspace{-3mm}
    \centering
    \caption{\small{Ablation study of WALoss on the PubChemQH dataset. The best-performing models are highlighted in bold.}}
    \resizebox{\textwidth}{!}{
        \begin{tabular}{lcccccccccc}
            \toprule
            Model & Hamiltonian MAE $\downarrow$ & $\epsilon_\mathrm{HOMO} \text{ MAE} \downarrow$ &  $\epsilon_\mathrm{LUMO} \text{ MAE}\downarrow$ &  $\epsilon_\Delta$ & $\epsilon_\mathrm{occ} \text{ MAE}\downarrow$ & $\epsilon_\mathrm{orb} \text{MAE}\downarrow$ & C$\uparrow$ & System Energy MAE $\downarrow$ & SCF Iteration $\downarrow$  \\
            \midrule
            Naive Loss           & 0.4912  & 50.174 & 55.630 & 3.634 & 632.220  & 486.322 & 5.36\% & 13562.7 & 306\% \\
            WALoss without Reweighting  &0.4973&8.241&7.993&3.988& 41.230 & 21.614 & 28.28\% & 55.492 &  88\%       \\
           WALoss Complete & \textbf{0.4744} & \textbf{0.7122}  & \textbf{0.730}  & \textbf{1.327} & \textbf{18.835}   & \textbf{7.330}  & \textbf{48.03\%} & \textbf{47.193} & \textbf{82\%} \\ 
            \bottomrule
        \end{tabular}

  }
    \label{tab:ablation-waloss}
\vspace{-7mm}
\end{table}

% \section{
% Related Work}
% \vspace{-1.5mm}
% \textbf{Predicting Kohn-Sham Hamiltonians} \quad Early methods used kernel ridge regression~\citep{hegde2017large}, while recent approaches employ neural networks, including direct wave function prediction~\citep{schutt2019machine,hermann2020deep}. Equivariant~\citep{unke2021se3,yu2023efficient,li2022marriage,zhong2023transferable} and hybrid architectures~\citep{yin2024harmonizing} predict molecular Hamiltonians. Novel training methods tackle data scarcity~\citep{zhang2024selfconsistency}, and benchmarks standardize evaluations~\citep{khrabrov2022nabla,yu2024qh9,khrabrov2024nabla2dftuniversalquantumchemistry}. For more related work, please refer to Appendix~\ref{app:related-work}.

\section{Conclusion and Limitations}
\vspace{-1.5mm}
In this work, we introduced WALoss, a loss function designed to improve the accuracy of predicted Hamiltonians. Our experiments demonstrate that incorporating WALoss achieves state-of-the-art performance by reducing prediction errors and accelerating SCF convergence. Additionally, we introduced a new dataset, \textit{PubChemQH}, and an efficient model, \textit{WANet}. However, limitations remain, such as the high computational cost of generating large training sets. Despite these challenges, deep learning approaches incorporating WALoss show great promise in advancing computational chemistry and materials science.

\section*{Acknowledgment}
Tha authors would like to thank Gaoyuan Wang from Yale University for the discussion on the scaling of the eigenvalues. This work is supported by ALW professorship funds to Mark Gerstein.

\bibliography{Styles/neurips2024}
\bibliographystyle{iclr2025_conference}

\appendix

\newpage

\DoToC
\newpage

\section{Notation Summary}
\begin{table}[h]
\centering
\caption{Summary of main notation used in the paper.}
\resizebox{\textwidth}{!}{%
\begin{tabular}{ll}
\toprule
Notation & Description \\
\midrule
$\mathcal{M} := \left\{ \mathbf{Z}, \mathbf{R} \right\}$ & Molecular system defined by nuclear charges $\mathbf{Z}$ and positions $\mathbf{R}$ \\
$\mathbf{r} \in \mathbb{R}^3$ & Spatial coordinate of an electron \\
$\psi_{i}(\mathbf{r})$ & Single-electron orbital/wavefunction $i$ as a function of electron spatial coordinate $\mathbf{r}$ \\
$\rho(\mathbf{r})$ & Electron density as a function of spatial coordinate $\mathbf{r}$ \\
$B$ & Number of basis functions in the basis set used to represent orbitals \\
${\phi_{\alpha}(\mathbf{r})}_{\alpha=1}^{B}$ & Basis set consisting of $B$ basis functions $\phi_{\alpha}(\mathbf{r})$ \\
$\mathbf{C} \in \mathbb{R}^{B \times N}$ & Coefficient matrix where each column contains the basis function coefficients for an orbital \\
$\mathbf{H} \in \mathbb{R}^{B \times B}$ & Hamiltonian matrix in the chosen basis representation \\
$\mathbf{S} \in \mathbb{R}^{B \times B}$ & Overlap matrix with elements $S_{\alpha\beta} := \int \phi^\dagger_{\alpha}(\mathbf{r}) \phi_{\beta}(\mathbf{r}) \mathrm{d}\mathbf{r}$ \\
$\boldsymbol{\epsilon} \in \mathbb{R}^{N \times N}$ & Diagonal matrix containing orbital energies (eigenvalues of the Hamiltonian) \\
$\hat{\mathbf{H}}_{\theta}(\mathcal{M})$ & Machine learning model parameterized by $\theta$ for predicting the Hamiltonian $\mathbf{H}$ from a molecular structure $\mathcal{M}$ \\
$\mathcal{L}_{\text{WA}}$ & Wavefunction Alignment Loss function \\
$\hat{\boldsymbol{\epsilon}}, \boldsymbol{\epsilon}$ & Predicted and ground truth eigenvalues (orbital energies) of $\hat{\mathbf{H}}$ and $\mathbf{H}$ \\
$\hat{\mathbf{C}}, \mathbf{C}$ & Predicted and ground truth eigenvectors (basis coefficients) of $\hat{\mathbf{H}}$ and $\mathbf{H}$ \\
$\mathbf{H}^\star$ & Ground truth Hamiltonian matrix \\
$\hat{\mathbf{M}}$ & Predicted transformed Hamiltonian matrix in an orthogonal basis \\
$\mathbf{G}$ & Matrix obtained from the eigen decomposition of the overlap matrix $\mathbf{S}$ \\
$k$ & Number of occupied orbitals \\
$\mathcal{I}$ & Set of indices corresponding to the $k+1$ lowest eigenvalues \\
$\rho$ & Hyperparameter controlling the weight of occupied and unoccupied orbitals in the loss function \\
$n$ & Number of samples in the dataset / the $n$-th experts in MoE nodel \\
$\mathcal{D}$ & Dataset used for training the machine learning model \\
$N$ & Number of electrons in the system / total number of experts in MoE model \\
$\mathbf{H}^{(0)}$ & Initial guess for the Hamiltonian matrix \\
$\mathbf{H}^{(k)}$ & Hamiltonian matrix at the $k$-th SCF iteration \\
$\mathbf{C}^{(k)}$ & Coefficient matrix at the $k$-th SCF iteration \\
$\boldsymbol{\epsilon}^{(k)}$ & Diagonal matrix of orbital energies at the $k$-th SCF iteration \\
$\delta$ & Convergence threshold for the SCF procedure \\
$\ell$ & Degree of the irreducible representation (irrep) in $\mathrm{SO}(3)$ equivariant networks \\
$D^{(\ell)}(g)$ & Wigner D-matrix representation of group element $g \in \mathrm{SO}(3)$ of degree $\ell$ \\
$Y^{(\ell)}(\hat{\mathbf{r}})$ & Real spherical harmonics of degree $\ell$ evaluated at unit vector $\hat{\mathbf{r}}$ \\
$C^{(\ell_3,m_3)}_{(\ell_1,m_1),(\ell_2,m_2)}$ & Clebsch-Gordan coefficients coupling irreps of degree $\ell_1$ and $\ell_2$ into irrep of degree $\ell_3$ \\
$\mathbf{r}_{ts}$ & Vector pointing from node $s$ to node $t$ \\
$\hat{\mathbf{r}}_{ts}$ & Unit vector pointing from node $s$ to node $t$ \\
$\mathcal{L}_{\text{total}}$ & Total loss function combining $\mathcal{L}_{\text{align}}$ and mean squared error (MSE) loss \\
$\lambda_1, \lambda_2, \lambda_3$ & Hyperparameters controlling the weights of different loss terms in $\mathcal{L}_{\text{total}}$ \\
\bottomrule
\end{tabular}%
}

\label{tab:notation}
\end{table}

\newpage
\section{Additional Related Work}
\label{app:related-work}
\textbf{Predicting Kohn-Sham Hamiltonians} \quad  Early work on  predicting Kohn-Sham Hamiltonians used kernel ridge regression~\citep{hegde2017large}, while newer approaches use neural networks (NNs). Some NNs predict the wavefunction itself~\citep{schutt2019machine,gastegger2020emergence,hermann2020deep}, while others use equivariant~\citep{unke2021se3,kochkov2021learning,yu2023efficient,li2022marriage,zhong2023transferable} or hybrid architectures~\citep{yin2024harmonizing} to predict the molecular Hamiltonian. A novel training method has been proposed to address the scarcity of labeled data~\citep{zhang2024selfconsistency}, and two benchmark datasets aim to standardize evaluation of molecular Hamiltonian prediction~\citep{khrabrov2022nabla,yu2024qh9,khrabrov2024nabla2dftuniversalquantumchemistry}.

\textbf{Equivariant Graph Neural Networks (EGNNs)} \quad It is often desired that machine learning
(ML) models exhibit equivariance to rotations,
translations, or reflections, which guarantee that they
respect certain physical symmetries. Foundational work
introduced group equivariant convolutional neural
networks (ECNNs)~\citep{cohen2016group}, whose
importance was underscored by a proof that
equivariance and convolutional structure are
equivalent given certain ordinary
constraints~\citep{kondor2018generalization}. One
appealing way to implement convolution is with
geometric graph neural networks (geometric GNNs),
which apply naturally to atomic systems by
encoding them as graphs embedded in
$\mathbb{R}^3$. Two important families of
geometric GNNs are invariant GNNs and Cartesian
equivariant GNNs. Over the past several years,
invariant GNNs have achieved state-of-the-art
results in predicting properties of
molecules, crystals, and other materials ~\citep{schutt2018schnet}~\citep{sanyal2018mt}~\citep{chen2019graph}~\citep{gasteiger2020directional}~\citep{liu2022spherical}~\citep{gasteiger2021gemnet}~\citep{wang2022comenet},
as well as in predicting the folding structure of
proteins~\citep{jumper2021highly}. Cartesian
equivariant GNNs have seen success in similar
areas, benefiting from the greater flexibility of their
representations~\citep{jing2020learning}~\citep{satorras2021n}~\citep{du2022se}~\citep{simeon2024tensornet}.
Cartesian equivariant GNNs have also seen recent innovation
in Cartesian equivariant transformer layers~\citep{frank2022so3krates}. 
A third significant family of geometric GNNs
is spherical equivariant GNNs, which use spherical
tensors rather than Cartesian tensors. As as a
result, spherical equivariant GNNs behave more
naturally under rotations and avail themselves of
many results of the representation theory of
$\mathrm{SO}(3)$. They have shown dexterity in tasks in
geometry, physics, and
chemistry~\citep{thomas2018tensor}; modeled dynamic
molecular systems~\citep{anderson2019cormorant};
enabled $\mathrm{SO}(3)$- and $\mathrm{SE}(3)$-equivariant transformer
layers~\citep{fuchs2020se}~\citep{liao2022equiformer};
accuarately and efficiently calculated interatomic
potentials~\citep{batzner2022e3}~\citep{batatia2022mace}~\citep{musaelian2023learning};
enabled $E(3)$-equivariant fluid
mechanical modeling ~\citep{toshev20233}; and
improved efficiency by reducing certain
convolution computations in $\mathrm{SO}(3)$ to equivalent
ones in $\mathrm{SO}(2)$~\citep{passaro2023reducing}. ~\citet{li2025e2former} propose to use $\mathrm{Wiger}$-$6j$ Convolution to shift computation from edges to nodes.

\newpage
\section{Additional Experiments}

\subsection{Evaluation Metrics}
\label{app:eval-metrics}
In this section, we provide a detailed description of the metrics  in the main-text:

\textbf{MAE for Hamiltonian} \quad The MAE for Hamiltonian assesses the accuracy of the predicted Hamiltonian matrices. This metric is crucial for evaluating the quality of the predicted electronic structure and its components, which are foundational in quantum chemistry calculations.

\textbf{MAE for $\boldsymbol{\epsilon}_\mathrm{HOMO}$} \quad The MAE for the Highest Occupied Molecular Orbital (HOMO) energy ($\boldsymbol{\epsilon}_\mathrm{HOMO}$) evaluates the precision of the predicted HOMO levels. Accurate HOMO predictions are essential as they influence a molecule's chemical reactivity and stability.

\textbf{MAE for $\boldsymbol{\epsilon}_\mathrm{LUMO}$} \quad The MAE for the Lowest Unoccupied Molecular Orbital (LUMO) energy ($\boldsymbol{\epsilon}_\mathrm{LUMO}$) measures the accuracy of LUMO level predictions. Accurate LUMO predictions are critical for understanding a molecule's electron affinity and chemical behavior.

\textbf{MAE for $\boldsymbol{\epsilon}_\Delta$} \quad The MAE for the energy gap between HOMO and LUMO ($\boldsymbol{\epsilon}_\Delta$) assesses the precision of this important property, which determines the electronic properties and conductivity of materials.

\textbf{MAE for $\boldsymbol{\epsilon}_\mathrm{occ}$} \quad The MAE for occupied orbital energies ($\boldsymbol{\epsilon}_\mathrm{occ}$) assesses the accuracy of predicted energies for all occupied molecular orbitals, providing a comprehensive measure of how well the model captures the electronic structure.

\textbf{MAE for $\boldsymbol{\epsilon}_\mathrm{orb}$} \quad The MAE for all orbital energies ($\boldsymbol{\epsilon}_\mathrm{orb}$) measures the discrepancies in predicted energies for both occupied and unoccupied orbitals. This metric evaluates the overall accuracy of the model in predicting the entire spectrum of orbital energies.

\textbf{MAE on Total Energy} \quad The MAE on total energy assesses the accuracy of the predicted total energies derived from Hamiltonian matrices using \texttt{pyscf}. This metric is crucial for validating the model's accuracy in predicting the overall energy of the system, which is fundamental for understanding molecular stability and reactions.

\textbf{Cosine Similarity for Wavefunction/Eigenvectors (C)} \quad The cosine similarity for wavefunction/eigenvectors (C) measures the similarity between the predicted and actual wavefunctions or eigenvectors of the system. High cosine similarity indicates that the predicted wavefunction distribution closely matches the actual distribution, which is important for accurately modeling electronic properties.

\textbf{SCF Iteration} \quad The SCF (Self-Consistent Field) iteration count evaluates the number of iterations required to achieve convergence in DFT (Density Functional Theory) calculations using the predicted Hamiltonian matrices when comparing with the inital guess. Mathmatically, it is defined as:

\[
\sigma = \frac{\text{Predicted Hamiltonian Iteractions}}{\text{Inital Hamiltonian Iteractions}}.
\]

This metric assesses the efficiency of the predicted matrices in expediting the DFT calculations.

% \begin{figure}[t]
% \vspace{-2mm}
% \begin{center}
% \includegraphics[width=\columnwidth]{pics/Figure - S1.pdf}

% \caption{
%  (A) The energy level for a saturated hydrocarbon system (C20H42). The energy is centered at 
% HOMO for comparison. (B) The predicted energy for differnet models on the satured hydrocarbon system. The y-axis indicates the energy in kcal/mol. The x-axis denotes the number of the carbon atoms in the elo-
% ngated carbon-chains. v2 indicates Equiformer V2. ‘wanet’ indicates WANet with WALoss. ‘groundtruth in-
% dicates DFT calculations. ’init’ indicates a Fock matrix initialization algorithm using minao. ‘without’ indica-
% tes WANet without WALoss. (C) The HOMO-LUMO gap prediction for different models}
% \label{fig:learning-curve}
% \end{center}
% \vspace{-7mm}
% \end{figure} 

\subsection{SAD Phenomenon}
\label{app:sad}
We present an additional graphical illustration of the SAD phenomenon discussed in the main text, depicting the learning curve of the QHNet model using only elementwise loss on a \textit{subset} of the \textit{PubChemQH} dataset \footnote{This differs from the dataset used in the maintexts.}. Figure~\ref{fig:learning-curve} shows that as the MAE decreases, the system energy exhibits significant fluctuations. Notably, when the Hamiltonian's MAE is around 0.3, the system energy MAE reaches 30,000 kcal/mol. This highlights the non-monotonic relationship between the Hamiltonian MAE and the resulting system energy. This extreme instability in derived properties, despite a seemingly small Hamiltonian MAE, is a key characteristic of the SAD phenomenon, hindering the applicability of the predicted Hamiltonian.

\begin{figure}[t]
\vspace{-2mm}
\begin{center}
\includegraphics[width=\columnwidth]{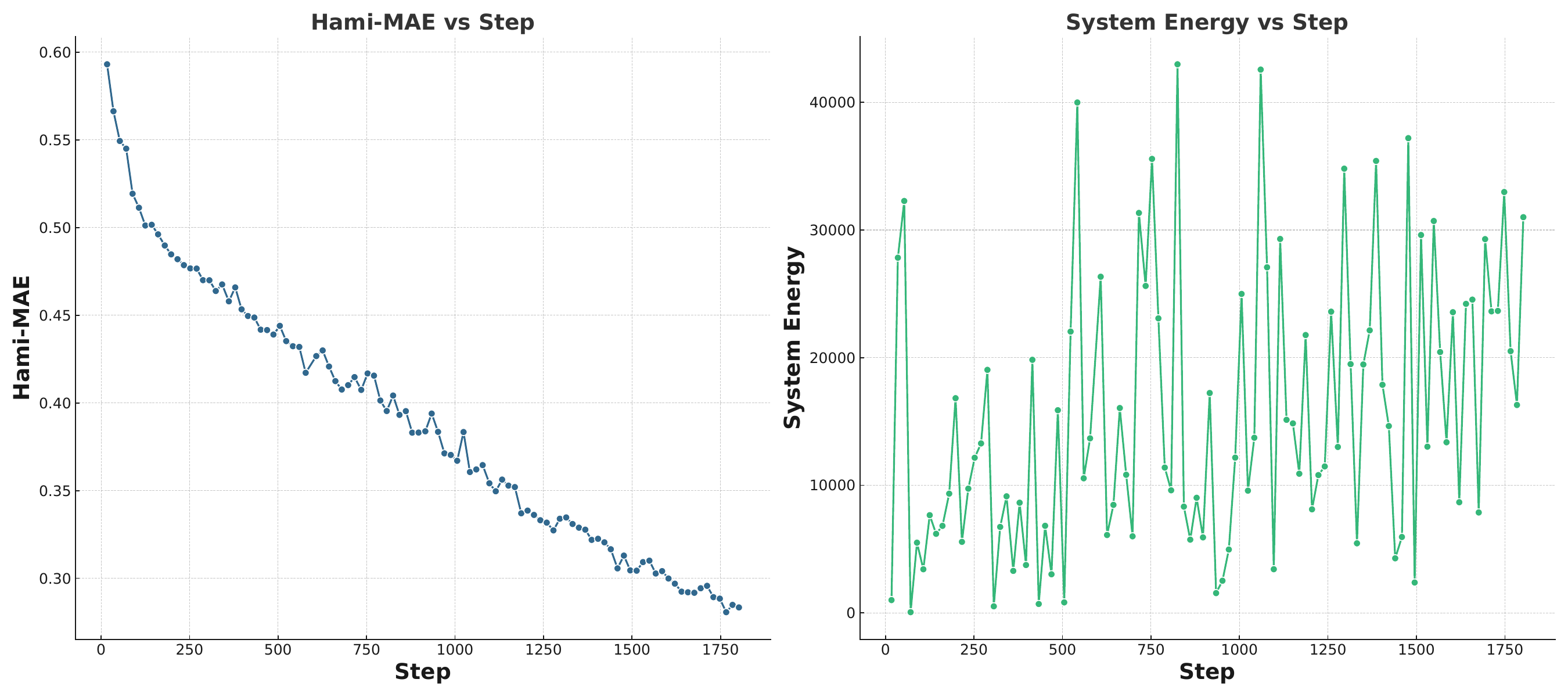}

\caption{
An additional graphical illustration of the SAD phenomenon discussed in the main text, depicting the learning curve of the QHNet model using only elementwise loss. The y-axis represents the metric, and the x-axis represents the steps. The figure shows that as the MAE decreases, the system energy exhibits significant fluctuations. }
\label{fig:learning-curve}
\end{center}
\vspace{-7mm}
\end{figure}

\subsection{Error Scaling Analysis}
\label{app:scaling}
In Figure~\ref{fig:carbonchain} and Table ~\ref{tab:scaling}, we analyzed the error trend scaling for "intensive" properties such as $\eps_{\text{HOMO}}$ and $\eps_{\text{LUMO}}$ with an increasing number of electrons. This challenge, widely recognized in the field~\citep{yu2024qh9, zhang2024selfconsistency}, highlights a difficult out-of-distribution (OOD) generalization scenario: models trained on smaller molecules often struggle to generalize to larger ones, limiting their applicability. Previous state-of-the-art models have exhibited significant scaling errors when extrapolating to larger systems~\citep{yu2024qh9, zhang2024selfconsistency}.

We conducted an error scaling analysis for saturated carbon chains to assess the scalability of our model. The error scaling coefficients for the HOMO and LUMO energies, as well as the energy gap, are presented in Table~\ref{tab:scaling}. Our results show significantly lower error scaling compared to baseline models for saturated carbon chains. Our model substantially outperforms the one without WALoss, as reflected by the scaling coefficients in Table R2, marking a considerable improvement in extrapolation capabilities within the field.

We further analyze the smallest eigenvalue distribution which is now included in Figure~\ref{fig:dist}.

\begin{figure}
    \centering
    \includegraphics[width=0.7\linewidth]{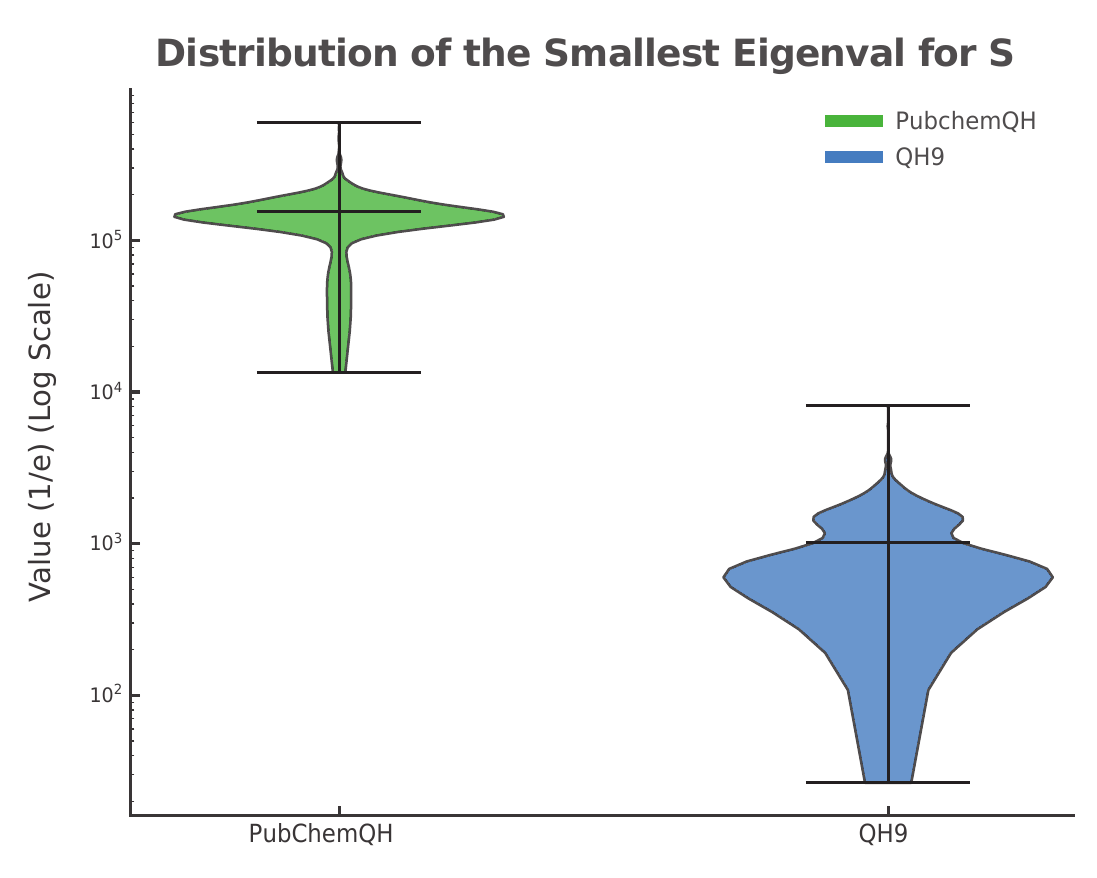}
    \caption{\small  the distribution of $\frac{\kappa(\mathbf{S})}{\|\mathbf{S}\|}$ on QH9 and PubChemQH dataset}
    \label{fig:dist}
\end{figure}

\begin{table}[h!]
\centering
\caption{Scaling Coefficients of HOMO, LUMO, and Energy Gap with Carbon Atom Count}
\resizebox{\textwidth}{!}{%
\begin{tabular}{lccc}
\hline
Model       & HOMO Scaling Coefficients & LUMO Scaling Coefficients & Gap Scaling Coefficients \\
\hline
WANet (with WALoss)  & \textbf{0.4211}                   & \textbf{0.4003}                   & \textbf{0.0206}                   \\
Initial Guess        & 0.4221                            & 0.4006                            & 0.0214                            \\
Without WALoss       & 1.3639                             & 2.7646                            & 1.4007                            \\
\hline
\end{tabular}
}
\label{tab:scaling}
\end{table}

\subsection{Overfitting Risk Evaluation}

We evaluated the risk of overfitting to the PubChemQH dataset. As shown in Table~\ref{tab:overfitting}, the evaluation metrics on unseen data are comparable to those on the training set, indicating that our model generalizes well and does not overfit.

\begin{table}[h]
\centering
\caption{Overfitting Risk Evaluation}
\resizebox{\textwidth}{!}{%
\begin{tabular}{lcccccccc}
\toprule
PubChem      & Hamiltonian MAE $\downarrow$ & $\epsilon_\mathrm{HOMO} \text{ MAE} \downarrow$ &  $\epsilon_\mathrm{LUMO} \text{ MAE}\downarrow$ &  $\epsilon_\Delta$ & $\epsilon_\mathrm{occ} \text{ MAE}\downarrow$ & $\epsilon_\mathrm{orb} \text{MAE}\downarrow$ & C$\uparrow$ & System Energy MAE $\downarrow$  \\
\midrule
Training MAE & 0.4736   & 6.073 & 4.414 & 3.753 & 17.850 & 6.159 & 48.01\% & 46.036        \\
Test MAE     & 0.4744   & 0.7122  & 0.730  & 1.327 & 18.835   & 7.330  & 48.03\% & 47.193  \\
\bottomrule
\end{tabular}
}
\label{tab:overfitting}
\end{table}

\subsection{Energy Predictions for Elongated Alkanes}

We have incorporated additional plots (Figure~\ref{fig:R}) that illustrate the predicted and actual values for both LUMO and HOMO energies, as well as the HOMO-LUMO gap. As the number of carbon atoms increases, both HOMO and LUMO energies exhibit an initial rise followed by a plateau with slight increases. Our model, utilizing the WALoss method, effectively captures this trend, outperforming baseline methods. Notably, Equiformer V2's HOMO predictions fail to generalize to saturated carbon systems, as these systems were not included in the training set. This observation underscores the advantage of Hamiltonian-based models in capturing physical principles and achieving better generalization compared to property regression models.

For isolated molecular systems, such as those studied here, computational chemists are primarily concerned with per-electron energy levels. To demonstrate this, we plot the energy levels within a reasonable window centered around the HOMO level and compare these with ground truth values obtained from DFT calculations. As shown in Figure~\ref{fig:R}
\textbf{A}, our model closely replicates the ground truth energy levels, outperforming baseline models. These results highlight the effectiveness of our approach in accurately predicting the electronic structure of molecules.

\begin{figure}[t]
\begin{center}
\includegraphics[width=\columnwidth]{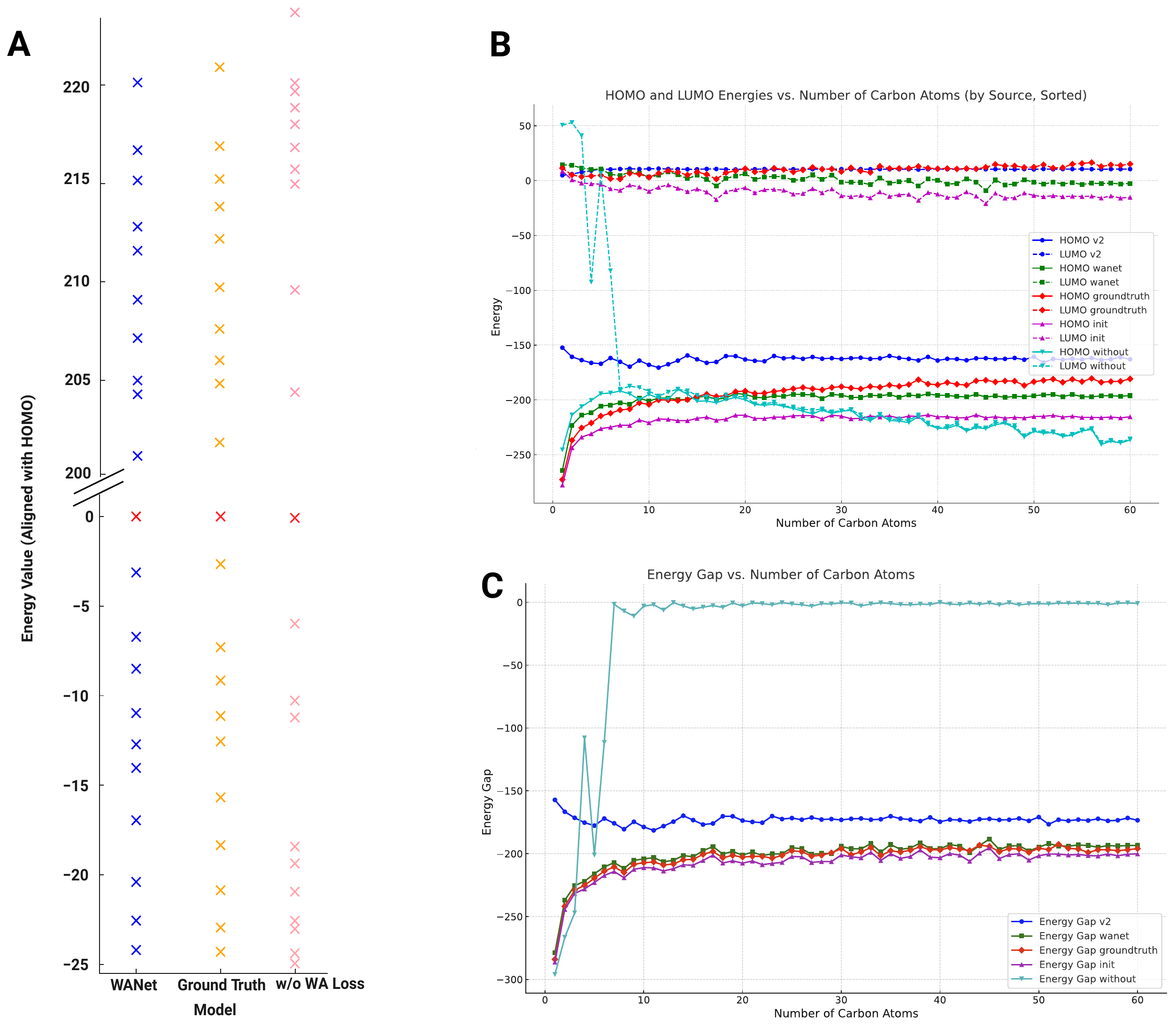}

\caption{
(\textbf{A}) The energy level for a saturated hydrocarbon system ($C_{20}H_{42}$). The energy is centered at 
HOMO for comparison. (\textbf{B}) The predicted energy for differnet models on the satured hydrocarbon system.
The y-axis indicates the energy in kcal/mol. The x-axis denotes the number of the carbon atoms in the elongated carbon-chains. v2 indicates Equiformer V2. `wanet' indicates WANet with WALoss. `groundtruth' indicates DFT calculations. `init' indicates a Fock matrix initialization algorithm using minao. `without' indicates WANet without WALoss. (\textbf{C}) The HOMO-LUMO gap prediction for different models.
}
\label{fig:R}
\vspace{-7mm}
\end{center}
\end{figure} 
\subsection{Additional Ablation Study}

In contrast to prior research, our study addresses the substantial challenges presented by the scale of the PubChemQC t-zvp dataset, particularly in providing the model with a strong numeric starting point. To overcome this difficulty, we shifted our focus to a more tractable objective: making predictions based on an easily obtainable initial guess. Our ablation study shows that while predictions based on the initial guess significantly improve performance in Hamiltonian prediction, they struggle when applied to the prediction of physical properties.

\begin{table}[h]
    \centering
    \caption{Performance Comparison of Baseline and Initial Guess Models on PubChemQH Dataset. The best models are bolded.}
    \resizebox{\textwidth}{!}{
        \begin{tabular}{lccccccccccc}
            \toprule
            Model & Hamiltonian MAE $\downarrow$ & $\epsilon_\mathrm{HOMO} \text{ MAE} \downarrow$ &  $\epsilon_\mathrm{LUMO} \text{ MAE}\downarrow$ &  $\epsilon_\Delta \text{ MAE}\downarrow$ & $\epsilon_\mathrm{occ} \text{ MAE}\downarrow$ & $\epsilon_\mathrm{orb} \text{MAE}\downarrow$ & $\mathbf{C}$ \text{ Similarity} $\uparrow$ & System Energy MAE $\downarrow$ & relative SCF Iterations$\downarrow$  \\
            \midrule
            WANet            & 0.6274 & 60.14 & 62.35 & 4.723 & 734.258  & 502.43 & 3.13\% & 63579.233 &  334\%\\
        WANet w/ Initial Guess           & \textbf{0.0379 } & 54.107 & 56.628 &3.618 & 695.418  & 481.513 & 4.42\% & 60078.633 &  306\% \\         
          WANet w/ Initial Guess \& WALoss & 0.4744   & \textbf{0.7122}  & \textbf{0.730}  & \textbf{1.327} & \textbf{18.835}   & \textbf{7.330}  & \textbf{48.03\%} & \textbf{47.193} & \textbf{82\%} \\
        \bottomrule
        \end{tabular}
    }
\end{table}

To evaluate the contributions of each component in the WANet model, we conducted ablation studies. The results are summarized in the Table~\ref{tab:performance}. The ablation studies confirm the importance of each architectural component in the WANet model.

\begin{table}[ht]
\centering
\caption{Ablation study results for the WANet model. The table shows the impact of different architectural components.}
\resizebox{\textwidth}{!}{
\begin{tabular}{ccccccccc}
\toprule
w/ SO(2) & w/ LSR-MoE & w/ Many Body & Inference Speed (it/s) & GPU Memory & Hamiltonian MAE & $\epsilon_\mathrm{HOMO}\text{ MAE}\downarrow$  & $\epsilon_\mathrm{LUMO}\text{ MAE}\downarrow$  & $\epsilon_\Delta \text{ MAE}\downarrow$ \\
\midrule
$\checkmark$ & $\checkmark$ &      & 1.34 & 13.87 & 0.5895 & 15.39 & 27.50 & 4.651\\
$\checkmark$ &     & $\checkmark$   & 1.13 & 12.80 & 0.4883 & 2.27  & 4.01  & 2.906\\
  & $\checkmark$ & $\checkmark$   & 0.51 & 24.47 & 0.4792 & 0.75  & 0.73  & 1.594\\
$\checkmark$ & $\checkmark$ & $\checkmark$   & 1.09 & 15.86 & 0.4744 & 0.71  & 0.73 & 1.327  \\
\bottomrule
\end{tabular}
}

\label{tab:performance}
\end{table}

\newpage

\section{eSCN convolution}
\label{app:escn}
\par{The equivariant Spherical Channel Network
(eSCN) is a graph neural network whose approach to
$\mathrm{SO}(3)$-equivariant convolutions significantly
reduces their computational burden. To see this,
let’s brush up on the notation used in the formula
for the $\ell_o$-th degree of the message $m_{ts}$
from source node $s$ to target node $t$ in an
$\mathrm{SO}(3)$ convolution: $W_{\ell_i, \ell_f, \ell_o}$
is a learnable weight; $x_s^{(\ell_i)}$ is the
$\ell_i$-th degree of the irrep feature of node
$s$, the source node; $\otimes$ is the tensor
product; $Y^{(\ell_f)}(\hat{\mathbf{r}}_{ts})$ is the
$\ell_f$-th degree spherical harmonic projection
of $\hat{\mathbf{r}}_{ts}$; and $\hat{\mathbf{r}}_{ts}$ is the
(normalized) relative position vector. Then, the
message is given by
\[
    m_{ts}^{(\ell_o)} = \sum_{\ell_i, \ell_f}W_{\ell_i, \ell_f, \ell_o}\Big(x_s^{(\ell_i)}\otimes Y^{(\ell_f)}(\hat{\mathbf{r}}_{ts})\Big)^{(\ell_o)}.
\] Traditionally, this involves an ordinary $\mathrm{SO}(3)$
tensor product, which is decomposed into its
irreps using the Clebsch-Gordon coefficients,
$C^{(\ell_o,m_o)}_{(\ell_i,m_i),(\ell_f,m_f)}$:
\[
    m_{ts}^{(\ell_o)} = \sum_{\ell_i, \ell_f}W_{\ell_i, \ell_f, \ell_o}\bigoplus_{m_o}
    \sum_{m_i, m_f}\big(x_s^{(\ell_i)}\big)_{m_i}C_{(\ell_i,m_i),(\ell_f,m_f)}^{(\ell_o,m_o)} \big(Y^{(\ell_f)}(\hat{\mathbf{r}}_{ts})\big)_{m_f}.
\] The tensor product, $\otimes$, has been exchanged 
for the proper Clebsch-Gordon coefficient. The
direct sum, $\bigoplus$, is included because the
space of tensors of degree $\ell_o$ has a basis
indexed by $m_o$.}

\par{However, the full $\mathrm{SO}(3)$ tensor product is
compute-intensive, $O(L^6)$, where $L$ is the
tensor degree. In practice, this means that only
tensors up to degree 2 or 3 are used, which is
especially unfortunate because higher-order
tensors allow more precise representation of
angular information. To lessen this computational
burden, eSCN rotates the irrep features,
$x_{s}^{(\ell_i)}$, of a node $x_s$, along with the
relative position vector, $\hat{\mathbf{r}}_{ts}$, by a
rotation chosen so that the relative position vector
aligns with the $y$-axis. Thus, it
may be interpreted that eSCN changes the irreps
into a more convenient basis,
computes the convolution, and changes back into
the original basis. This requires multiplication by a
change-of-basis Wigner D-matrix before and its
inverse after, but it is worthwhile because it
reduces the convolution to $O(L^3)$ by sparsifying
the Clebsch-Gordon coefficients. In particular,
eSCN guarantees that the output coefficient is
nonzero only if both $m_i = \pm m_o$ and $m_f = 0$
hold. The output tensors of some order, therefore,
are linear combinations of input tensors of that
order. It is not necessary to perform tensor
multiplication and decomposition. The non-zero
Clebsch-Gordon coefficients can then be denoted
$C_{(\ell_i,m),(\ell_f,0)}^{(\ell_o,m)}$, with
eSCN further guaranteeing that
$C_{(\ell_i,m),(\ell_f,0)}^{(\ell_o,m)} =
C_{(\ell_i,-m),(\ell_f,0)}^{(\ell_o,-m)}$, and
$C_{(\ell_i,m),(\ell_f,0)}^{(\ell_o,-m)} =
-C_{(\ell_i,-m),(\ell_f,0)}^{(\ell_o,m)}$. Thus,
eSCN guarantees that all but a few Clebsch-Gordon
coefficients will be zero, and provides simple
formulas involving those that remain.}

\par{More formally, let $R$ be a rotation matrix
chosen so that $R \cdot \hat{\mathbf{r}}_{ts} = (0,1,0)$,
and let $D^{(\ell)}$ be a Wigner D-matrix
representation of $R$ of degree $\ell$ (which
would more properly read $D^{(\ell)}(R)$, but
which is truncated for readability). Then, the
message can be written
\[
    m_{ts}^{(\ell_o)} = \big(D^{(\ell_o)}\big) ^{-1} \sum_{\ell_i, \ell_f} W_{\ell_i, \ell_f, \ell_o}\Big(D^{(\ell_i)} x_s^{(\ell_i)}\otimes Y^{(\ell_f)}(R \cdot \hat{\mathbf{r}}_{ts})\Big)^{(\ell_o)}.
\] Rewriting the tensor products using the Clebsch-Gordon coefficients yields
\[
    m_{ts}^{(\ell_o)} = \big(D^{(\ell_o)}\big)^{-1} \sum_{\ell_i, \ell_f}W_{\ell_i, \ell_f, \ell_o}\bigoplus_{m_o}
    \sum_{m_i, m_f}\big(D^{(\ell_i)} x_s^{(\ell_i)}\big)_{m_i}C_{(\ell_i,m_i),(\ell_f,m_f)}^{(\ell_o,m_o)} \big(Y^{(\ell_f)}(R \cdot \hat{\mathbf{r}}_{ts})\big)_{m_f}.
\] Since the Clebsch-Gordon coefficients are non-zero only when $m_f = 0$, there is no need to sum over $m_f$:
\[
    m_{ts}^{(\ell_o)} = \big(D^{(\ell_o)}\big)^{-1} \sum_{\ell_i, \ell_f}W_{\ell_i, \ell_f, \ell_o}\bigoplus_{m_o}
    \sum_{m_i}\big(D^{(\ell_i)} x_s^{(\ell_i)}\big)_{m_i}C_{(\ell_i,m_i),(\ell_f,0)}^{(\ell_o,m_o)} \big(Y^{(\ell_f)}(R \cdot \hat{\mathbf{r}}_{ts})\big)_0.
\] The rotation $R$ has been chosen so that $R \cdot \hat{\mathbf{r}}_{ts}$ yields a simple result, so the spherical harmonic term drops out:
\[
    m_{ts}^{(\ell_o)} = \big(D^{(\ell_o)}\big)^{-1} \sum_{\ell_i, \ell_f}W_{\ell_i, \ell_f, \ell_o}\bigoplus_{m_o}
    \sum_{m_i}\big(D^{(\ell_i)} x_s^{(\ell_i)}\big)_{m_i}C_{(\ell_i,m_i),(\ell_f,0)}^{(\ell_o,m_o)}.
\] Now, since one of the rotations on the righthand side has already dropped out, to make things simpler, define $\tilde{x}_s^{(\ell_i)} = D^{(\ell_i)} x_s^{(\ell_i)}$, yielding
\[
    m_{ts}^{(\ell_o)} = \big(D^{(\ell_o)}\big)^{-1} \sum_{\ell_i, \ell_f}W_{\ell_i, \ell_f, \ell_o}\bigoplus_{m_o}
    \sum_{m_i}\big(\tilde{x}_s^{(\ell_i)}\big)_{m_i}C_{(\ell_i,m_i),(\ell_f,0)}^{(\ell_o,m_o)}.
\] Since the Clebsch-Gordon coefficients are non-zero only when $m_i = \pm m_o$, the second summation can be omitted:
\[
    m_{ts}^{(\ell_o)} = \big(D^{(\ell_o)}\big)^{-1} \sum_{\ell_i, \ell_f}W_{\ell_i, \ell_f, \ell_o}\bigoplus_{m_o}
    \Big(\big(\tilde{x}_s^{(\ell_i)}\big)_{m_o}C_{(\ell_i,m_o),(\ell_f,0)}^{(\ell_o,m_o)} + \big(\tilde{x}_s^{(\ell_i)}\big)_{-m_o}C_{(\ell_i,-m_o),(\ell_f,0)}^{(\ell_o,m_o)}\Big).
\] Now, to avoid the need to sum over $\ell_f$, rather than learning parameters $W_{\ell_i, \ell_f, \ell_o}$, eSCN learns parameters $\tilde{W}_{m}^{(\ell_i, \ell_o)}$, defined for $m \geq 0$ as
\[
    \tilde{W}_{m}^{(\ell_i, \ell_o)} = \sum_{\ell_f}W_{\ell_i, \ell_f, \ell_o}C_{(\ell_i,m),(\ell_f,0)}^{(\ell_o,m)} = \sum_{\ell_f}W_{\ell_i, \ell_f, \ell_o}C_{(\ell_i,-m),(\ell_f,0)}^{(\ell_o,-m)},
\] and for $m < 0$ as
\[
    \tilde{W}_{m}^{(\ell_i, \ell_o)} = \sum_{\ell_f}W_{\ell_i, \ell_f, \ell_o}C_{(\ell_i,-m),(\ell_f,0)}^{(\ell_o,m)} = -\sum_{\ell_f}W_{\ell_i, \ell_f, \ell_o}C_{(\ell_i,m),(\ell_f,0)}^{(\ell_o,-m)}.
\] There exists a linear bijection between $W$ and $\tilde{W}$, so this parameterization loses no information. Finally, defining
\[\begin{aligned}
    \big(y_{ts}^{\ell_i,\ell_o}\big)_{m_o} & =
    \tilde{W}^{(\ell_i,\ell_o)}_{m_o}\big(\tilde{x}^{\ell_i}_s\big)_{m_o}
    -
    \tilde{W}^{(\ell_i,\ell_o)}_{-m_o}\big(\tilde{x}^{\ell_i}_s\big)_{-m_o},
    \: & m &> 0; \\
    \big(y_{ts}^{(\ell_i,\ell_o)}\big)_{m_o} & =
    \tilde{W}^{(\ell_i,\ell_o)}_{m_o}\big(\tilde{x}^{\ell_i}_s\big)_{-m_o}
    +
    \tilde{W}^{(\ell_i,\ell_o)}_{-m_o}\big(\tilde{x}^{\ell_i}_s\big)_{m_o},
    \: & m &< 0; \\
    \big(y_{ts}^{(\ell_i,\ell_o)}\big)_{m_o} & =
    \tilde{W}^{(\ell_i,\ell_o)}_{m_o}\big(\tilde{x}^{\ell_i}_s\big)_{m_o},
    \: & m &= 0,
\end{aligned}\] the message equation can very
    concisely be written as
\[
    m_{ts}^{\ell_o} = \big(D^{\ell_o}\big)^{-1}\sum_{\ell_i}\bigoplus_{m_o}\big(y_{ts}^{\ell_i,\ell_o}\big)_{m_o}.
\]}

\par{There is another way to interpret this that
is perhaps more intuitive. Fixing the direction of
the relative position vector, $\hat{\mathbf{r}}_{st}$
leaves a single rotational degree of freedom: the
roll rotation about this axis. Thus, eSCN reduces
$\mathrm{SO}(3)$ convolution to $\mathrm{SO}(2)$ convolution. More
formally, define the colatitude angle $\theta \in
[0, \pi]$ and longitudinal angle $\phi \in [0,
2\pi]$. Now, using Legendre polynomials
$P^{(\ell)}_m(\theta)$, which depend only on the
colatitude angle, $\theta$, the real spherical
harmonic basis functions can be written 
\[
    Y_m^{(\ell)}(\theta, \phi) = P_m^{(\ell)}(\theta)e^{im\phi}.
\] Aligning the relative 
position vector with the $y$-axis fixes $\theta$,
leaving behind basis functions of the form
$e^{im\phi}$, which are the circular harmonic
basis functions, used in $\mathrm{SO}(2)$ convolution. A
convolution about $\phi$ can take advantage of the
Convolution Theorem, reducing convolution to
point-wise multiplication, further harvesting
efficiency gains.}

\newpage
\section{Additional Background}
\par{Quantum mechanics is most often approached as
the study of the Schr\"odinger equation, a linear
partial differential equation whose solutions are
called wavefunctions. In this paper, as is
typical in molecular modeling, the
time-independent Schr\"odinger equation will be
used:
\[
    H\psi = E \psi,
\] where $H$, the Hamiltonian operator, corresponds to the total energy of the system, $E$ is an eigenvalue of $H$ corresponding to the energy of $\psi$, and $\psi$ is an eigenfunction of $H$, also called a wavefunction, or a solution to the Schr\"odinger equation.}

\par{As alluded to, the Hamiltonian operator
contains all the information regarding the kinetic
and potential energies for all particles of a
system. Thus,
\[\begin{aligned}
    H = -\frac{1}{2}\sum_{i=1}^{N}\nabla_i^2
    & - \frac{1}{2}\sum_{A=1}^{M}
    \frac{1}{M_A}\nabla_A^2 -
    \sum_{i=1}^{N}\sum_{A=1}^{M}
    \frac{Z_a}{r_{iA}} \\
    & + \sum_{i=1}^{N}\sum_{j>1} \frac{1}{r_{ij}}
    + \sum_{A=1}^{M}\sum_{B>A} \frac{Z_A
    Z_B}{R_{AB}},
\end{aligned}\] 
where the first summation is
due to the kinetic energy of the electrons,
the second summation is due to the kinetic energy of the
nuclei, the first double summation is due to the
potential of the attraction between
electrons and nuclei, the second double
summation is due to the potential of the repulsion between
electrons and electrons, and the third double
summation is due to the potential of the
repulsion between nuclei.}

\par{This formula is unwieldy, however, and it can
be simplified significantly without discarding
much information that would be useful in chemical
prediction. The mass difference between electrons
and protons, which is the minimum mass difference
between electrons and nuclei, is more than 3
orders of magnitude. Thus, given the same kinetic
energy, electrons will be traveling several times
faster than nuclei. From the perspective of the
electrons, the nuclei are nearly fixed, and from
the perspective of the nuclei, the electrons
change their position instantaneously. Therefore,
it suffices to consider the positions of the
nuclei fixed, setting 
\[
    -\frac{1}{2}\sum_{A=1}^{M} \frac{1}{M_A}\nabla_A^2,
\]
nuclear kinetic energy, to 0, and
\[
    \sum_{A=1}^{M}\sum_{B>A} \frac{Z_A Z_B}{R_{AB}},
\] 
nuclear repulsive potential, to a constant. This
is called the Born-Oppenheimer Approximation, and
it leaves behind the so-called electronic
Hamiltonian,
\[
    H_{\mathrm{elec}} = -\frac{1}{2}\sum_{i=1}^{N}\nabla_i^2 -
    \sum_{i=1}^{N}\sum_{A=1}^{M}
    \frac{Z_a}{r_{iA}} 
    + \sum_{i=1}^{N}\sum_{j>1} \frac{1}{r_{ij}},
\] which can more succinctly be written
\[
    H_{\mathrm{elec}} = T + V_{\mathrm{Ne}} + V_{\mathrm{ee}}.
\] The electronic Hamiltonian is still solved for its eigenfunctions, giving
\[ 
    H_{\mathrm{elec}} \Psi_{\mathrm{elec}}(\mathbf{r}_1,\hdots,\mathbf{r}_N) =  E_{\mathrm{elec}} \Psi_{\mathrm{elec}}(\mathbf{r}_1,\hdots,\mathbf{r}_N),
\] ignoring for simplicity’s sake electron spin to write $\Psi_{\mathrm{elec}}$ as a function of the positions of the electrons only. For readability, the subscripts of $H_{\mathrm{elec}}$, $E_{\mathrm{elec}}$, and $\Psi_{\mathrm{elec}}$ are henceforth omitted.}

\par{For a system of $N$ electrons, $\Psi$ is a
function of $3N$ arguments, one for each spatial
dimension of each electron. The Hartree-Fock
Approximation, also called the Hartree Product,
further simplifies this equation by decomposing
the wavefunction into the product of $N$ wave
functions, each of three arguments, corresponding
to each electron individually:
\[
    \Psi(\mathbf{r}_1,\hdots,\mathbf{r}_n) \approx
    \psi_1(\mathbf{r}_1) \cdots \psi_N(\mathbf{r}_N).
\] It is impossible to observe the wavefunction itself, but the wavefunction can be used to derive the probability of observing the system’s electrons anywhere in space. In particular, the probability of observing the electrons at positions $\mathbf{r}_1$ through $\mathbf{r}_N$ is the square of the amplitude of the wavefunction with $\mathbf{r}_1$ through $\mathbf{r}_N$ as input:
\[
    \vert \Psi(\mathbf{r}_1, \hdots, \mathbf{r}_N) \vert^2 = \Psi^{\dagger}(\mathbf{r}_1, \hdots, \mathbf{r}_N)\Psi(\mathbf{r}_1, \hdots, \mathbf{r}_N).
\] Given the Hartree-Fock Approximation, this equation can be converted into an electron density function that sums over the wavefunction of each electron individually,
\[
    n(\mathbf{r}) = 2 \sum_{i=1}^{N} \psi_{i}^{\dagger}(\mathbf{r}) \psi_{i}(\mathbf{r}),
\] multiplying by 2 to account for both spin up and spin down, which were neglected previously. This means that $n(\mathbf{r})$ gives the density of electrons at a point in space, $\mathbf{r}$. This is a function of only 3 inputs, but it contains much of the information that is observable from the full wavefunction, which, recall, is a function of $3N$ inputs.}

\par{Density Functional Theory (DFT), which allows
the electron density function, $n(\mathbf{r})$, to be
exploited, rests on two fundamental theorems.
First, the ground-state energy, $E_0$, which is the
smallest eigenvalue of the Hamiltonian and
corresponds to the energy of the lowest-energy
wavefunction, is a unique functional of the
electron density function. Second, the electron
density that minimizes the energy of this
functional is the true electron density, which
means that it is the electron density that
corresponds to the full solution of the
Schr\"odinger equation. Therefore, after the
problem has been reduced from one of $3N$ inputs to one of $3$
inputs, significant information about the original
problem can still be found.}

\par{It is worth examining the first statement in
more detail. Recall that a functional is a
function that maps functions to scalars. For
example, $F$, defined by 
\[
    F[f(x)] = \int_{-1}^{1}f(x)dx,
\] is a functional that takes in an
arbitrary real-valued function and gives out its
integral from $-1$ to $1$. Thus, for example, if
$f(x) = x^2 + 1$, then $F[f(x)] = \frac{8}{3}$.
The first statement, then, holds that there exists
a functional that uniquely determines the ground
state energy of a particular electron density
function. It promises that no more information is
needed. Written as an equation, $E_0 =
F(n(\mathbf{r}))$.}

\par{However, this set-up still relies on some
means of finding the individual-electron wave
function. This is provided by the Kohn-Sham
Equation:
\[
    \Big[\frac{1}{2}\nabla^2 + V(\mathbf{r}) + V_\mathrm{H}(\mathbf{r}) + V_{\mathrm{XC}}(\mathbf{r})\Big] \psi_i(\mathbf{r}) = \epsilon_i \psi_i(\mathbf{r}).
\] The meaning of each term is as follows. The value $\frac{1}{2}\nabla^2$ is kinetic energy. The value $V(\mathbf{r})$ is the potential due to the interaction between the electron and the nuclei. The value $V_\mathrm{H}(\mathbf{r})$, called the Hartree potential, is defined
\[
    V_\mathrm{H}(\mathbf{r}) = \int \frac{n(\mathbf{r} ^\prime)}{\vert \mathbf{r} - \mathbf{r} ^\prime \vert} d^3\mathbf{r}^\prime,
\] and it is due to the Coulomb repulsion between the electron and the electron density function, defined by all electrons in the system. This means that the Hartree potential includes the Coulomb repulsion between the electron and itself, since the electron itself is included in the electron density function, $n(\mathbf{r})$. This is an unphysical result, and it is one of several effects accounted for in the exchange-correlation potential, $V_{\mathrm{XC}}(\mathbf{r})$, which is defined
\[
    V_{\mathrm{XC}}(\mathbf{r}) = \frac{\delta E_{\mathrm{XC}}(\mathbf{r})}{\delta n(\mathbf{r})}.
\] This is the functional derivative of the exchange-correlation energy with respect to electron density. Exchange-correlation potential is due to quantum-chemical effects, and its form is not known because the exact form of $E_{\mathrm{XC}}(\mathbf{r})$ is not known. The Kohn-Sham Equation can be summarized as
\[
    H_{\mathrm{KS}}\psi_i(\mathbf{r}) = \epsilon_i \psi_i(\mathbf{r}).
\]
}

\par{It is evident from the definitions of the
terms in the Kohn-Sham Equation that the present
approach is circular. The Kohn-Sham Equation
count on the Hartree potential, $V_H(\mathbf{r})$.
The Hartree potential counts on the electron
density function, $n(\mathbf{r})$. The electron
density function counts on the single-electron
wavefunctions, $\psi_i(\mathbf{r})$. And the
single-electron wavefunctions count on the
Kohn-Sham Equation.}

\par{This is no problem. The following process
leverages this circularity to check the soundness
of a particular $n(\mathbf{r})$:
\begin{enumerate}
    \item Define an initial trial electron density
    function, $n_{\mathrm{trial}}(\mathbf{r})$.
    \item Solve the Kohn-Sham Equation using the
    trial electron density function,
    $n_{\mathrm{trial}}(\mathbf{r})$, to find the
    single-electron wavefunctions,
    $\psi_i(\mathbf{r})$.
    \item Calculate the electron density function,
    $n_{\mathrm{KS}}(\mathbf{r})$, implied by these
    single-electron wavefunctions, by
    $n_{\mathrm{KS}}(\mathbf{r}) = 2
    \sum_{i=1}^{N}\psi_i^{\dagger}(\mathbf{r})\psi_i(\mathbf{r})$.
    \item Compare the calculated electron density,
    $n_{\mathrm{KS}}(\mathbf{r})$, with the trial
    electron density,
    $n_{\mathrm{trial}}(\mathbf{r})$. If the two
    densities are the same, or nearly so, then
    this electron density is accepted as correct,
    and it can be used to find the ground-state
    energy, $E_0$. If not, the trial electron
    density is updated somehow, and the process
    repeats. \end{enumerate}}

\par{To make this process simpler and more
efficient, it is common to represent the
single-electron wavefunctions as linear
combinations of some predefined basis set
$\{\phi_\alpha(\mathbf{r})\}_{\alpha=1}^{B}$, where
$B$ is defined as the cardinality of the basis.
Often, the basis set is composed of atomic
orbitals, especially Gassian-type orbitals, which
are particularly convenient in calculation. The
expansion coefficients of these wavefunctions can
be organized in a matrix $\mathbf{C} \in \mathbb{R}^{B
\times N}$, where each column $i$ contains the
coefficients of wavefunction $i$. Each wave
function can then be recovered as
\[
    \psi_i(\mathbf{r}) = \sum_{\alpha=1}^{B}\mathbf{C}_{\alpha i} \phi_\alpha(\mathbf{r}).
\]}

\par{This permits the Kohn-Sham Equation as a
whole to be written in matrix form. Recall the
form $H_{\mathrm{KS}}\psi_i(\mathbf{r}) =
\epsilon_i \psi_i(\mathbf{r})$. Now, it is possible
to rewrite
$\psi_i(\mathbf{r})$ according to its decomposition
in the basis:
\[
    H_{\mathrm{KS}} \sum_{\alpha=1}^{B}\mathbf{C}_{\alpha i} \phi_\alpha(\mathbf{r})=
    \epsilon_i \sum_{\alpha=1}^{B}\mathbf{C}_{\alpha i} \phi_\alpha(\mathbf{r}).
\] Now, for any $\phi_{\beta}(\mathbf{r})$ in the basis, left-multiplying both sides by $\phi^\dagger_{\beta}(\mathbf{r})$ yields
\[
    \phi^\dagger_{\beta}(\mathbf{r}) H_{\mathrm{KS}} \sum_{\alpha=1}^{B}\mathbf{C}_{\alpha i} \phi_\alpha(\mathbf{r}) =
    \phi^\dagger_{\beta}(\mathbf{r}) \epsilon_i \sum_{\alpha=1}^{B}\mathbf{C}_{\alpha i} \phi_\alpha(\mathbf{r}),
\] and integrating with respect to $\mathbf{r}$ yields
\[
    \int \phi^\dagger_{\beta}(\mathbf{r}) H_{\mathrm{KS}} \sum_{\alpha=1}^{B}\mathbf{C}_{\alpha i} \phi_\alpha(\mathbf{r}) d\mathbf{r} =
    \int \phi^\dagger_{\beta}(\mathbf{r}) \epsilon_i \sum_{\alpha=1}^{B}\mathbf{C}_{\alpha i} \phi_\alpha(\mathbf{r}) d\mathbf{r}.
\] More succinctly, defining ${(H_{\mathrm{KS}})}_{\beta \alpha} = \int \phi^\dagger_\beta H_{\mathrm{KS}} \phi_\alpha = \langle \phi_\beta \vert H_{\mathrm{KS}} \vert \phi_\alpha \rangle$ and $S_{\beta \alpha} = \int \phi^\dagger_\beta \phi_\alpha = \langle \phi_\beta \vert \phi_\alpha \rangle$, this is
\[
    \sum_{\alpha=1}^{B}{(H_{\mathrm{KS}})}_{\beta \alpha} \mathbf{C}_{\alpha i} = \epsilon_i \sum_{\alpha = 1}^{B}S_{\beta \alpha} \mathbf{C}_{\alpha i}.
\] This allows the construction of the matrix $\mathbf{H}$, with elements 
$\mathbf{H}_{\beta \alpha}$ defined by $\mathbf{H}_{\beta \alpha} =
{(H_{\mathrm{KS}})}_{\beta \alpha} = \langle
\phi_\beta \vert H_{\mathrm{KS}} \vert
\phi_\alpha \rangle$, and the matrix $\mathbf{S}$, with
elements $\mathbf{S}_{\beta \alpha} = \langle \phi_\beta
\vert \phi_\alpha \rangle$. The matrix $\mathbf{H}$
represents the Hamiltonian matrix, which depends
on the coefficient matrix $\mathbf{C}$ and is calculated
using a method called Density-Fitting, whose time
complexity is $O(B^3)$. The matrix $\mathbf{S}$ is the
overlap matrix, which depends on the basis set
$\{\phi_\alpha(\mathbf{r})\}_{\alpha=1}^{B}$ and
accounts for its non-orthogonality. This means
that if the basis set is orthonormal, $\mathbf{S}$ is the
identity. The Kohn-Sham Equation in matrix form is
therefore
\[
    \mathbf{H}\mathbf{C} = \mathbf{S}\mathbf{C}\mathbf{\epsilon},
\] or 
\[
    \mathbf{H}(\mathbf{C}) \mathbf{C} = \mathbf{S}\mathbf{C}\mathbf{\epsilon},
\] making clearer that $\mathbf{H}$ depends on $\mathbf{C}$, where $\mathbf{\epsilon}$ is a diagonal matrix containing the orbital energies. This forms a generalized eigenvalue problem, the principal difficulty of which is the dependence of $\mathbf{H}$ on $\mathbf{C}$.}

\par{As a result, traditional DFT uses the
Self-Consistent Field method (SCF), which
iteratively refines the coefficient matrix by
approximating the Hamiltonian. Using superscripts
to denote the iteration to which each matrix
belongs, at each iteration, the Kohn-Sham Equation is
\[
    \mathbf{H}^{(k)}\big(\mathbf{C}^{(k-1)}\big)\mathbf{C}^{(k)} = \mathbf{S}\mathbf{C}^{(k)}\mathbf{\epsilon}^{(k)}.
\] Therefore, $\mathbf{H}^{(k)}$ is computed using $\mathbf{C}^{k-1}$, and the generalized eigenvalue problem is solved for $\mathbf{C}^{(k)}$ and $\mathbf{\epsilon}^{(k)}$. This process continues until convergence, which is formalized as $\Vert \mathbf{H}^{(k+1)} - \mathbf{H}^{(k)} \Vert \leq \delta$.}

\par{The objective of Hamiltonian prediction is to
eliminate the need for this computationally
expensive SCF iteration by directly estimating the
target Hamiltonian, $\mathbf{H}^\star$, for a given molecular
structure, $\mathcal{M}$. To predict the
Hamiltonian, a machine learning model
$\hat{\mathbf{H}}_\theta(\mathcal{M})$ is parameterized by
$\theta$, guided by an optimization process
defined as
\[
    \theta^\star = \arg \min_\theta \frac{1}{\vert \mathcal{D} \vert} \sum_{(\mathcal{M}, \mathbf{H}^\star_\mathcal{M}) \in \mathcal{D}} \mathrm{dist}(\hat{\mathbf{H}}_\theta(\mathcal{M}), \mathbf{H}^\star_\mathcal{M}), 
\] where $\mathcal{D}$ is the dataset, $\vert \mathcal{D} \vert$ is its cardinality, and $\mathrm{dist}(\cdot, \cdot)$ is a predefined metric. Machine learning models hold the potential to drastically improve the efficiency of DFT calculations without sacrificing accuracy.}

\subsection{Example}

In the example section, we explore the molecular structure and basis set details of water (\(\text{H}_2\text{O}\)), which consists of one oxygen and two hydrogen atoms. We expand each electron's wavefunctions using a basis set, which, in practice, means each atom is represented using a predefined basis set that collectively describes a set of electrons. Specifically, we use the STO-3G minimal basis set, where oxygen is described by five basis functions (\(1s, 2s, 2p_x, 2p_y, 2p_z\)), and each hydrogen atom has two basis functions (1s for each). Altogether, this totals seven basis functions.

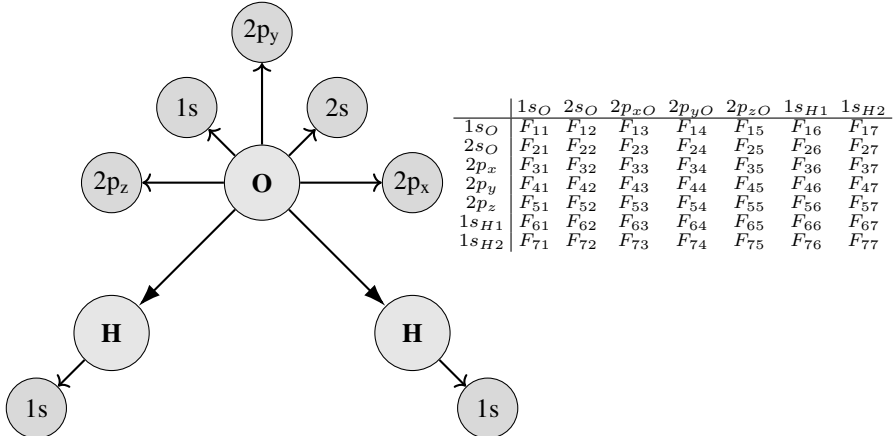
\begin{figure}[h]
    \centering
    \begin{tikzpicture}[
        basis/.style={circle, draw=black, fill=gray!30, minimum size=0.8cm, inner sep=0},
        atom/.style={circle, draw=black, fill=gray!20, minimum size=1.0cm, inner sep=0},
        legend/.style={rectangle, draw=none, fill=none, inner sep=0, font=\small}
    ]
        % Oxygen atom
        \node[atom] (O) at (0,0) {\textbf{O}};

        % Hydrogen atoms
        \node[atom] (H1) at (-2,-2) {\textbf{H}};
        \node[atom] (H2) at (2,-2) {\textbf{H}};

        % Basis functions on Oxygen
        \node[basis] (O1s) at (-1,1) {1s};
        \node[basis] (O2s) at (1,1) {2s};
        \node[basis] (O2px) at (2,0) {2p\textsubscript{x}};
        \node[basis] (O2py) at (0,2) {2p\textsubscript{y}};
        \node[basis] (O2pz) at (-2,0) {2p\textsubscript{z}};

        % Basis functions on Hydrogen
        \node[basis] (H1s1) at (-3,-3) {1s};
        \node[basis] (H1s2) at (3,-3) {1s};

        % Draw bonds
        \draw[-{Latex[length=3mm,width=2mm]}, thick] (O) -- (H1);
        \draw[-{Latex[length=3mm,width=2mm]}, thick] (O) -- (H2);

        % Draw basis function arrows
        \draw[->, thick] (O) -- (O1s);
        \draw[->, thick] (O) -- (O2s);
        \draw[->, thick] (O) -- (O2px);
        \draw[->, thick] (O) -- (O2py);
        \draw[->, thick] (O) -- (O2pz);
        \draw[->, thick] (H1) -- (H1s1);
        \draw[->, thick] (H2) -- (H1s2);

        % Legend for matrix (optional placement)
        \node[legend] (matrix) at (5.4,0.1) {
            \scriptsize

                \setlength{\arraycolsep}{2pt}
                \renewcommand{\arraystretch}{0.9}
                            
                $\begin{array}{c|ccccccc}
                    & 1s_O & 2s_O & 2p_{xO} & 2p_{yO} & 2p_{zO} & 1s_{H1} & 1s_{H2} \\
                    \hline
                    1s_O & F_{11} & F_{12} & F_{13} & F_{14} & F_{15} & F_{16} & F_{17} \\
                    2s_O & F_{21} & F_{22} & F_{23} & F_{24} & F_{25} & F_{26} & F_{27} \\
                    2p_x & F_{31} & F_{32} & F_{33} & F_{34} & F_{35} & F_{36} & F_{37} \\
                    2p_y & F_{41} & F_{42} & F_{43} & F_{44} & F_{45} & F_{46} & F_{47} \\
                    2p_z & F_{51} & F_{52} & F_{53} & F_{54} & F_{55} & F_{56} & F_{57} \\
                    1s_{H1} & F_{61} & F_{62} & F_{63} & F_{64} & F_{65} & F_{66} & F_{67} \\
                    1s_{H2} & F_{71} & F_{72} & F_{73} & F_{74} & F_{75} & F_{76} & F_{77}
                \end{array} $
        };
    \end{tikzpicture}
    \caption{Molecular structure and basis functions for water (\(\text{H}_2\text{O}\)).}
    \label{fig:water_structure}
\end{figure}

When considering the effect of rotations on the Hamiltonian matrix, it is important to understand how the matrix elements transform under such operations. 

For instance, if we consider a rotation that affects only the \(2p_x\), \(2p_y\), and \(2p_z\) orbitals of the oxygen atom, the corresponding submatrix of the Hamiltonian \( \mathbf{H}_{p} \) within the \(2p\) orbital space would transform as:

\[
\mathbf{H}'_p = \mathbf{R}_{p} \mathbf{H}_{p} \mathbf{R}_{p}^{\top}
\]

where \( \mathbf{H}_{p} \) is the submatrix containing the interactions between the \(2p_x\), \(2p_y\), and \(2p_z\) orbitals. The rotation matrix \( \mathbf{R}_{p} \) is specific to the rotation in the \(2p\) orbital space.

\begin{center}
    \[
        \setlength{\arraycolsep}{2pt}
        \renewcommand{\arraystretch}{0.9}
        \begin{array}{c|ccccccc}
            & 1s_O & 2s_O & 2p_x & 2p_y & 2p_z & 1s_{H1} & 1s_{H2} \\
            \hline
            1s_O & F_{11} & F_{12} & F_{13} & F_{14} & F_{15} & F_{16} & F_{17} \\
            2s_O & F_{21} & F_{22} & F_{23} & F_{24} & F_{25} & F_{26} & F_{27} \\
            2p_x & F_{31} & F_{32} & \color{red}{F_{33}} & \color{red}{F_{34}} & \color{red}{F_{35}} & F_{36} & F_{37} \\
            2p_y & F_{41} & F_{42} & \color{red}{F_{43}} & \color{red}{F_{44}} & \color{red}{F_{45}} & F_{46} & F_{47} \\
            2p_z & F_{51} & F_{52} & \color{red}{F_{53}} & \color{red}{F_{54}} & \color{red}{F_{55}} & F_{56} & F_{57} \\
            1s_{H1} & F_{61} & F_{62} & F_{63} & F_{64} & F_{65} & F_{66} & F_{67} \\
            1s_{H2} & F_{71} & F_{72} & F_{73} & F_{74} & F_{75} & F_{76} & F_{77}
        \end{array}
    \]
    \[
            \setlength{\arraycolsep}{2pt}
            \renewcommand{\arraystretch}{0.9}
            \begin{array}{c|ccc}
                & 2p_x & 2p_y & 2p_z \\
                \hline
                2p_x & \color{red}{F'_{33}} & \color{red}{F'_{34}} & \color{red}{F'_{35}} \\
                2p_y & \color{red}{F'_{43}} & \color{red}{F'_{44}} & \color{red}{F'_{45}} \\
                2p_z & \color{red}{F'_{53}} & \color{red}{F'_{54}} & \color{red}{F'_{55}}
            \end{array}
        \]

\end{center}

This transformation ensures that the physical properties described by the Hamiltonian remain consistent under rotational operations, a fundamental requirement for accurately modeling molecular systems.

\section{Dipole Moment and Electronic Spatial Extent}
\label{app:dipole}
 In this section, we describe how to compute the dipole moment and the electronic spatial extent using the Kohn-Sham orbitals derived from the Kohn-Sham Hamiltonian.

The Kohn-Sham equations are defined as

\begin{equation}
\left[ -\frac{\hbar^2}{2m} \nabla^2 + V_{\text{eff}}(\mathbf{r}) \right] \psi_i(\mathbf{r}) = \epsilon_i \psi_i(\mathbf{r}),
\end{equation}

where $\psi_i(\mathbf{r})$ denotes the Kohn-Sham orbitals, $\epsilon_i$ are the orbital energies, and $V_{\text{eff}}(\mathbf{r})$ is the effective potential. The electron density is then expressed as

\begin{equation}
\rho(\mathbf{r}) = \sum_{i}^{\text{occ}} |\psi_i(\mathbf{r})|^2.
\end{equation}

The dipole moment $\mathbf{d}$ is computed as the sum of electronic and nuclear contributions. The electronic contribution is given by

\begin{equation}
\mathbf{d}^{\text{elec}} = -e \int \mathbf{r} \rho(\mathbf{r}) \, d\mathbf{r},
\end{equation}

and the nuclear contribution is

\begin{equation}
\mathbf{d}^{\text{nuc}} = -e \sum_{A} Z_A \mathbf{R}_A,
\end{equation}

where $e$ represents the elementary charge, $Z_A$ denotes the atomic number of nucleus $A$, and $\mathbf{R}_A$ is the position vector of nucleus $A$. The total dipole moment is therefore the sum of these two components:

\begin{equation}
\mathbf{d} = \mathbf{d}^{\text{elec}} + \mathbf{d}^{\text{nuc}}.
\end{equation}

Next, we define the electronic spatial extent, denoted by $\langle r^2 \rangle$, which provides a measure of the spatial distribution of the electron density. It is calculated as

\begin{equation}
\langle r^2 \rangle = \int r^2 \rho(\mathbf{r}) \, d\mathbf{r}.
\end{equation}

To compute these quantities in practice, one often works in a basis set representation. The electron density matrix $P_{\mu\nu}$ is defined as

\begin{equation}
P_{\mu\nu} = 2 \sum_{i}^{\text{occ}} C_{i\mu} C_{i\nu},
\end{equation}

where $C_{i\mu}$ are the coefficients of the molecular orbitals in terms of the basis functions $\phi_\mu$. The dipole integrals $\mu_{\alpha}^{\mu\nu}$ for a Cartesian direction $\alpha$ are expressed as

\begin{equation}
\mu_{\alpha}^{\mu\nu} = \int \phi_\mu(\mathbf{r}) \, r_\alpha \, \phi_\nu(\mathbf{r}) \, d\mathbf{r}.
\end{equation}

The electronic contribution to the dipole moment in the basis set representation is then given by

\begin{equation}
d_{\alpha}^{\text{elec}} = -e \sum_{\mu\nu} P_{\mu\nu} \mu_{\alpha}^{\mu\nu}.
\end{equation}

Similarly, the electronic spatial extent in the basis set representation is computed using the integral

\begin{equation}
\langle r^2 \rangle^{\mu\nu} = \int \phi_\mu(\mathbf{r}) \, r^2 \, \phi_\nu(\mathbf{r}) \, d\mathbf{r},
\end{equation}

leading to the final expression

\begin{equation}
\langle r^2 \rangle = \sum_{\mu\nu} P_{\mu\nu} \langle r^2 \rangle^{\mu\nu}.
\end{equation}

\section{Group Theory}
\par{If a function is equivariant to the action of
a group, it does not matter whether the group acts
on the function’s input or output. More formally,
for vector spaces $V$ and $W$ equipped with
arbitrary group representations $D_V(g)$
and $D_W(g)$, a function $f: V \to W$ is
equivariant to $G$ if 
\[
    f(D_V(g)v) = D_W(g)f(v)
\] for all $g \in G$ and for all $v \in V$. In the
case of $\mathrm{SE}(3)$, a function $f$ is equivariant if
the output is the same whether the input is slid
and rotated, then put through $f$, or whether the
input is put through $f$, then slid and rotated in
the same way. A function $f$ is invariant to a
group $G$ if
\[
    D_W(g) = e,
\]the identity element in $W$, for all $g \in G$. In the case of $\mathrm{SE}(3)$, a function $f$ is
equivariant if it gives the same output no matter
how its input is slid and rotated. Equivariance is
a fundamental notion in the modeling of physical
systems. In the context of this paper, it is
desirable that the function that predicts the
Hamiltonian matrix be $\mathrm{SE}(3)$-equivariant,
reflecting that the molecule’s energetic
properties are equivalent if the molecule is
rotated or translated.}

\par{Group representations are an instance of the
more general notion of group homomorphisms. Given
a group $G$ with group operation $\circ$ and
another group $H$ with group operation $*$, a
group homomorphism is a map $\rho: G \to H$ such
that 
\[\rho(g_1 \circ g_2) = \rho(g_1) *
\rho(g_2).\] A group homomorphism, then, must
preserve the structure of the group, which means
that it does not matter whether the group
operation is performed in $G$ or $H$. Note,
however, that a homomorphism might respect the
group structure only trivially. For example,
$\rho: G \to H$, defined by $\rho(g) = e$, is a
trivial group homomorphism. A group representation
is simply a group homomorphism where $H = V$ is
some vector space. In plainer language, this means
that a group representation is a group written as
a set of matrices, whose group operation is matrix
multiplication. It is always the case that $V
\subseteq GL$, since any matrix that is
degenerate or non-square lacks an inverse and
therefore fails to satisfy the inverse axiom.}

\par{Two group representations $D$ and $D^\prime$
are equivalent if there is a fixed matrix $P$ such
that $D(g) = P^{-1}D^\prime(g)P$ for all $g \in
G$. In this case, $D$ and $D^\prime$ can be
interpreted as the same representation defined
with respect to different bases. A representation
$D$ is reducible if it acts on independent
subspaces of $V$; otherwise, it is irreducible.
An irreducible representation is called an
irrep. More formally, $D$ is reducible if
\[
    D(g) = P^{-1}
    \begin{bmatrix}
        D^{(\ell_0)}(g) & & \\
        & D^{(\ell_1)}(g) & \\
        & & \ddots
    \end{bmatrix}
    P = P^{-1}\Big(\bigoplus_iD^{(\ell_i)}(g)\Big) P,
\] which means that it $D$ is block-diagonal with respect to some basis.
It is convenient to decompose group
representations into irreducible representations
because this reduces the group operation
calculation to several smaller independent
calculations. Irreps are the atoms of group
representations in the sense that arbitrary
representations can be composed with the direct
sum of irreps. The irreps of $\mathrm{SO}(3)$ are called
Wigner D-matrices, with $D^{(\ell)}(g)$ denoting a
Wigner D-matrix representation of $g$ of degree
$\ell$. Wigner D-matrices are of size $(2\ell + 1)
\times (2\ell + 1)$, with higher-degree
representations allowing for more precise handling
of angular information.}

\section{Exensive Details on PubChemQH}
\label{app:dataset}
Here, we present a detailed comparison between the PubChemQH dataset and the curated QH9 dataset. This comparison aims to highlight the key differences and similarities. Additionally, we provide a comprehensive analysis of the atom number distribution within the PubChemQH dataset, supported by a Figure \ref{fig:node-dis}.
\begin{table}[ht]
    \centering
    \caption{Comparison of PubChemQH and QH9}
    \resizebox{0.8\textwidth}{!}{%
    \begin{tabular}{lcc}
        \toprule
        Feature & PubChemQH & QH9 \\
        \midrule
        Source & PubChem Database & QM9 \\
        Number of Molecules & 50,321 & 130,831 (QH9-stable) \\
        Functional & B3LYP & B3LYP \\
        Basis Set & Def2TZVP & Def2SVP \\
        SCF Convergence Tolerance & $10^{-8}$ & $10^{-13}$ \\
        SCF Gradient threshold: & $10^{-4}$ & $3.16 \times 10^{-4}$ \\
        Grid Density Level & 3 & 3\\
        Mean of Node Number & 61.85 & 18 \\
        Mean of Hamiltonian Size & 1025 & 141 \\
        \bottomrule
    \end{tabular}
    }

    \label{tab:comparison}
\end{table}

\begin{figure}
    \centering
    \includegraphics[width=0.7\linewidth]{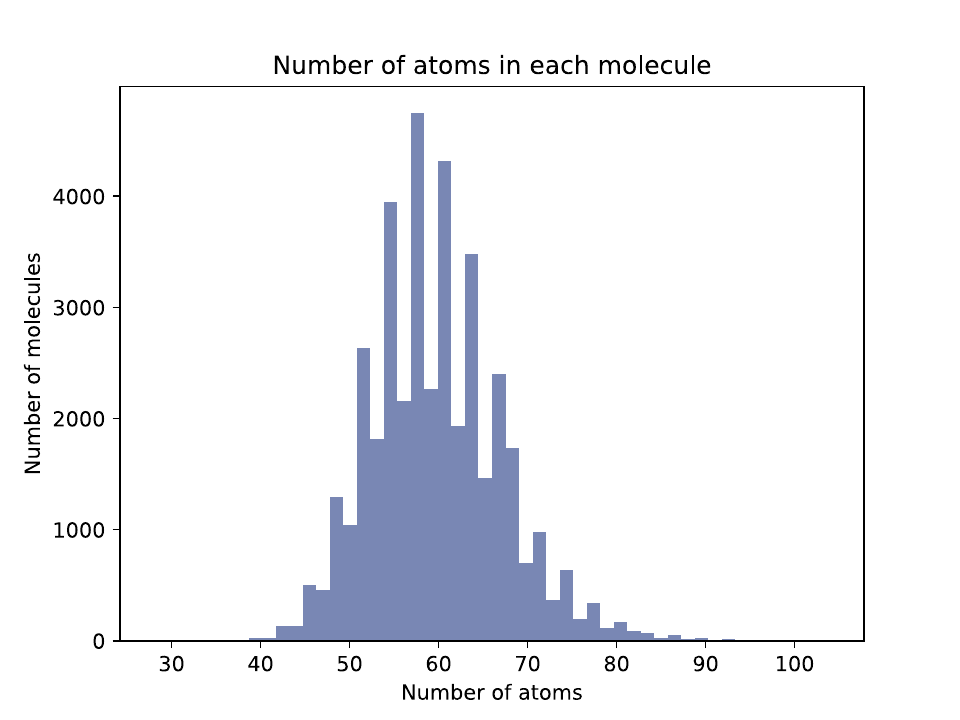}
    \caption{Node atom distribution for PubChemQH Dataset}
    \label{fig:node-dis}
\end{figure}

\section{Extensive Details on WANet}
\begin{figure}
    \centering
    \includegraphics[width=1\linewidth]{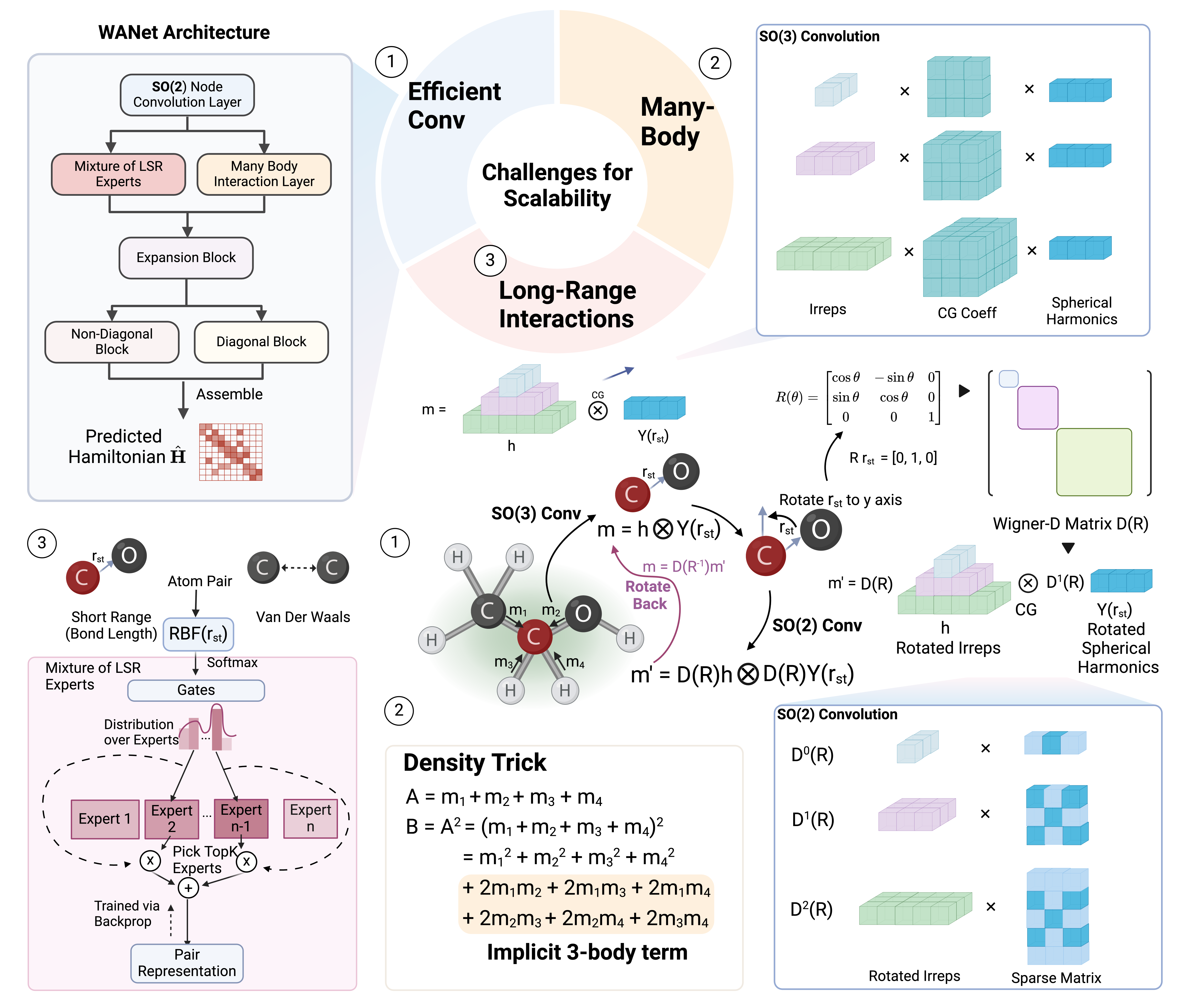}
    \caption{ \textbf{Scalable architecture for predicting the Hamiltonian matrix (\(\mathbf{\hat{H}}\))}.  
The framework addresses three key scalability challenges in quantum many-body systems: (1) efficient convolutions, (2) many-body interactions, and (3) long-range interactions. (1) Efficient convolutions are achieved through an SO(3)-equivariant convolution reduced to SO(2) for computational efficiency. (2) Many-body interactions are captured using a density trick, which implicitly includes three-body terms by combining pairwise interactions quadratically. (3) Long-range interactions are modeled with a mixture of Local-Structure Representation (LSR) experts, where a softmax-based gating mechanism allocates pairwise inputs to top-performing experts. Finally, the diagonal and non-diagonal components are assembled into the predicted Hamiltonian matrix, providing a scalable and accurate representation of the system dynamics.}
    \label{fig:wanet}
\end{figure}
We here provide extensive details on the motivations of the WANet architecture.The WANet architecture builds upon well-established equivariant neural network designs, which have been extensively studied and validated in the field. 

We provide a visual representation illustrating WANet's architecture and its key components, which is given in Figure\ref{fig:wanet} . This visual aid will help readers better understand how WANet enhances the scalability of Hamiltonian prediction for large molecular systems. Each architectural element was carefully chosen to address specific computational and physical challenges at scale:

First, the motivation behind WANet’s architecture stems from a critical limitation in existing Hamiltonian prediction methods—their inability to scale to larger molecular systems. When using larger basis sets, such as Def2-TZVP, higher-order irreducible representations are required to accurately capture angular dependencies in molecular orbitals. This poses a severe computational challenge for traditional SE(3)-equivariant methods, including QHNet, which become extremely expensive due to their computational complexity. WANet addresses this bottleneck by introducing SO(2) convolutions, which reduce computational complexity from $O(L^3)$  to $O(L^6)$ . This improvement enables WANet to process high tensor degrees efficiently.  Without $\mathrm{SO}(2)$ convolutions, handling the scale of our PubChemQH dataset and larger systems would be significantly more challenging and computationally expensive.

Second, large molecular systems exhibit fundamentally different physics at various distance scales. As molecular size increases, long-range interactions become more prominent. WANet's Mixture-of-Experts architecture is designed to model this complex physics efficiently. It employs specialized experts for different interaction ranges, capturing both short-range effects (like covalent bonding) and long-range phenomena (such as electrostatics). By sparsifying these experts, WANet achieves a rich representation of molecular interactions while maintaining computational efficiency, making it particularly well-suited for large-scale systems.

Third, WANet's architecture is designed to accurately capture the intrinsic properties of molecular systems, particularly for large molecules. The Hamiltonian matrix, which fully characterizes the quantum state and electron distribution, presents unique challenges in prediction due to complex electron correlation effects as the system size grows. To address this, WANet incorporates the MACE architecture's density trick, enabling efficient computation of many-body interactions without explicit calculation of all terms. This approach is crucial for maintaining accuracy as molecular size increases and electron correlation effects become more pronounced, ensuring WANet's scalability and precision in Hamiltonian prediction for large systems.

\textbf{Pair Construction Layer} \quad The objective of the pair construction layer is to extend the model's capacity to consider non-diagonal node pairs by introducing a tensor product filter. This filter modulates the projection of irreducible representations onto the space of node pair irreducible representations, denoted by \(f_{ts}\). It is important to note that, in contrast to the Node Convolution Layer, which performs graph convolution on a radius graph or KNN graph, the pair construction layer considers all possible node interactions by operating on a complete graph. The mathematical formulation of this layer is given as follows:

\begin{equation}
f_{ts}^{\ell_o} =  \sum_{l_i, l_j}W_{l_{i}, l_{j}, l_{o}}  \left( {x}_{s}^{l_{i}} \otimes {x}_{t}^{l_{j}} \right)^{l_{o}},
\end{equation}

where \( {x}_s^{l_i} \) and \( {x}_t^{l_j} \) are the \( l_i \)-th and \( l_j \)-th irreducible representations of source node \( s \) and target node \( t \), respectively, and \( W_{l_{i}, l_{j}, l_{o}} \) are the learned weights that couple these representations into the output representation \( \ell_o \).

To improve the efficiency of the tensor product, we employed the channel-grouped tensor product, where the channels of the first and second tensors are tied to collectively construct the output channel path\footnote{This is also called ``uuw'' tensor product in \texttt{e3nn} implementation.}. Additionally, to accommodate the symmetry inherent in the Hamiltonian matrix, we implement a symmetric structure in the pair representations. Specifically, the representations for node pairs \( (s, t) \) and \( (t, s) \) must be identical, reflecting the symmetrical nature of physical interactions. This is formulated as:

\[
{f_{ts}^{\ell_o}}' = {f_{st}^{\ell_o}}' = \frac{1}{2} (f_{ts}^{\ell_o} + f_{st}^{\ell_o}),
\]

where \( f_{ts}^{\ell_o} \) and \( f_{st}^{\ell_o} \) denote the initial, unsymmetrized tensor products for the node pairs \( s \) and \( t \) before the application of symmetry. The primed notations \( {f_{ts}^{\ell_o}}' \) and \( {f_{st}^{\ell_o}}' \) represent the final, symmetrized outputs.

\textbf{Graphical Illustration of MoE}\quad Here we provide an additional illustration of the Mixture of the Long-Short-Range Pair experts described in the maintext, which is shown in Figure~\ref{fig:MoE}.
\begin{figure}
    \centering
    \includegraphics[width=0.5\linewidth]{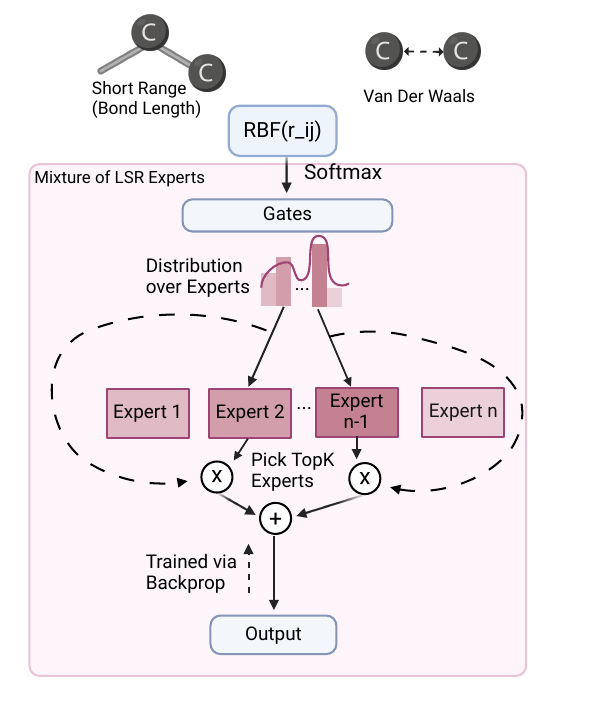}
    \caption{illustration of the Mixture of the Long-Short-Range Pair experts}
    \label{fig:MoE}
\end{figure}

\textbf{Expansion Block}\quad In the construction of final Hamiltonian blocks that encompass full orbital information using pair irreducible representations and many-body irreducible representation, a tensor expansion operation is employed alongside the filtering process. This expansion is  defined by the following relation:

\begin{equation}
\left( \overline{\otimes}_{\ell_o} f^{\ell_o} \right)_{(\ell_i,\ell_j)}^{(m_i,m_j)} = \sum_{m_o=-\ell_o}^{\ell_o} C_{(\ell_i,m_i),(\ell_j,m_j)}^{(\ell_o,m_o)} f_{m_o}^{\ell_o},
\end{equation}

where \( C \) denotes the Clebsch-Gordan coefficients, and \( \overline{\otimes} \) symbolizes the tensor expansion which is the converse operation of tensor product. \( u^{\ell_i} \otimes v^{\ell_j} \) can be expressed as a sum over tensor expansions:

\begin{equation}
u^{\ell_i} \otimes v^{\ell_j} = \sum_{\ell_3} W_{l_{i}, l_{j}, l_{o}}\overline{\otimes} f^{\ell_o},
\end{equation}

subject to the coupling constraints \( |\ell_i - \ell_j| \leq \ell_o \leq \ell_i + \ell_j \).

\textbf{Loss Weighting.} We determined the loss weights through a principled ablation study on a validation set. While keeping $\lambda_1$ and $\lambda_2$ fixed at 1, we varied $\lambda_3$ from 0.5 to 3 to understand the impact of WALoss weighting. We found the method to be robust across a reasonable range of values, with $\lambda_1=\lambda_2=1$, and $\lambda_3 = 2.5$ providing consistently strong performance. 
The table below summarizes the performance metrics for different values of $\lambda_1, \lambda_2$ and $\lambda_3$. 

\begin{table}[ht]
\centering
\caption{Performance metrics for with different $\lambda_1$, $\lambda_2$, and $\lambda_3$ values.}
\resizebox{\textwidth}{!}{
\begin{tabular}{ccc|cccccccc}
\hline
$\lambda_1$ & $\lambda_2$ & $\lambda_3$ & Hamiltonian MAE $\downarrow$ & $\epsilon_\mathrm{HOMO} \text{ MAE} \downarrow$ &  $\epsilon_\mathrm{LUMO} \text{ MAE}\downarrow$ &  $\epsilon_\Delta$ & $\epsilon_\mathrm{occ} \text{ MAE}\downarrow$ & $\epsilon_\mathrm{orb} \text{MAE}\downarrow$ & C$\uparrow$ & System Energy MAE $\downarrow$ \\
\hline
1    & 1    & 3    & 0.5258   & 0.8688   & 1.9116   & 2.0407  & 21.92  & 8.18  & 48.90\%   & 54.77           \\
1    & 1    & 2.5  & 0.4744   & 0.7122   & 0.73     & 1.327   & 18.84  & 7.33  & 48.03\%   & 47.193          \\
1    & 1    & 2    & 0.4918   & 0.7847   & 1.24981  & 1.46951 & 21.15  & 7.7   & 47.72\%   & 59.79           \\
1    & 1    & 1.5  & 0.4586  & 1.50825  & 3.29805  & 3.1123  & 23.57  & 8.56  & 47.29\%   & 64.857          \\
1    & 1    & 1    & 0.4807   & 1.1596   & 3.45767  & 3.2068  & 23.16  & 9.2   & 47.15\%   & 62.61719        \\
1    & 1    & 0.5  & 0.4223  & 2.28836  & 6.83532  & 5.69888 & 28.96  & 11.78 & 45.30\%   & 72.9            \\
\hline
\end{tabular}
}

\end{table}

\newpage
\section{Additional Theory}

\subsection{Proof of Theorem}
\label{app:weyl-davis}

% \begin{theorem}[Weyl and Davis-Kahan Perturbation Theorem \citep{golub2013matrix, franklin2012matrix}]
%     Let \(\mathbf{H}\) and \(\hat{\mathbf{H}}\) represent Hamiltonian matrices, and \(\mathbf{S}\) the overlap matrix. Define the perturbation matrix as \(\Delta \mathbf{H} := \hat{\mathbf{H}} - \mathbf{H}\). Let \(\lambda_i(\mathbf{H}, \mathbf{S})\) and \(\lambda_i(\hat{\mathbf{H}}, \mathbf{S})\) be the \(i^{\text{th}}\) generalized eigenvalues of \(\mathbf{H}\) and \(\hat{\mathbf{H}}\), respectively. Assume a spectral gap \(\delta\) separates the generalized eigenvalues of \(\mathbf{H}\) and \(\hat{\mathbf{H}}\). $\kappa(\cdot)$ denotes the condition number of a given matrix, and $\|\cdot\|_2$ represents the spectral norm. Then, the difference in eigenvalues and the angle \(\theta\) between the eigenspace of \(\mathbf{H}\) and \(\hat{\mathbf{H}}\) satisfy:
%     \[
%     \| \lambda_i(\hat{\mathbf{H}}, \mathbf{S}) - \lambda_i(\mathbf{H}, \mathbf{S}) \| \leq \frac{\kappa(\mathbf{S})}{\|\mathbf{S}\|_2} \cdot \|\Delta \mathbf{H}\|_\mathrm{F}, \quad \sin \theta \leq \frac{\kappa(\mathbf{S})}{\|\mathbf{S}\|_2} \cdot \frac{\|\Delta \mathbf{H}\|_\mathrm{F}}{\delta}.
%     \]
% \end{theorem}
\begin{theorem}
Let $\mathbf{H}, \hat{\mathbf{H}} \in \mathbb{R}^{n \times n}$ be symmetric matrices, and let $\mathbf{S} \in \mathbb{R}^{n \times n}$ be a symmetric positive definite matrix. Consider the generalized eigenvalue problems:
\[
\mathbf{H} \mathbf{C} = \mathbf{S} \mathbf{C} \boldsymbol{\epsilon}, \quad \hat{\mathbf{H}}\, \hat{\mathbf{C}} = \mathbf{S}\, \hat{\mathbf{C}}\, \hat{\boldsymbol{\epsilon}},
\]
where $\boldsymbol{\epsilon}$ and $\hat{\boldsymbol{\epsilon}}$ are diagonal matrices of eigenvalues, and $\mathbf{C}$ and $\hat{\mathbf{C}}$ are the corresponding eigenvector matrices. Define $\Delta \mathbf{H} = \hat{\mathbf{H}} - \mathbf{H}$, and let $\delta$ be the minimum distance between the eigenvalue of interest and the rest of the spectrum of $\mathbf{S}^{-1}\mathbf{H}$. Then, the following bounds hold:

\begin{enumerate}
    \item \textbf{Eigenvalue Differences:}
    \[
    \left|\lambda_i(\hat{\mathbf{H}}, \mathbf{S}) - \lambda_i(\mathbf{H}, \mathbf{S})\right| \leq \frac{\kappa(\mathbf{S})}{\|\mathbf{S}\|_2} \|\Delta \mathbf{H}\|_{\mathrm{F}} \leq \frac{\kappa(\mathbf{S})}{\|\mathbf{S}\|_2} \|\Delta \mathbf{H}\|_{1,1},
    \]
    where $\kappa(\mathbf{S}) = \|\mathbf{S}\|_2 \|\mathbf{S}^{-1}\|_2$ is the condition number of $\mathbf{S}$ with respect to the spectral norm, $\|\cdot\|_{\mathrm{F}}$ denotes the Frobenius norm, and $\|\cdot\|_{1,1}$ denotes the element-wise $l_1$ norm.

    \item \textbf{Eigenspace Angle:}
    \[
    \sin \theta \leq \frac{\kappa(\mathbf{S})}{\|\mathbf{S}\|_2} \cdot \frac{\|\Delta \mathbf{H}\|_{\mathrm{F}}}{\delta} \leq \frac{\kappa(\mathbf{S})}{\|\mathbf{S}\|_2} \cdot \frac{\|\Delta \mathbf{H}\|_{1,1}}{\delta},
    \]
    where $\theta$ is the angle between the eigenspaces corresponding to $\lambda_i(\mathbf{H}, \mathbf{S})$ and $\lambda_i(\hat{\mathbf{H}}, \mathbf{S})$.
\end{enumerate}
\end{theorem}

\begin{proof}
Consider the generalized eigenvalue problems for $\mathbf{H}$ and $\hat{\mathbf{H}}$ with respect to $\mathbf{S}$:
\[
\mathbf{H} \mathbf{C} = \mathbf{S} \mathbf{C} \boldsymbol{\epsilon}, \quad \hat{\mathbf{H}}\, \hat{\mathbf{C}} = \mathbf{S}\, \hat{\mathbf{C}}\, \hat{\boldsymbol{\epsilon}}.
\]
Since $\mathbf{S}$ is symmetric positive definite, it is invertible. We can transform the generalized eigenvalue problems into standard eigenvalue problems by multiplying both sides by $\mathbf{S}^{-1}$:
\[
\mathbf{A} = \mathbf{S}^{-1}\mathbf{H}, \quad \hat{\mathbf{A}} = \mathbf{S}^{-1}\hat{\mathbf{H}} = \mathbf{A} + \mathbf{E},
\]
where $\mathbf{E} = \mathbf{S}^{-1} \Delta \mathbf{H}$.

\medskip

\textbf{Weyl's Perturbation Theorem}

\begin{theorem}[Weyl's Perturbation Theorem]
Let $\mathbf{A}, \hat{\mathbf{A}} \in \mathbb{R}^{n \times n}$ be symmetric matrices with eigenvalues $\lambda_1 \leq \lambda_2 \leq \dots \leq \lambda_n$ and $\hat{\lambda}_1 \leq \hat{\lambda}_2 \leq \dots \leq \hat{\lambda}_n$, respectively. Then, for all $i = 1, \dots, n$,
\[
\left| \hat{\lambda}_i - \lambda_i \right| \leq \|\hat{\mathbf{A}} - \mathbf{A}\|_2.
\]
\end{theorem}

Applying Weyl's Theorem to $\mathbf{A}$ and $\hat{\mathbf{A}}$, we have:
\[
\left| \lambda_i(\hat{\mathbf{A}}) - \lambda_i(\mathbf{A}) \right| \leq \|\mathbf{E}\|_2 = \|\mathbf{S}^{-1} \Delta \mathbf{H}\|_2.
\]

\medskip

\textbf{Bounding the Eigenvalue Differences}

Using the sub-multiplicative property of the spectral norm:
\[
\|\mathbf{S}^{-1} \Delta \mathbf{H}\|_2 \leq \|\mathbf{S}^{-1}\|_2 \|\Delta \mathbf{H}\|_2.
\]
Since the spectral norm is bounded by the Frobenius norm:
\[
\|\Delta \mathbf{H}\|_2 \leq \|\Delta \mathbf{H}\|_{\mathrm{F}}.
\]
Combining these inequalities:
\[
\|\mathbf{S}^{-1} \Delta \mathbf{H}\|_2 \leq \|\mathbf{S}^{-1}\|_2 \|\Delta \mathbf{H}\|_{\mathrm{F}}.
\]
The condition number $\kappa(\mathbf{S})$ is defined as:
\[
\kappa(\mathbf{S}) = \|\mathbf{S}\|_2 \|\mathbf{S}^{-1}\|_2 \implies \|\mathbf{S}^{-1}\|_2 = \frac{\kappa(\mathbf{S})}{\|\mathbf{S}\|_2}.
\]
Substituting back:
\[
\|\mathbf{S}^{-1} \Delta \mathbf{H}\|_2 \leq \frac{\kappa(\mathbf{S})}{\|\mathbf{S}\|_2} \|\Delta \mathbf{H}\|_{\mathrm{F}}.
\]
Therefore:
\[
\left| \lambda_i(\hat{\mathbf{A}}) - \lambda_i(\mathbf{A}) \right| \leq \frac{\kappa(\mathbf{S})}{\|\mathbf{S}\|_2} \|\Delta \mathbf{H}\|_{\mathrm{F}}.
\]
Since the eigenvalues of $\mathbf{A}$ and $\hat{\mathbf{A}}$ correspond to the generalized eigenvalues of $(\mathbf{H}, \mathbf{S})$ and $(\hat{\mathbf{H}}, \mathbf{S})$, respectively, we have:
\[
\left| \lambda_i(\hat{\mathbf{H}}, \mathbf{S}) - \lambda_i(\mathbf{H}, \mathbf{S}) \right| \leq \frac{\kappa(\mathbf{S})}{\|\mathbf{S}\|_2} \|\Delta \mathbf{H}\|_{\mathrm{F}}.
\]
Since $\|\Delta \mathbf{H}\|_{\mathrm{F}} \leq \|\Delta \mathbf{H}\|_{1,1}$, we can further bound:
\[
\left| \lambda_i(\hat{\mathbf{H}}, \mathbf{S}) - \lambda_i(\mathbf{H}, \mathbf{S}) \right| \leq \frac{\kappa(\mathbf{S})}{\|\mathbf{S}\|_2} \|\Delta \mathbf{H}\|_{1,1}.
\]

\medskip

\textbf{Davis-Kahan $\sin \theta$ Theorem}

\begin{theorem}[Davis-Kahan $\sin \theta$ Theorem]
Let $\mathbf{A}, \hat{\mathbf{A}} \in \mathbb{R}^{n \times n}$ be symmetric matrices, and let $\mathcal{U}$ and $\hat{\mathcal{U}}$ be the invariant subspaces of $\mathbf{A}$ and $\hat{\mathbf{A}}$ corresponding to eigenvalues in intervals $\mathcal{I}$ and $\hat{\mathcal{I}}$, respectively. If $\delta = \mathrm{dist}(\mathcal{I}, \hat{\mathcal{I}}^c) > 0$, then
\[
\|\sin \Theta\|_2 \leq \frac{\|\hat{\mathbf{A}} - \mathbf{A}\|_2}{\delta},
\]
where $\Theta$ is the matrix of principal angles between $\mathcal{U}$ and $\hat{\mathcal{U}}$.
\end{theorem}

\medskip

\textbf{Bounding the Eigenspace Angle}

Applying the Davis-Kahan Theorem to our case:
\[
\sin \theta \leq \frac{\|\mathbf{E}\|_2}{\delta} = \frac{\|\mathbf{S}^{-1} \Delta \mathbf{H}\|_2}{\delta}.
\]
Using the previously established bound:
\[
\|\mathbf{S}^{-1} \Delta \mathbf{H}\|_2 \leq \frac{\kappa(\mathbf{S})}{\|\mathbf{S}\|_2} \|\Delta \mathbf{H}\|_{\mathrm{F}},
\]
we obtain:
\[
\sin \theta \leq \frac{\kappa(\mathbf{S})}{\|\mathbf{S}\|_2} \cdot \frac{\|\Delta \mathbf{H}\|_{\mathrm{F}}}{\delta}.
\]
Similarly, since $\|\Delta \mathbf{H}\|_{\mathrm{F}} \leq \|\Delta \mathbf{H}\|_{1,1}$:
\[
\sin \theta \leq \frac{\kappa(\mathbf{S})}{\|\mathbf{S}\|_2} \cdot \frac{\|\Delta \mathbf{H}\|_{1,1}}{\delta}.
\]
This completes the proof.
\end{proof}

% \begin{definition}[Crawford Number]
%     The concept of the Crawford number was introduced by Stewart in 1979 as part of a perturbation theory for symmetric pencils \(\mathbf{H} - \lambda \mathbf{S}\). This theory considers pencils that satisfy the condition:
%     \[
%     c(\mathbf{H}, \mathbf{S}) = \min_{\|x\|_2=1} \left( (x^{\top} \mathbf{H}x)^2 + (x^{\top} \mathbf{S}x)^2 \right) > 0.
%     \]
%     Here, \(c(\mathbf{H}, \mathbf{S})\) denotes the Crawford number of the pencil \(\mathbf{H} - \lambda \mathbf{S}\).
% \end{definition}

% \begin{theorem}
%     Consider an \(n\)-by-\(n\) symmetric-definite pencil \(\mathbf{H} - \lambda \mathbf{S}\) with eigenvalues \(\lambda_1 \geq \lambda_2 \geq \cdots \geq \lambda_n\). Let \(E_A\) and \(E_B\) be symmetric \(n\)-by-\(n\) matrices fulfilling:
%     \[
%     \epsilon^2 = \|E_A\|^2 + \|E_B\|^2 < c(\mathbf{H}, \mathbf{S}).
%     \]
%     Then, the perturbed pencil \((\mathbf{H} + E_A) - \lambda(\mathbf{S} + E_B)\) remains symmetric-definite with eigenvalues \(\mu_1 \geq \cdots \geq \mu_n\), and for each \(i = 1\) to \(n\), the eigenvalues satisfy:
%     \[
%     \left|\arctan(\lambda_i) - \arctan(\mu_i)\right| \leq \arctan\left(\frac{\epsilon}{c(\mathbf{H}, \mathbf{S})}\right).
%     \]
% \end{theorem}

% \begin{proof}
% Consult~\citep{stewart1979pertubation}.
% \end{proof}

\subsection{Perturbation Analysis and Growth of Eigenvalues}
\label{app:Eigen-sensitivity}
\begin{theorem}[Generalized Bai-Yin's law]
Let \(A\) be an \(n \times n\) random matrix with independent entries having mean \(\mu\) and variance \(\sigma^2\). Then, with high probability, the spectral norm of \(A\) is bounded by:
\[
\|A\|_2 \leq |\mu| n + 2\sigma \sqrt{n}.
\]
\end{theorem}

\begin{proof}
Consider the matrix \(A\) decomposed into its mean and fluctuation components:
\[
A = \mu J + B,
\]
where \(J\) is the matrix of all ones and \(B\) is a random matrix with zero-mean entries and variance \(\sigma^2\).

First, we bound the spectral norm of the mean component \( \mu J \). The matrix \(J\) is a rank-1 matrix with all entries equal to 1. Its largest singular value is \(n\), so:
\[
\|\mu J\|_2 = |\mu| \cdot n.
\]

Next, we bound the spectral norm of the fluctuation component \(B\). Since \(B\) has i.i.d. entries with zero mean and variance \(\sigma^2\), using Bai-Yin's law, we get:
\[
\|B\|_2 \leq 2\sigma \sqrt{n},
\]
with high probability.

Combining these bounds using the triangle inequality, we have:
\[
\|A\|_2 \leq \|\mu J\|_2 + \|B\|_2 \leq |\mu| n + 2\sigma \sqrt{n}.
\]

Thus, with high probability, the spectral norm of \(A\) is bounded by:
\[
\|A\|_2 \leq |\mu| n + 2\sigma \sqrt{n}.
\]
\end{proof}
\begin{theorem}[Perturbation Sensitivity Scaling]
\label{thm:perturbation}
Consider the generalized eigenvalue problem \(\mathbf{H} \mathbf{C} = \mathbf{S} \mathbf{C} \boldsymbol{\epsilon}\), where \(\mathbf{H} \in \mathbb{R}^{B \times B}\) is a symmetric Hamiltonian matrix, \(\mathbf{S} \in \mathbb{R}^{B \times B}\) is a positive definite overlap matrix, \(\mathbf{C} \in \mathbb{R}^{B \times k}\) is the eigenvector matrix, and \(\boldsymbol{\epsilon} \in \mathbb{R}^{k \times k}\) is a diagonal matrix of eigenvalues \(\lambda_i(\mathbf{H}, \mathbf{S})\). Let \(\hat{\mathbf{H}} = \mathbf{H} + \Delta \mathbf{H}\) be the perturbed Hamiltonian, with \(\Delta \mathbf{H} \in \mathbb{R}^{B \times B}\) a perturbation matrix whose entries have mean \(\mu\) and variance \(\sigma^2\). Suppose the smallest eigenvalue of \(\mathbf{S}\) satisfies:
\[
\lambda_{\min}(\mathbf{S}) = c + \frac{A}{1 + \left( \frac{B}{N_0} \right)^\alpha},
\]
where \(c > 0\), \(A > 0\), \(\alpha > 0\), and \(N_0 > 0\) are fixed constants. Then, for each eigenvalue, there exists a constant \(K > 0\) such that:
\[
\left| \lambda_i(\hat{\mathbf{H}}, \mathbf{S}) - \lambda_i(\mathbf{H}, \mathbf{S}) \right| \leq K \left( \sigma B^{1/2} + B |\mu| \right),
\]
for all sufficiently large \(B\).
\end{theorem}

\begin{proof}
Consider the unperturbed generalized eigenvalue problem \(\mathbf{H} \mathbf{c}_i = \lambda_i \mathbf{S} \mathbf{c}_i\), where \(\mathbf{H}\) and \(\mathbf{S}\) are symmetric, \(\mathbf{S}\) is positive definite, and \(\mathbf{c}_i^T \mathbf{S} \mathbf{c}_i = 1\). The perturbed system is \(\hat{\mathbf{H}} \hat{\mathbf{c}}_i = \hat{\lambda}_i \mathbf{S} \hat{\mathbf{c}}_i\), with \(\hat{\mathbf{H}} = \mathbf{H} + \Delta \mathbf{H}\), \(\hat{\mathbf{c}}_i = \mathbf{c}_i + \delta \mathbf{c}_i\), and \(\hat{\lambda}_i = \lambda_i + \delta \lambda_i\). Substituting and expanding yields:
\[
(\mathbf{H} + \Delta \mathbf{H}) (\mathbf{c}_i + \delta \mathbf{c}_i) = (\lambda_i + \delta \lambda_i) \mathbf{S} (\mathbf{c}_i + \delta \mathbf{c}_i).
\]
Subtracting the unperturbed equation and projecting onto \(\mathbf{c}_i^T\), noting \(\mathbf{c}_i^T (\mathbf{H} - \lambda_i \mathbf{S}) = 0\), we obtain:
\[
\mathbf{c}_i^T \Delta \mathbf{H} \mathbf{c}_i + \mathbf{c}_i^T \Delta \mathbf{H} \delta \mathbf{c}_i = \delta \lambda_i + \delta \lambda_i \mathbf{c}_i^T \mathbf{S} \delta \mathbf{c}_i.
\]
Solving for the perturbation:
\[
\delta \lambda_i = \frac{\mathbf{c}_i^T \Delta \mathbf{H} \mathbf{c}_i + \mathbf{c}_i^T \Delta \mathbf{H} \delta \mathbf{c}_i}{1 + \mathbf{c}_i^T \mathbf{S} \delta \mathbf{c}_i}.
\]
For small \(\Delta \mathbf{H}\), \(\delta \mathbf{c}_i\) is \(O(\|\Delta \mathbf{H}\|_2)\), so \(\mathbf{c}_i^T \Delta \mathbf{H} \delta \mathbf{c}_i\) and \(\mathbf{c}_i^T \mathbf{S} \delta \mathbf{c}_i\) are higher-order terms.  Here, we only consider the first-order term arises from evaluating the perturbation with the unperturbed eigenvector. Higher-order terms account for eigenvector shifts. We claim that analyzing the higher-order terms gives rise to unnecessary complication and does not alter the nature of catastrophic scaling. Specifically:
\[
\delta \lambda_i \approx \mathbf{c}_i^T \Delta \mathbf{H} \mathbf{c}_i.
\]
Bounding this term:
\[
|\delta \lambda_i| \leq |\mathbf{c}_i^T \Delta \mathbf{H} \mathbf{c}_i| \leq \|\mathbf{c}_i\|_2^2 \|\Delta \mathbf{H}\|_2.
\]
Since \(\mathbf{c}_i^T \mathbf{S} \mathbf{c}_i = 1\) and \(\mathbf{S} \succeq \lambda_{\min}(\mathbf{S}) \mathbf{I}\), we have:
\[
\|\mathbf{c}_i\|_2^2 \leq \frac{1}{\lambda_{\min}(\mathbf{S})} = \frac{1}{c + \frac{A}{1 + \left( \frac{B}{N_0} \right)^\alpha}} = \frac{1 + \left( \frac{B}{N_0} \right)^\alpha}{c \left[ 1 + \left( \frac{B}{N_0} \right)^\alpha \right] + A}.
\]
For \(\Delta \mathbf{H}\) as a random matrix, apply a spectral norm bound (by Bai-Yin law):
\[
\|\Delta \mathbf{H}\|_2 \leq C_1 (\sigma B^{1/2} + B |\mu|),
\]
with constant \(C_1\), holding with high probability. Combining:
\[
|\delta \lambda_i| \leq \frac{1 + \left( \frac{B}{N_0} \right)^\alpha}{c \left[ 1 + \left( \frac{B}{N_0} \right)^\alpha \right] + A} C_1 (\sigma B^{1/2} + B |\mu|).
\]
Thus, the leading-order bound is:
\[
\left| \lambda_i(\hat{\mathbf{H}}, \mathbf{S}) - \lambda_i(\mathbf{H}, \mathbf{S}) \right| = O\left( \frac{1 + \left( \frac{B}{N_0} \right)^\alpha}{c \left[ 1 + \left( \frac{B}{N_0} \right)^\alpha \right] + A} \left( \sigma B^{1/2} + B |\mu| \right) \right) = O \left( \sigma B^{1/2} + B |\mu| \right) .
\]
\end{proof}

\begin{remark}i y
If higher-order terms like \(\mathbf{c}_i^T \Delta \mathbf{H} \delta \mathbf{c}_i\) and the denominator correction are included, the perturbation includes an additional term of order \(O\left( \left( \frac{1 + \left( \frac{B}{N_0} \right)^\alpha}{c \left[ 1 + \left( \frac{B}{N_0} \right)^\alpha \right] + A} \right)^2 (\sigma B^{1/2} + B |\mu|)^2 \right)\), reflecting the quadratic contribution of the perturbation strength.
\end{remark}

\subsection{Scaling Law of the Smallest Eigenvalue of Overlap Matrix}
\label{app:scale-eigen}
\begin{figure}
    \centering
    \includegraphics[width=0.7\linewidth]{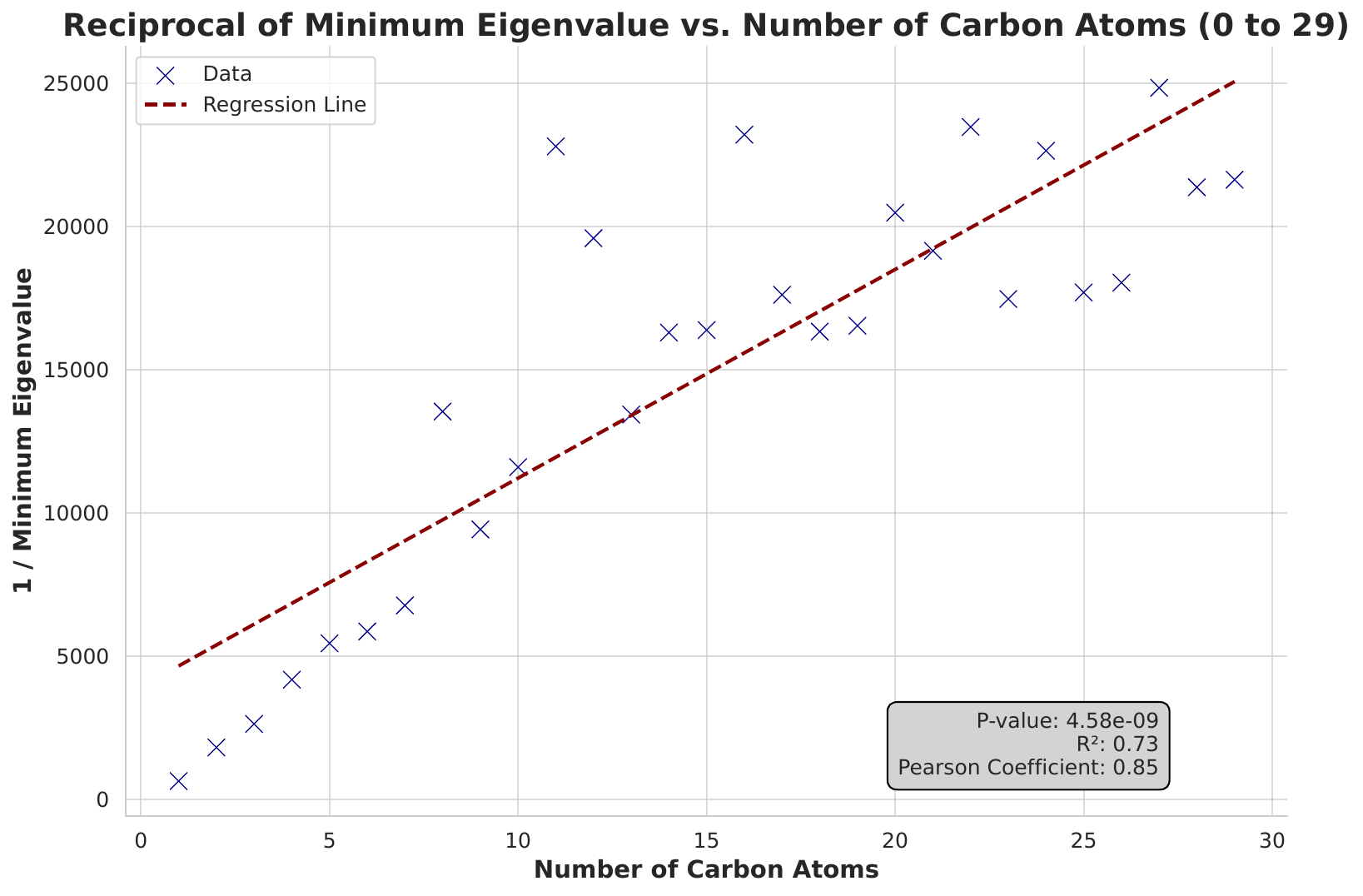}
    \caption{The scaling of the smallest eigenvalue of the overlap matrix in satured hydrocarbons.}
    \label{fig:eigen-carbon}
\end{figure}

Figure \ref{fig:eigen-carbon} shows the relationship between the number of carbon atoms in saturated hydrocarbons and the reciprocal of the smallest eigenvalue of the overlap matrix. The x-axis represents the number of carbon atoms, ranging from 0 to 29. The y-axis represents the reciprocal of the smallest eigenvalue, ranging up to 25,000. Specifically, the regression analysis yields a p-value of \( 4.58 \times 10^{-9} \), an \( R^2 \) value of 0.73, and a Pearson coefficient of 0.85, suggesting a strong linear relationship.

We also compare another scaling law suggested by \citet{spctrum}. As shown in Figure \ref{fig:eigen-carbon2}, we investigated the relationship between \( \frac{1}{\lambda_{\text{min}}} \) and the number of basis functions for various molecules. Our results indicate a clear power-law scaling, which can be expressed as \( y = a \cdot x^b \). The fitted parameters for each molecule are as follows: Eicosane (C20H42): \( y = 0.000107 \cdot x^{2.43} \), Graphene (C24): \( y = 0.00158 \cdot x^{2.83} \), and Diamond (C10): \( y = 0.0178 \cdot x^{4.16} \). This power-law behavior emphasizes the significant impact of the number of basis functions on \( \frac{1}{\lambda_{\text{min}}} \), with different scaling exponents for each molecule.

\subsubsection{\textit{k}-Nearest-Neighbor Overlaps}
Here, we analyzed a system where each basis function overlaps with its \( k \)-nearest neighbors. The overlap matrix \( S \succ 0 \) in this case is a symmetric banded matrix with bandwidth \( 2k + 1 \):

\[
S_{ij} = 
\begin{cases} 
1, & \text{if } i = j, \\
s_{|i-j|}, & \text{if } 0 < |i - j| \leq k, \\
0, & \text{if } |i - j| > k.
\end{cases}
\]

To derive a scaling equation for the smallest eigenvalue, consider the case where \( k = 1 \), which corresponds to the nearest-neighbor model. Other cases can be derived similarly using plane waves and discrete Fourier Transform. The nearest-neighbor model holds physical significance in the context of tight-binding models. For example, in chain-like polyenes, the Huckel model is based on a nearest-neighbor analysis, providing a meaningful framework for discussing this case. The overlap matrix \( S \) is given as:

\[
S =
\begin{bmatrix}
1 & s & 0 & \cdots & 0 \\
s & 1 & s & \ddots & \vdots \\
0 & s & 1 & \ddots & 0 \\
\vdots & \ddots & \ddots & \ddots & s \\
0 & \cdots & 0 & s & 1
\end{bmatrix}
\]

The eigenvalues \( \lambda_n \) of this tridiagonal matrix can be derived using the known solutions for such matrices:

\[
\lambda_n = 1 + 2s \cos \left( \frac{n\pi}{B+1} \right), \quad n = 1, 2, \dots, B.
\]

The smallest eigenvalue corresponds to \( n = B \):

\[
\lambda_{\min} = 1 + 2s \cos \left( \frac{B\pi}{B+1} \right).
\]

Using the trigonometric identity \( \cos(\pi - \theta) = -\cos(\theta) \) and approximating for large \( B \):

\[
\cos \left( \frac{B\pi}{B+1} \right) = -\cos \left( \frac{\pi}{B+1} \right) \approx -\left(1 - \frac{1}{2} \left( \frac{\pi}{B+1} \right)^2 \right).
\]

Substituting back:

\[
\lambda_{\min} \approx 1 - 2s + s \left( \frac{\pi}{B} \right)^2.
\]

Since \( S \succ 0 \), this yields \( 1 - 2s \geq 0 \). Here, we investigated the smallest eigenvalue \( \lambda_{\min} \) as \( B \to \infty \), identifying two distinct cases:

\begin{enumerate}
    \item \textbf{Saturating Case:} When \( s < 0.5 \) the smallest eigenvalue saturates at a finite value \( 1 - 2s \) as \( B \to \infty \).
    \item \textbf{Non-Saturating Case:} For \( s = 0.5 \), \( \lambda_{\min} \) decreases indefinitely with increasing \( B \).
\end{enumerate}

\begin{figure}
    \centering
    \includegraphics[width=1.0\linewidth]{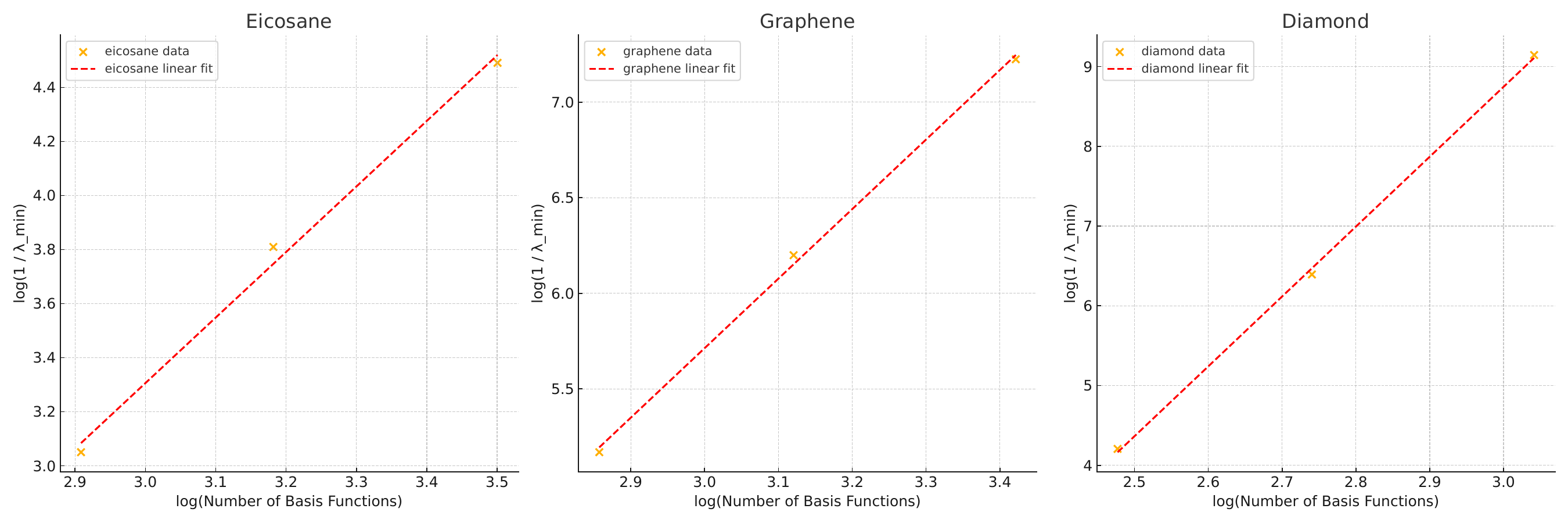}
    \caption{The scaling of the smallest eigenvalue of the overlap matrix in three systems. Data adapted from \cite{spctrum}.}
    \label{fig:eigen-carbon2}
\end{figure}

\section{Training}

\textbf{Predicting the Gap from Initial Guess} \quad In contrast to prior research, our study tackles the significant challenges posed by the scale of the PubChemQC t-zvp dataset. It is very challenging for the model to develop a good numeric starting point. To address this, we have shifted our focus towards a more tractable objective: making predictions based on an initial guess, which is easy to obtain. This adjustment is formalized in our objective function as follows:

\begin{equation}
 \theta^\star = \arg\min_{\theta} \frac{1}{|\mathcal{D}|} \sum_{(\mathcal{M}, \mathbf{H}^\star_\mathcal{M}) \in \mathcal{D}}  \mathrm{dist}\left(\hat{\mathbf{H}}_\theta (\mathcal{M}) , \mathbf{H}^\star_\mathcal{M} -  \mathbf{H}^{(0)}_\mathcal{M} \right),
\end{equation}
where \( |\mathcal{D}| \) denotes the cardinality of dataset \( \mathcal{D} \), and \(\mathrm{dist}(\cdot, \cdot)\) is a predefined distance metric, \( \mathbf{H}^{(0)}_\mathcal{M}\) is the initial guess of the Hamiltonian. We conducted an ablation study, showing that although predictions based on the initial guess significantly improve performance in Hamiltonian prediction, they struggle with predicting physical properties.

% \textbf{Encoding Inital Guess into the Model}. Since the quality of the initial guess can vary across different molecules, we aim to make the model aware of these differences. To achieve this, we explicitly disassemble the initial Hamiltonian matrix into self-level and pair-level irreducible representation. This representation allows the model to capture the unique characteristics of each molecule's initial guess and adapt its predictions accordingly. The detailed process of disassembling the initial Hamiltonian matrix is described in the Appendix. By incorporating this molecule-specific information, our model can provide more accurate and tailored predictions for each molecule in the dataset.

\textbf{Total Loss Function} \quad The total loss function used to train our model is a combination of the orbital alignment loss and the mean squared error (MSE) loss between the predicted and true Hamiltonian matrices:

\begin{equation}
\mathcal{L}_{total} = \lambda_1 \left( \frac{1}{n} \sum^n_{i = 1} \| \hat{\mathbf{H}}_i + \mathbf{H}_i^{(0)} - \mathbf{H}_i^\star \|_\mathrm{F}^2 \right) + \lambda_2 \left( \frac{1}{n} \sum^n_{i = 1} \| \hat{\mathbf{H}}_i + \mathbf{H}_i^{(0)}- \mathbf{H}_i^\star \|_\mathrm{F}^1 \right) + \lambda_3 \mathcal{L}_{\mathrm{align}} 
\end{equation}

where $\lambda_1, \lambda_2, \lambda_3$ are hyperparameters that control the relative importance of each  loss component.

\textbf{Experimental Setting} \quad The experimental settings are presented in Table \ref{tab:experimental_settings}. For the PubChemQH dataset, we used an 80/10/10 train/validation/test split, resulting in 40,257 training molecules, 5,032 validation molecules, and 5,032 test molecules. We trained all models of maximum 300,000 steps with a batch size of 8, using early stopping with a patience of 1,000 steps. WANet converges at 278,391 steps, QHNet at 258,267 steps, and PhisNet at 123,170 steps. All models used the Adam optimizer with a learning rate of 0.001 for PubChemQH, along with a polynomial learning rate scheduler with 1,000 warmup steps. We used gradient clipping of 1, and a radius cutoff of 5 \r{A}. For QHNet and PhisNet, we referred to the official implementation for these two models.

%For both datasets, we employed a polynomial learning rate scheduler with a warmup phase consisting of 1,000 steps. The total number of training steps was set to 300,000. We used a maximum radius of 5. The learning rates and batch sizes were adjusted for each dataset: a learning rate of 0.001 and a batch size of 8 for PubChemQH, and a learning rate of 5e-4 and a batch size of 32 for QH9, following the settings used in \citep{yu2024qh9}.  For the model without WALoss, we set the remove initial guess to be False.
\begin{table}[h]
    \centering
    \caption{Hyperparameter settings for the experimental study using the PubChemQH and QH9 datasets.}
    \vspace{2mm}
    \resizebox{\textwidth}{!}{
    \begin{tabular}{lcccccl}
        \hline
        Hyperparameter & PubChemQH (WANet) & QH9 (WANet) & PubChemQH (QHNet) & QH9 (QHNet) & PubChemQH (PhiSNet) & Description \\
        \hline
        Learning Rate & 0.001 & 5e-4 & 0.001 & 5e-4 & 0.001 &\\
        Batch Size & 8 & 32 & 8 & 32 & 8 &  \\
        Scheduler & Polynomial & Polynomial & Polynomial& Polynomial& Polynomial& \\
        LR Warmup Steps & 1,000 & 1,000 & 1,000 & 1,000 & 1,000 & Number of steps to linearly increase the learning rate \\
        Max Steps & 300,000 & 300,000 & 300,000 & 300,000 & 300,000 & Maximum number of training steps \\
        Model Order & 6 & 4 & 6 & 4 & 6 & Maximum degree of the spherical harmonics \\
        Embedding Dimension & 128 & 128 & 128 & 128 & 128 & Dimension of the node embedding \\
        Bottle Hidden Size & N/A & N/A & 32 & 32 & N/A & Size of the hidden layer in the bottleneck \\
        Number of GNN Layers & 5 & 5 & 5 & 5 & 5 & Number of graph neural network layers \\
        Max Radius & 5 & 15 & 5 & 15 & 5 & Maximum distance between nieghboring atoms \\
        % Number of Layers & 12 & 12 & Number of layers in the model \\
        Sphere Channels & 128 & 128 & 128 & 128 & 128 & Number of channels in the spherical harmonics \\
        % FFN Hidden Channels & 512 & 512 & Number of hidden channels in the feed-forward network \\
        % Number of Sphere Samples & 128 & 128 &  Number of samples in the spherical harmonics \\
        Edge Channels & 128 & 128 & 32 & 32 & 32 & Number of channels for edge features \\
        % Number of Distance Basis & 512 & 512 & Number of basis functions for distance encoding \\
        Drop Path Rate & 0.1 & 0.1 & N/A & N/A & N/A & Probability of dropping a path in the network \\
        Projection Drop & 0.0 & 0.0 & N/A & N/A & N/A & Dropout rate for the projection layer \\
        \hline
    \end{tabular}}
    \label{tab:experimental_settings}
\end{table}

The training details for the HOMO, LUMO, and GAP predictions are presented in Table~\ref{tab:unimol_settings}. 
For a fair comparison, all regression models used identical dataset splits (40,257 training molecules, 5,032 validation molecules, and 5,032 test molecules). We used batch size of 32 and consistent optimizer settings across all models: Adam optimizer with a learning rate of 0.001 and a polynomial learning rate scheduler with 1,000 warmup steps. WANet followed the same training setup described in Table~\ref{tab:experimental_settings}. The complete regression model hyperparameters are detailed in Table~\ref{tab:unimol_settings}. For Equiformer V2, Uni-Mol+, and Uni-Mol2 models, we used their original implementations.
  %For a fair comparison, all models were trained on the same dataset splits (i.e. 40,257 training molecules, 5,032 validation molecules, and 5,032 test molecules) and used the same optimizer settings (i.e. learning rate = 0.001 Adam, polynomial learning rate scheduler with 1,000 warmup steps).  In general, model hyperparameters were derived from the corresponding publications. Equiformer V2, Uni-Mol+ and Uni-Mol2 models code was based on the original implementation.

\begin{table}[h]
    \centering
    \caption{The training details for the HOMO, LUMO, and GAP predictions.}
    \vspace{2mm}
    \resizebox{\textwidth}{!}{
    \begin{tabular}{lcccl}
        \hline
        Hyperparameter & Uni-Mol+~\citep{lu2023highly} & Uni-Mol2~\citep{ji2024uni} & Equiformer v2 & Description \\
        \hline
        Learning Rate & 0.001 & 0.001 & 0.001 &  \\
        Batch Size & 32 & 32 & 32 &  \\
        Scheduler & Polynomial & Polynomial & Polynomial & \\
        LR Warmup Steps & 1,000 & 1,000 & 1,000  & Number of steps to linearly increase the learning rate \\
        Max Steps & 300,000 & 300,000 & 300,000 & Maximum number of training steps \\
        Embedding Dimension & 768 & 768 & 128 & Dimension of the input embedding \\
        % Pair Dim & 64 & 64 & 64 & 64 & Dimension of the pairwise feature embedding \\
        % Pair Hidden Dim & 32 & 32 & 32 & 32 & Hidden dimension for the pairwise features \\
        FFN Embedding Dim & 3072 & 3072 & 512 & Embedding dimension for the feed-forward network \\
        Number of Encoder Layers & 6 & 6 & 5 & Number of layers in the transformer encoder \\
        % Attention Heads & 8 & 8 & 8 & 8 & Number of attention heads \\
        % Dropout & 0.1 & 0.1 & 0.1 & 0.1 & dropout rate \\
        % Attention Dropout & 0.1 & 0.1 & 0.1 & 0.1 & Dropout rate for the attention mechanism \\
        % Activation Dropout & 0.1 & 0.1 & 0.1 & 0.1 & Dropout rate after the activation function \\
        Drop Path Rate & 0.0 & 0.0 & 0.1 & Probability of dropping a path in the network \\
        % Pair Dropout & 0.25 & 0.25 & 0.25 & 0.25 & Dropout rate for the pairwise features \\
        Activation Function & GELU & GELU & SiLU & Activation function used in the model \\
        \hline
    \end{tabular}}
    \label{tab:unimol_settings}
\end{table}
\newpage
\section{Discussion}
\subsection{Discussion on System Energy Error}
The mean absolute error (MAE) of ~47 kcal/mol for system energy prediction is notably above the threshold of chemical accuracy (typically ~1 kcal/mol). This highlights the inherent challenges of accurately predicting system energies for large molecular systems. To our knowledge, this work is \textit{the first} to evaluate Hamiltonian predictions on system energy, which presents a non-trivial implementation challenge. Previous studies have typically focused on metrics such as cosine similarity of eigenvectors or MAE of occupied energy levels. By directly assessing system energy, our approach provides a more comprehensive and practical evaluation of Hamiltonian accuracy, which is critical for large-scale molecular simulations.

Our model demonstrates a significant improvement over baseline models in predicting system energies for these large systems, indicating enhanced prediction accuracy. The primary goal of this work is to showcase the scalability and applicability of our approach to large molecular systems, where achieving absolute precision is inherently difficult due to their complexity. Despite this, our results represent a meaningful step toward more accurate and efficient predictions in this challenging domain. The advances we have made underscore the potential of our approach to be further refined and applied across a range of complex molecular systems.

Moreover, other critical molecular properties such as the highest occupied molecular orbital (HOMO), the lowest unoccupied molecular orbital (LUMO), and dipole moments are predicted with reasonable accuracy. These properties, often vital in quantum chemistry workflows, demonstrate the practical utility of our model beyond system energy prediction. The ability to predict multiple important molecular properties with competitive error rates reinforces the broader applicability of our model to a variety of quantum chemistry and materials science applications

\subsection{Additional Conformational Energy Analysis}

In quantum chemistry applications, relative energy differences between conformational states are typically more relevant than absolute energies. To better evaluate our method's practical utility, we conducted an additional conformational energy analysis. For each molecule in our test set, we:

\begin{enumerate}
    \item Sampled 100 conformations.
    \item Computed $\Delta E$ between each conformation and a reference structure using both DFT and WANet.
    \item Calculated the MAE of these energy differences.
\end{enumerate}

For sampling, we employed Gaussian perturbation to introduce noise in the atomic positions. Starting from each molecule's initial  conformation from the test sets, we generated 100 perturbed geometries by adding random Gaussian noise ($\sigma=0.1 \text{\r{A}} $) to the atomic coordinates. The potential energy values for these perturbed conformations were computed using the B3LYP/def2-TZVP level of theory. The resulting energy standard deviation between sampled conformations ranged from 10 to around 100 kcal/mol (with average std =49.84kcal/mol), depending on the molecule's flexibility. For baseline models (included in the Table~\ref{tab:mae_gap}), the PhisNet model withWALoss achieved 9.90 kcal/mol MAE, the QHNet model with WALoss achieved 10.98 kcal/mol MAE, and the WANet model without WALoss achieved 48.92 kcal/mol MAE. QHNet without WALoss yielded MAEs of 50.56 kcal/mol. Interestingly, we observed that the model without WALoss produce meaningless prediction, demonstrating the effectiveness of WALoss. It worth noting that while more complex sampling strategy could be used, those methods are more computationally expensive as they require DFT calculations.

\begin{table}[h]
    \centering
    \caption{Relative energy differences between conformational states.}
    \vspace{2mm}
    \resizebox{0.7\textwidth}{!}{
    \begin{tabular}{lc}
    \toprule
    Model &	MAE of Energy Differences (kcal/mol) \\
        \hline
WANet w/WALoss & 	1.12 \\
PhisNet w/WALoss & 9.90\\
QHNet w/WALoss	 & 10.98\\
WANet	 & 48.92\\
QHNet & 	50.56\\
        \hline
    \end{tabular}
    }
\label{tab:mae_gap}
\end{table}

\subsection{Discussion on SCF acceleration ratio} We report an 18\% reduction in SCF cycles, which is a notable improvement, especially considering that our work targets significantly larger and more complex molecular systems compared to previous studies such as QH9 and QHNet, where SCF reductions of 18–35\% were achieved. Given the increased complexity and size of our systems, this 18\% reduction represents a considerable improvement, reflecting the effectiveness of our model in accelerating convergence for more challenging molecular simulations. This result highlights the scalability of our approach in handling large-scale quantum chemical calculations efficiently.

\end{document}